\documentclass[10pt]{article}

\newcommand{\Hcal}{\mathcal{H}}
\newcommand{\Kcal}{\mathcal{K}}
\newcommand{\Ncal}{\mathcal{N}}
\newcommand{\DF}{\mathcal{N}(\mathcal{P})}

\newcommand{\Lcal}{\mathcal{L}}

\newcommand{\Pcal}{\mathcal{P}}

\newcommand{\Scal}{\mathcal{S}}
\newcommand{\Dcal}{\mathcal{D}}
\newcommand{\Ecal}{\mathcal{E}}
\newcommand{\Bcal}{\mathcal{B}}

\newcommand{\Ebb}{\mathbb{E}}

\newcommand{\Lbb}{\mathbb{L}}

\newcommand{\Dent}[2]{D\left(#1\,||\,#2\right)}

\newcommand{\eps}{\varepsilon}

\newcommand{\sigtr}{{\sigma_{\tr}}}

\newcommand{\Tr}{\text{Tr}\,}

\usepackage{fullpage}
\usepackage[bookmarks]{hyperref}
\usepackage{amssymb, slashed}
\usepackage{amsmath} 
\usepackage{amsthm}
\usepackage{graphicx,color,colordvi}
\usepackage{bbm}
\usepackage{cleveref}
\usepackage{stmaryrd}
\usepackage[utf8]{inputenc}
\usepackage{fullpage}
\usepackage[blocks]{authblk}
\usepackage{mathtools}
\usepackage[all]{xy}
\usepackage{setspace}
\usepackage{verbatim}
\usepackage{thm-restate}

\def\Lcal{\mathcal{L}}
\def\Tr{\operatorname{Tr}}
\def\openone{\leavevmode\hbox{\small1\kern-3.8pt\normalsize1}}

\def\II{\mathbb I}

\def\tr{{\operatorname{Tr}}}

\def\Scal{\mathcal{S}}
\def\CC{\mathbb{C}}
\def\RR{\mathbb{R}}

\def\CB{\operatorname{CB}}
\def\NN{\mathbb{N}}
\def\LL{\mathbb{L}}

\def\11{\mathbf{1}}
\def\LL{\mathcal{L}}

\def\Pcal{\mathcal{P}}

\def\Ncal{\mathcal{N}}

\theoremstyle{plain}
\newtheorem{theorem}{Theorem}[section]
\newtheorem{lemma}[theorem]{Lemma}
\newtheorem{proposition}[theorem]{Proposition}
\newtheorem{corollary}[theorem]{Corollary}

\theoremstyle{definition}
\newtheorem{definition}[theorem]{Definition}

\newtheorem{example}[theorem]{Example}

\theoremstyle{remark}
\newtheorem{remark}[theorem]{Remark}

\newcommand{\norm}[1]{\left\| #1 \right\|}

\newcommand{\outerp}[2]{\ket{#1}\!\bra{#2}}

\newcommand{\vertiii}[1]{{\left\vert\kern-0.25ex\left\vert\kern-0.25ex\left\vert #1 
		\right\vert\kern-0.25ex\right\vert\kern-0.25ex\right\vert}}

\newcommand{\sca}[2]{\langle #1 ,#2\rangle}

\def\reff#1{(\ref{#1})}
\def\eps{\varepsilon}

\newcommand{\supp}{\mathop{\rm supp}\nolimits}

\newcommand{\bra}[1]{\langle#1|}
\newcommand{\ket}[1]{|#1\rangle}

\newcommand{\cB}{{\cal B}}
\newcommand{\cD}{{\cal D}}
\newcommand{\cE}{{\cal E}}

\newcommand{\cN}{{\cal N}}

\newcommand{\cH}{{\cal H}}
\newcommand{\cK}{{\cal K}}

\newcommand{\cS}{\mathcal{S}}
\newcommand{\cP}{\mathcal{P}}

\newcommand{\cL}{{\cal L}}

\newcommand{\C}{{\mathbb{C}}}
\newcommand{\R}{{\mathbb{R}}}

\def\e{\mathrm{e}}

\numberwithin{equation}{section}

\DeclareRobustCommand\openone{\leavevmode\hbox{\small1\normalsize\kern-.33em1}}

\newcommand{\id}{{\rm{id}}}
\newcommand{\be}{\begin{equation}}
	\newcommand{\ee}{\end{equation}}
\newcommand{\bea}{\begin{eqnarray}}
	\newcommand{\eea}{\end{eqnarray}}
\newcommand{\beas}{\begin{eqnarray*}}
	\newcommand{\eeas}{\end{eqnarray*}}

\title{Hypercontractivity and logarithmic Sobolev Inequality\\
	 for non-primitive quantum Markov semigroups\\
	  and estimation of decoherence rates.}
\author[1]{Ivan Bardet}
\affil[1]{\small Institut des Hautes \'Etudes Scientifiques, Université Paris-Saclay, 35 Route de Chartres, 91440 Bures-sur-Yvette, France}
\author[2]{Cambyse Rouz\'{e}}
\affil[2]{\small Statistical Laboratory, Centre for Mathematical Sciences, University of Cambridge, Cambridge~CB30WB, UK}
\begin{document}

\maketitle

\begin{abstract}
We generalize the concepts of weak quantum logarithmic Sobolev inequality (LSI) and weak hypercontractivity (HC), introduced in the quantum setting by Olkiewicz and Zegarlinski, to the case of non-primitive quantum Markov semigroups (QMS). The originality of this work resides in that this new notion of hypercontractivity is given in terms of the so-called \emph{amalgamated $\mathbb{L}_p$ norms} introduced recently by Junge and Parcet in the context of operator spaces theory. We make three main contributions. The first one is a version of Gross' integration lemma: we prove that (weak) HC implies (weak) LSI. Surprisingly, the converse implication differs from the primitive case as we show that LSI implies HC but with a weak constant equal to the cardinal of the center of the \emph{decoherence-free} algebra. Building on the first implication, our second contribution is the fact that strong LSI and therefore strong HC do not hold for non-trivially primitive QMS. This implies that the amalgamated $\Lbb_p$ norms are not uniformly convex for $1\leq p \leq 2$. As a third contribution, we derive universal bounds on the (weak) logarithmic Sobolev constants for a QMS on a finite dimensional Hilbert space, using a similar method as Diaconis and Saloff-Coste in the case of classical primitive Markov chains, and Temme, Pastawski and Kastoryano in the case of primitive QMS. This leads to new bounds on the decoherence rates of decohering QMS. Additionally, we apply our results to the study of the tensorization of HC in non-commutative spaces in terms of the completely bounded norms (CB norms) recently introduced by Beigi and King for unital and trace preserving QMS. We generalize their results to the case of a general primitive QMS and provide estimates on the (weak) constants.
\end{abstract}

\section{Introduction}
The study of open quantum systems originated from the observation that a quantum system is never perfectly isolated and therefore undergoes dissipative effects induced by the environement. Such features, known as environment-induced decoherence~\cite{Z81,Z82}, impose strict practical restrictions on the development of quantum information processing~\cite{[NC02]}, since it results in the dynamical loss of the quantum correlations of the system that these theories rely on~\cite{BO03,Giulini2003,SchlosshauerMaximilianA.}. Therefore, esimating the typical time of decoherence appears to be crucial if one is interested in the construction of quantum computers and quantum memories that keep quantum correlations over a long period of time~\cite{[CLBT12],Muller-Hermes2015,[T14]}. As regards to foundations, decoherence is also believed by some to be a partial solution to the measurement problem~\cite{Giulini2003,Z81}. The study of the speed of decoherence hence appears to be of crucial importance for both foundational and practical reasons. The goal of this article is to develop tools coming from functional analysis in order to tackle this task.

The first attempt to mathematically formalize the concept of decoherence is due to \cite{BO03}. Under the Markovian approximation, the evolution of an open quantum system can be modeled by a quantum Markov semigroup (QMS) acting on the algebra of observables. The environment is then said to induce decoherence on the system when any initial observable converges in a $*$-subalgebra of \textit{effective observables} (also known as \textit{decoherence-free} (DF)-\textit{algebra}) on which the QMS acts unitarily. When this subalgebra is not trivially reduced to the multiple of the unit, the QMS necessarily admits more than one invariant state (actually an infinite number) and therefore is not primitive. 

In the present literature, however, the study of the speed of decoherence almost exclusively focuses on QMS in the \textit{primitive} case, that is a QMS possessing a unique full-support invariant state towards which it converges. In this case, characteristic times of decoherence are usually referred to as \textit{mixing times}. For typical systems such as finite dimensional many-body systems, one can actually hope to get an exponentially fast convergence toward the invariant state, a property called \emph{rapid-mixing}. Rapid mixing has found many applications in the recent theory of dissipative engineering, where the run time of various quantum algorithms depends on the mixing time of a QMS~\cite{[KB14],[TOVPV11],van2017quantum,[VWC09]}. In particular, it was shown to imply robustness of dissipative state preparation against perturbations~\cite{[CLM15],kastoryano2013rapid,[SW13]}, area law~\cite{brandao2015area} and exponential decay of correlations~\cite{kastoryano2013rapid}. Generalizing these concepts to the non-primitive case, where the evolution converges to a non-trivial algebra of effective observables, could also potentially lead to more applications in quantum error correction~\cite{bacon2000universal,lidar1998decoherence,ticozzi2008quantum}.

In the classical theory of continuous-time Markov chains, functional analytic tools have been extensively developed and studied in order to prove rapid-mixing and obtain estimates of the mixing time. The most well-known ones are the spectral gap method, or Poincar\'e inequality (PI)~\cite{diaconis1991geometric,lawler1988bounds}, and the (modified) logarithmic Sobolev inequality (LSI) and its equivalent notion of hypercontractivity (HC)~\cite{bobkov2006modified,Gross1975}. A systematic and comprehensive study of these latter concepts for Markov chains on finite set can be found in \cite{Diaconis1996a}. Largely inspired by this paper, our goal in this article is to develop the theory of LSI and HC for non-primitive QMS and its use in proving rapid decoherence. Note that, due to the non-commutativity of quantum systems, typical quantum features arise in this situation that are absent from the classical theory (see~\cite{BarEID17} for a discussion of this point). In the following informal presentation, we shall highlight the key differences between the theory we develop and the (quantum) primitive case.
\paragraph{Informal presentation:}
The theory of hypercontractivity for primitive QMS was fully formalized in the article of Zegarlinski and Olkiewicz~\cite{OZ99}, using Kosaki's theory of non-commutative interpolating weighted $\mathbb{L}_p$ spaces~\cite{K84,Majewski1996}, where the weights here are given in terms of the unique invariant state of the evolution. This study was further pursued by different authors~\cite{TPK,CM15} and applied to the problem of estimating mixing times in \cite{[KT13]}. Here, we briefly describe the main ideas of this article.

Consider a quantum state modeled by an initial density matrix $\rho$ and denote by $(\rho_t)_{t\geq0}$ the solution of a quantum master equation
\[\frac{d}{dt}\,\rho_t=\Lcal(\rho_t)\,,\qquad \rho_0=\rho\,,\]
where $\Lcal$ is the so-called Lindbladian (precise definitions will be given in the next section). When the evolution is primitive, there exists a unique density matrix $\sigma$ such that for any initial state:
\begin{equation}\label{eq_intro_limit}
\rho_t\underset{t\to+\infty}{\longrightarrow} \sigma\,.
\end{equation}
The mixing time is then defined as the first time $\rho_t$ comes to a distance $\eps>0$ of $\sigma$ in trace distance:
\[\tau(\eps)=\inf\,\{t\geq 0\,;\,\norm{\rho_t-\sigma}_1\leq\eps\quad\forall \rho\}\,.\]
Here $\eps$ is chosen arbitrarily but the choice of the $1$-norm $\norm{\cdot}_{1}=\Tr|\cdot|$ is primordial since it has the appropriate operational interpretation as a measure of indistinguishability between two states for an external observer allowed to perform any measurement on the system \cite{[FWG99]}. The first approach to obtain an upper bound on this mixing time is through the spectral gap method, which goes as follows:
\begin{align}\label{mixspec}
\|\rho_t-\sigma\|_1& \le \norm{X_t-\Tr[\sigma\,X_t]\,\mathbb{I}}_{2,\,\sigma} \nonumber\\
& \leq \norm{X}_{2,\sigma}\,\norm{Y\mapsto \hat\Pcal_t(Y)-\Tr[\sigma\,Y]\,\mathbb{I}}_{2\to2,\,\sigma} \nonumber\\
& \leq \norm{X}_{2,\sigma}\,e^{-\lambda(\Lcal)\,t}
\end{align}
where $X_t=\sigma^{-\frac12}\rho_t\,\sigma^{-\frac12}$ (and $X_0=X$) can be thought of as the relative density of $\rho_t$ with respect to $\sigma$, where $Y\mapsto\hat{\Pcal}_t(Y)$ is the quantum Markov semigroup solution of the master equation for the relative density and where $\norm{\cdot}_{2,\sigma}$ is the weighted $2$-norm mentioned above. The first inequality was proved by Ruskai~\cite{Rus94} and can be seen as the quantum generalization of the inequality between the total-variation distance and the $2$-norm in a probability space. The second inequality is just the definition of the norm of an operator from one Banach space to another. The third inequality is given by the spectral gap method: $\lambda(\Lcal)$ is the spectral gap of the Lindbladian, that is the difference between the eigenvalue $0$ and the second largest eigenvalue of $\frac{\Lcal+\hat{\Lcal}}{2}$. It is well-known that it is also given by the optimal constant appearing in the so-called \emph{Poincar\'e} inequality.

It is important to notice that $\norm{X}_{2,\sigma}\leq\sqrt{1/\sigma_{\min}}$, where $\sigma_{\min}$ is the smallest eigenvalue of $\sigma$. In most applications $1/\sigma_{\min}$ scales linearly with the dimension of the system. It is known that this method does not usually lead to the best estimate of the mixing time because mixing can be much faster at short times. One idea that greatly improves this estimate is to use hypercontractivity of the QMS instead, which leads to the following chain of inequalities:
	\begin{align*}
\|\rho_{t+s}-\sigma\|_1&\le\|X_{t+s}-\tr[\sigma X]\,\mathbb{I}\|_{2,\sigma}\\
	&\le \|X\|_{1,\sigma}\,\norm{\hat{\Pcal}_s}_{1\to2,\,\sigma}\,\|Y\mapsto \hat{\cP}_{t}(Y)-\tr[\sigma Y]\,\mathbb{I}\|_{2\to 2,\sigma}\\
	&\le\,\norm{\Pcal_s}_{2\to\infty,\,\sigma}\,\e^{-\lambda(\cL)t}\,,
\end{align*}
as one has $\|X\|_{1,\sigma}=1$. 
Since $\Pcal_s$ is contractive for any $p$-norm with $p\geq1$, and in view of the limit in \eqref{eq_intro_limit}, one can hope to find a time $s>0$ such that $\norm{\Pcal_s}_{2\to\infty,\,\sigma}\leq 2$ ($2$ here is of course arbitrary). However, even for classical Markov chains, $\norm{\Pcal_s}_{2\to\infty,\,\sigma}$ is in practice difficult to estimate. The concept of hypercontractivity hence provides a tool to interpolate between this norm and the $2\to 2$ norm given by the spectral gap method, where one uses instead the $2\to p$ norm for $p>2$. In this case, the factor $\norm{X}_{p,\sigma}\leq(1/\sigma_{\min})^{\frac1p}$ appears, which indeed interpolates between the two previous methods. The great discovery of Gross was that finding a time $t\geq0$ for which $\Pcal_t$ becomes a contractive operator from $\mathbb{L}_2$ to $\mathbb{L}_p$ is an equivalent problem to the one of optimizing the so-called logarithmic Sobolev inequality. Exploiting this equivalence, Diaconis and Saloff-Coste were able to find optimal or near to optimal upper bounds of the mixing time \cite{Diaconis1996a,DSC93,DSC93group}.

In practice, the Poincar\'e inequality (or spectral gap method) can lead to an upper bound of the mixing time of order $\ln 1/\sigma_{\min}$, whereas hypercontractivity leads to an upper bound of order $\ln\ln1/\sigma_{\min}$. Thus hypercontractivity improves on the Poincar\'e inequality by a logarithmic factor. Of course, the hypercontractive property depends highly on the choice of the interpolating family of $\mathbb{L}_p$ norms. In particular, a QMS which is hypercontractive for Kosaki's $\mathbb{L}_p$ norms will be primitive. The main contribution of the present work is to study hypercontracitivity with respect to a generalisation of Kosaki's norms, called the amalgamated norms and defined by Junge and Parcet in \cite{JP10}. Using these norms, we will be able to reproduce the above steps for non-primitive QMS.

One other motivation for considering non-primitive QMS is that they naturally appear when considering the tensorization of the logarithmic Sobolev inequality for primitive QMS. Indeed, one central property of the classical LSI is that the LSI constant, i.e. the best constant that satisfies the inequality, is stable when considering several non-interacting systems: the global LSI constant is equal to the smallest constant of the individual systems. For classical systems, this property follows directly from the multiplicativity of the $\mathbb{L}_p\to \mathbb{L}_q$ norms and the equivalence between HC and LSI. However, this property is strongly believed to be false for quantum channels with respect to the usual quantum $\mathbb{L}_p\to \mathbb{L}_q$ norms. Several methods have been proposed in order to lower bound the global LSI constant~\cite{TPK,MSFW}. A promising approach was to consider HC with respect to the completely bounded (CB) $\mathbb{L}_p$ norms for which the multiplicativity is restored \cite{[BK16]}. When dealing with such norms, one has to consider a ``regularisation'' of the primitive QMS, that is, one has to embed the QMS into a bigger one for which primitivity does not hold any longer.
\subsection*{Our contribution}
In this paper, we extend log-Sobolev inequalities and the related notion of hypercontractivity to the case of non-primitive QMS $(\cP_t)_{t\ge 0}$, based on the properties of the amalgamated $\mathbb{L}_p$ norms. Among other properties, we shall provide some elementary proofs of the following ones: these norms satisfy H\"{o}lder's inequality, are dual to each other, and reduce to the usual weighted $\mathbb{L}_p$ norms when $(\cP_t)_{t\ge 0}$ is primitive, that is, when the algebra of effective observables is trivial. 

Following ideas from~\cite{[BK16]}, we derive a formula for the differential of the amalgamated $\mathbb{L}_p$ norms (see \Cref{diffnorm}), with respect to the index $p$. This leads to the definition of the \textit{weak decoherence-free logarithmic Sobolev inequality} (DF-wLSI) and the \textit{weak decoherence-free hypercontractivity} (DF-wHC), and allows us to extend Gross' integration lemma to this setting (see \Cref{theogross}). A first difference compared to the primitive case is that LSI implies HC but with a larger weak constant which depends on the structure of the DF-algebra.

In the primitive case, the uniform convexity of the $\mathbb{L}_p$ norms was used in~\cite{OZ99} to show that wLSI together with PI imply the so-called \textit{strong logarithmic Sobolev inequality} (sLSI). We show that a similar analysis can be performed in our extended framework, in order to derive universal upper bounds on the log-Sobolev constants (see \Cref{logsob} and \Cref{coro_univconstants1}). We also prove that, except in the primitive case, the strong LSI does not hold and therefore neither does the related notion of strong hypercontractivity. This implies that the uniform convexity no longer holds for the amalgamated $\mathbb{L}_p$-norms.

We then show how the techniques introduced can be used to derive decoherence rates for non-primitive QMS, based on the method explained above. Finally, our framework also allows for the definition of the \textit{weak completely bounded hypercontractivity} (CB-wHC) and \textit{log-Sobolev inequality} (CB-wLSI) for non-unital primitive QMS, which extends the framework of~\cite{[BK16]}. In particular, we prove Gross' integration lemma (see \Cref{grossCB}), and derive universal bounds on the weak CB log-Sobolev constants (see \Cref{CBHCLSI} and \Cref{coro_univconstantsCB1}). 

\subsection*{Layout of the paper}
In \Cref{sec2}, we provide the notations and basic tools that will be used throughout this paper, namely quantum Markov semigroups and environment-induced decoherence, and state our main results. In \Cref{norms}, we introduce the amalgamated $\mathbb{L}_p$ norms and study their properties. The notions of decoherence-free log-Sobolev inequality and hypercontractivity are studied in \Cref{logSob}, where we prove Gross' integration Lemma as well as a universal upper bound on the constants. In \cref{sec5bis} we prove that the strong LSI fails for non-trivially primitive QMS. Some applications of our framework to the derivation of decoherence rates are provided in \Cref{sec6}. We highlight our result in a special class of decohering QMS arising from Lie-group representation theory in \Cref{sect_example}. We conclude with the analysis of the CB case in \Cref{CBlogsob}.

\section{Preliminaries and statement of the main results}\label{sec2}
This part is organised as follows: in \Cref{sect22} we introduce our notations and recall the definitions of quantum Markov semigroups, their decoherence-free algebra and the notion of environment-induced decoherence. \Cref{sect23} is devoted to the exposition of the weighted $\mathbb{L}_p$ norms and the $\mathbb{L}_p$ Dirichlet forms associated to a quantum Markov semigroup. The main results of this article are presented in \Cref{sect24}, namely the equivalence between hypercontractivity and logarithmic Sobolev inequality in the context of amalgamated $\mathbb{L}_p$ spaces, and the existence of universal constants. In \Cref{sect25} we apply our framework to the estimation of decoherence rates. Finally, the study of hypercontractivity for the CB-norms is presented in \Cref{sect26}.

\subsection{Quantum Markov semigroups and environment-induced decoherence}\label{sect22}
Let $(\cH,\langle .|.\rangle)$ be a finite dimensional Hilbert space of dimension $d_\cH$. We denote by $\cB(\cH)$ the Banach space of bounded operators on $\cH$, by $\cB_{\text{sa}}(\cH)$ the subspace of self-adjoint operators on $\cH$, i.e. $\cB_{\text{sa}}(\cH)=\left\{X=\cB(\cH);\ X=X^*\right\}$, and by $\cB_{\text{sa}}^+(\cH)$ the cone of positive semidefinite operators on $\cH$, where the adjoint of an operator $Y$ is written as $Y^*$. The identity operator on $\cH$ is denoted by $\mathbb{I}_\cH$, dropping the index $\cH$ when it is unnecessary. In the case when $\cH\equiv \CC^k$, we will also use the notation $\mathbb{I}_k$ for $\mathbb{I}_{\CC^k}$. Similarly, we will denote by $\id_{\cH}$, or simply $\id$, resp. $\id_k$, the identity superoperator on $\cB(\cH)$, resp. $\cB(\CC^k)$. We denote by $\mathcal{D}(\cH)$ the set of positive semidefinite, trace one operators on $\cH$, also called \emph{density operators}, and by $\cD_+(\cH)$ the subset of full-rank density operators. In the following, we will often identify a density matrix $\rho\in\mathcal{D}(\cH)$ and the \emph{state} it defines, that is the positive linear functional $\cB(\cH)\ni X\mapsto\tr(\rho \,X)$. 

The basic model for the evolution of an open system in the Markovian regime is given by a quantum Markov semigroup (or QMS) $(\cP_t)_{t\ge0}$ acting on $\cB(\cH)$. Such a semigroup is characterised by its generator, called the Lindbladian $\LL$, which is defined on $\Bcal(\Hcal)$ by $\Lcal(X)={\lim}_{t\to 0}\,\frac{1}{t}\,(\Pcal_t(X)-X)$ for all $X\in\Bcal(\Hcal)$. Recall that by the GKLS Theorem \cite{Lind,[GKS76]}, $\cL$ takes the form:
\begin{equation}\label{eqlindblad}
\cL(X)=i[H,X]+\frac{1}{2}\sum_{k=1}^l{\left[2\,L_k^*XL_k-\left(L_k^*L_k\,X+X\,L_k^*L_k\right)\right]}\,,\quad ~~~~~~~~~~\text{ for all }X\in\cB(\cH)\,,
\end{equation}
where $H\in\cB_{\text{sa}}(\cH)$, where the sum runs over a finite number of \textit{Lindblad operators} $L_k\in\Bcal(\Hcal)$, and where $[\cdot,\cdot]$ denotes the commutator defined as $[X,Y]:=XY-XY$, $\forall X,Y\in\cB(\cH)$.\\
We denote by $(\cP_{*t})_{t\geq0}$ the predual of the QMS $(\cP_t)_{t\ge 0}$ for the Hilbert-Schmidt inner product $\langle A,B\rangle:=\tr( A^* B)$, that is the unique trace-preserving QMS such that for all $X,Y\in\cB(\cH)$ and all $t\ge 0$,
\[\tr[\cP_t(X)\,Y]=\tr[X\,\cP_{*t}(Y)]\,.\]
Its generator $\Lcal_*$ is the predual of $\Lcal$ and takes the form:
\begin{equation*}
\cL_*(\rho)=-i[H,\rho]+\frac{1}{2}\sum_{k=1}^l{\left[2L_k\rho L_k^*-\left(L_k^*L_k\,\rho +\rho\, L_k^*L_k\right)\right]}\,,\quad ~~~~~~~~~~ \text{for all }\rho\in\cB(\cH)\,.
\end{equation*}
We shall always assume that $(\cP_t)_{t\ge 0}$ admits an invariant state, that is a density operator $\sigma$ in $\cD(\cH)$ such that for all time $t\geq0$ and all $X\in\cB(\cH)$, $\tr(\sigma \Pcal_t(X))=\tr(\sigma X)$. Equivalently, one has $\Pcal_{*t}(\sigma)=\sigma$ for all $t\geq0$. Furthermore, we shall also assume that $\sigma$ is \emph{faithful}, that is, $\sigma\in\cD_+(\cH)$. Under this condition, it was proved for instance in \cite{CSU4} that the maximal algebra on which $(\cP_t)_{t\ge 0}$ acts as a $*$-automorphism is the \emph{decoherence-free subalgebra} of $(\cP_t)_{t\ge 0}$, defined by
\begin{equation}\label{eq_def_DFalgebra}
\Ncal(\cP)=\left\{X\in\cB(\cH),\ \Pcal_t(X^*X)=\Pcal_t(X)^*\,\Pcal_t(X)\text{ and }\Pcal_t(XX^*)=\Pcal_t(X)\,\Pcal_t(X)^*\ \forall t\geq0\right\}\,.
\end{equation}
Consequently, there exists a one-parameter group of unitary operators $(U_t)_{t\in\R}$ on $\cH$ such that for any $X\in\DF$ and all $t\ge 0$:
\begin{equation}\label{eq_unitary_evolution}
\Pcal_t(X)=U^*_t\,X\,U_t\,,
\end{equation}
and $\DF$ is the largest subalgebra of $\cB(\cH)$ such that this property holds. In this case, the following result is known (we state it in a form more convenient to our analysis). We recall that a conditional expectation between two subalgebras $\mathcal{M}$ and $\mathcal{N}$ of $\cB(\cH)$ is a completely positive unital contraction ${E}_\cN:\mathcal{M}\to \mathcal{N}$ such that for any $A,B\in\mathcal{N}$ and $X\in\mathcal{M}$ \cite{Tom59},
	\begin{align}\label{condexp}
	{E}_\cN[A\,X\,B]=A\,{E}_\cN[X]\,B\,.
	\end{align}
We also denote be $E_{\Ncal*}$ the predual of this conditional expectation, defined as the unique operator on $\Bcal(\Hcal)$ such that for all $X,Y\in\Bcal(\Hcal)$,
\[\tr[E_{\Ncal*}(X)\,Y]=\tr[X\,E_{\Ncal}[Y]]\,.\]
\begin{theorem}[Proposition 8 of \cite{[CSU13]}, Theorem 19 of \cite{CSU4}]\label{theo_deco}
Assume that $(\cP_t)_{t\ge 0}$ has a faithful invariant state $\sigma$. Then there exists a unique conditional expectation $E_\Ncal$ from $\cB(\cH)$ to $\Ncal(\cP)$ compatible with $\sigma$, that is for which $\sigma=E_{\cN^*}(\sigma)$, and such that for all observables $X\in\Bcal(\Hcal)$,
\begin{equation}\label{eq_theo_deco_heis}
\underset{t\to+\infty}{\lim}\,\Pcal_t\left(X-E_\Ncal[X]\right)=0\,.
\end{equation}
Equivalently, the predual $E_{\Ncal*}$ of $E_\Ncal$ is such that for all states $\rho\in\Dcal(\Hcal)$,
\begin{equation}\label{eq_theo_deco_schro}
\underset{t\to+\infty}{\lim}\,\Pcal_{*t}\left(\rho-E_{\Ncal*}(\rho)\right)=0\,.
\end{equation}
\end{theorem}

Notice that consequently, since $E_\Ncal$ is a projection, the following decomposition of $\cB(\cH)$ takes place:
\begin{equation*}
\cB(\cH)=\DF\oplus\,\operatorname{Ker}\,E_\Ncal\,,\qquad\text{ where }\qquad \underset{t\to+\infty}{\lim}\,\Pcal_t(X)=0\quad \forall X\in\operatorname{Ker}\,E_\Ncal\,.
\end{equation*}
This is the so-called notion of \textit{environment-induced decoherence} (EID). In what follows, we simply call a QMS possessing a faithful invariant state a \textit{decohering QMS}. In the case of a primitive QMS, with associated unique invariant state $\sigma$, $E_\cN[X]=\tr(\sigma X)\mathbb{I}$. When the QMS is not primitive, there necessarily exists an infinity of invariant states and it will be relevant to pick one as a reference state. We define:
\begin{equation}\label{eq_reference_state}
	\sigma_\tr:=E_{\cN*}\left( \frac{\mathbb{I}_\cH}{d_\cH}\right)\,.
\end{equation}
This choice appeared to be particularly relevant when defining analogues of Poincar\'{e}'s- and the modified log-Sobolev- inequalities in \cite{BarEID17}. This comes from the fact that $\sigma_\tr$ is \emph{tracial} on $\DF$, that is, for all $X\in\DF$ and all $Y\in\cB(\cH)$,
\[\tr (\sigma_\tr\, XY)=\tr (\sigma_\tr\, YX)\,.\]

A basic result from the theory of $*$-algebras on finite dimensional Hilbert spaces states that $\DF$ can always be decomposed into a direct sum of subparts where it restricts to a factor \cite{[KR15]}. More precisely, up to a unitary transformation, the Hilbert space $\cH$ admits the following decomposition
\begin{equation}\label{eqtheostructlind1}
\cH=\bigoplus_{i\in I}{\cH_i\otimes\cK_i}\,,
\end{equation}
such that $\Ncal(\cP)$ is unitarily isomorphic to the algebra
\begin{equation}\label{eqtheostructlind2}
\DF=\bigoplus_{i\in I}{\cB(\cH_i)\otimes \mathbb{I}_{\cK_i}}\,.
\end{equation}
Finally, as proved in \cite{[DFSU14]}, there exists a family of density operators $\{\tau_i:~i\in I\}$ such that for all $\rho\in\cD_+(\cH)$ and any $X\in\cB(\cH)$,
\begin{align}\label{statedecomp}
	\rho_\cN\equiv E_{\cN*}(\rho)=\sum_{i\in I}\tr_{\cK_i}(P_i \rho P_i)\otimes \tau_i\,~~~~~~~~~~~E_\cN[X]=\sum_{i\in I}\tr_{\mathcal{K}_i}((\mathbb{I}_{\cH_i}\otimes \tau_i) P_i X P_i)\otimes \mathbb{I}_{\cK_i} \,,
	\end{align}
where for each $i$, $P_i$ denotes the projection onto $\cH_i\otimes \cK_i$, and $\Tr_{\Kcal_i}$ is the partial trace with respect to $\Kcal_i$, defined as the unique operator from $\Bcal(\Hcal_i\otimes\Kcal_i)$ to $\Bcal(\Hcal_i)$ such that for all operators $X\in\Bcal(\Hcal_i\otimes\Kcal_i)$,
\[\Tr\big[Y\,\Tr_{\Kcal_i}[X]\big]=\Tr\big[\left(Y\otimes \mathbb{I}_{\Kcal_i}\right)\,X\big]\qquad \text{for all }Y\in\Bcal(\Hcal_i)\,. \]
In particular,
\begin{align}\label{sigmatrace}
	\sigma_\tr=\frac{1}{d_\cH}\sum_{i\in I} d_{\cK_i}\mathbb{I}_{\cH_i}\otimes \tau_i\,.
\end{align}	

\subsection{Non-commutative weighted $\mathbb{L}_p$ spaces and $\mathbb{L}_p$ Dirichlet forms}\label{sect23}

For $p\ge 1$ and an operator $X\in\cB(\cH)$, we denote by $\| X\|_{p}:=(\tr|X|^p)^{\frac{1}{p}}$ the Schatten $p$-norm of $X$, embedding $\cB(\cH)$ into a normed vector space $\Scal_p(\cH)$. In the study of non-commutative functional inequalities, a natural family of $\mathbb{L}_p$ spaces is given by the following weighted versions of the Schatten norms \cite{OZ99}: Let $(\Pcal_t)_{t\geq0}$ be a QMS with a faithful invariant state and denote by $\sigtr$ the faithful density operator defined in \Cref{eq_reference_state}. The space $\cB(\cH)$ is naturally endowed with a complex Hilbert space structure with respect to $\sigtr$, with inner product defined for all $X,Y\in\cB(\cH)$ by:
\begin{align}\label{inner}
	\sca{X}{Y}_\sigtr:=\Tr\left[\sigtr^{\frac{1}{2}}X^*\sigtr^{\frac{1}{2}}Y\right]\,.
\end{align}
One can show that the conditional expectation $E_\Ncal$ is actually the orthogonal projection on $\DF$ for this inner product (cf. \cite{BarEID17}): for all $X,Y\in\cB(\cH)$:
\begin{align}\label{En}
\langle X,E_\cN[Y]\rangle_{\sigma_\tr}=\langle E_\cN[X],Y\rangle_{\sigma_\tr}=\langle E_{\cN}[X],E_\cN[Y]\rangle_{\sigma_\tr}\,,
\end{align}
 which is one of the motivations behind the choice of $\sigtr$ as our reference state. This implies the interesting relation:
\begin{equation}\label{eq_com_cond_expt}
\sigma_{\Tr}^{\frac12}\,E_\Ncal[X]\,\sigma_{\Tr}^{\frac12}=E_{\Ncal*}(\sigma_{\Tr}^{\frac12}\,X\,\sigma_{\Tr}^{\frac12})\,.
\end{equation}
We now define the weighted norms $\norm{\cdot}_{p,\sigtr}$ on $\cB(\cH)$ for all $p\geq1$ as follows
\begin{equation*}
	\norm{X}_{p,\sigtr}:=\Tr\left[\left|\sigtr^{\frac{1}{2p}}X\sigtr^{\frac{1}{2p}}\right|^p\right]^{\frac{1}{p}}\,.
\end{equation*}
We denote the space $\cB(\cH)$ endowed with this norm by {$\mathbb{L}_{p}(\cH,\sigtr)$}, or $\mathbb{L}_p(\sigtr)$ for short, when it is clear what the underlying Hilbert space $\cH$ is. Among other properties, these spaces are in natural duality with respect to the inner product $\sca{\cdot}{\cdot}_\sigtr$ (we refer the reader to \cite{[KT13],OZ99} for more details). It will also be useful to denote by $\cS^+_{\mathbb{L}_1(\sigma_\tr)}$ the set of positive definite operators on the sphere of radius $1$ in $\mathbb{L}_1(\sigma_\tr)$. The $\mathbb{L}_p(\sigtr)$ norms are connected to the usual Schatten norms as follows: Define the map:
\begin{equation}\label{eq_def_isometry}
	\begin{array}{cccc}
		\Gamma_\sigtr\,: & X\in\cB(\cH) & \mapsto & \sigma_{\tr}^{\frac{1}{2}}\,X\,\sigma_{\tr}^{\frac{1}{2}}\,,~~~~~\text{ so that }~~~~~ \Gamma_{\sigma_\tr}^{\frac{1}{p}}(X)=\sigma_{\tr}^{\frac{1}{2p}}X\sigma_{\tr}^{\frac{1}{2p}}\,.
	\end{array}
\end{equation}
Then one has $\norm{X}_{p,\sigtr}=\|\Gamma_\sigtr^{\frac{1}{p}}(X)\|_p$. Thus, each of the maps $\Gamma_\sigtr^{\frac{1}{p}}$ defines an isometry between the weighted ${\mathbb{L}}_p(\sigma_{\Tr})$ spaces and the Schatten spaces $\Scal_p(\cH)$. There is also a natural map $I_{q,p}$ between $\mathbb{L}_p(\sigtr)$ and $\mathbb{L}_q(\sigtr)$ for $p,q\geq1$, defined for all $X\in\Bcal(\Hcal)$ by:
\begin{align}\label{Ipq}
I_{q,p}(X):=   \Gamma_{\sigma_\tr}^{-\frac{1}{q}}(|\Gamma_{\sigma_\tr}^{\frac{1}{p}}(X)|^{\frac{p}{q}}) =\sigtr^{-\frac1{2q}}\left|\sigtr^{\frac1{2p}}\,X\,\sigtr^{\frac1{2p}}\right|^\frac{p}{q}\sigtr^{-\frac1{2q}}\,,
\end{align}
so that $\|I_{q,p}(X)\|_{q,\sigma_\tr}^q=\|X\|_{p,\sigma_\tr}^p$. Another quantity that is going to play an important role is the \textit{$\mathbb{L}_p$ Dirichlet form}: for $p\ge1$ of H\"{o}lder conjugate $q$ (i.e.~such that $p^{-1}+q^{-1}=1$), and any $X\in \cB_{sa}(\cH)$,
\begin{align}\label{eq27}
	\cE_{p,\Lcal}(X):=-\frac{p}{2(p-1)}\langle I_{q,p}(X),\LL(X)\rangle_{\sigma_\tr}\,.
\end{align}
In the non-primitive case, the choice of $\sigma_\tr$ in the definition of the Dirichlet form is primordial. In the case $p=2$, we recognise the $\mathbb{L}_2$ Dirichlet form:
\begin{align*}
	\cE_{2,\Lcal}(X)=-\langle X,\LL(X)\rangle_{\sigma_\tr}\,.	
\end{align*}
One can also define the $\mathbb{L}_1$ Dirichlet form as a limit when $p\to 1$ of \reff{eq27}:
\begin{align*}
	\cE_{1,\Lcal}(X)=-\frac{1}{2}\tr (\Gamma_{\sigma_\tr}(\LL(X))(\ln\Gamma_{\sigma_\tr}(X)-\ln(\sigma_\tr)))\,.
\end{align*}
Finally, we say that the QMS $(\cP_t)_{t\ge 0}$ is \textit{reversible} (or satisfies the \textit{detailed balance property}) with respect to $\sigma_\tr$, if $\Lcal$ is self-adjoint with respect to $\sca{\cdot}{\cdot}_{\sigma_\tr}$ (or equivalently if $\cP_t$ is, for all $t\geq0$). That is, for any $X,Y\in\cB(\cH)$:
\begin{align}\label{DBC}
	\langle X,\LL(Y)\rangle_{\sigma_\tr}=	\langle \LL(X),Y\rangle_{\sigma_\tr}\,~~~ (\text{or equivalently } \langle X,\cP_t(Y)\rangle_{\sigma_\tr}=\langle \cP_t(X),Y\rangle_{\sigma_\tr}~\forall t\ge 0)\,.
\end{align}	
We insist once more on the fact that we defined reversibility with respect to the reference state $\sigma_\tr$ and that this choice is primordial in our analysis. In what follows, we will simply say that $(\cP_t)_{t\ge 0}$ is reversible, without mentioning the state.

Other definitions of the quantum detailed balance condition appear in the literature, depending on the choice of the inner product. One particularly relevant for us is with respect to the $1,\sigma_\tr$-inner product given by
\begin{align}\label{DBC1}
	\langle X,Y\rangle_{1,\sigma_\tr}:=\tr[\sigma_{\tr}\,X^*Y],~~~~~~~~~~X,\,Y\in\cB(\cH)\,.
\end{align}
We say that $(\cP_t)_{t\ge 0}$ satisfies the {$\sigma_\tr$-DBC} if $\Lcal$ is self-adjoint with respect to this inner product (or equivalently if $\Pcal_t$ is, for all $t\geq0$). As proved for instance in \cite{[CM16]}, this form of reversibility is stronger than (that is, implies) the one defined by \Cref{DBC}. In particular, it implies that the QMS commutes with the modular operator of $\sigma_\tr$:
\begin{align*}
	\Lcal\circ\Delta_{\sigma_\tr}=\Delta_{\sigma_\tr}\circ	\Lcal\,,
\end{align*}	
where $\Delta_{\sigma_\tr}(.):=\sigma_\tr\,(.)\,\sigma_\tr^{-1}$. A typical example of a QMS that satisfies the $\sigma_\tr$-DBC is the $\Ncal$-decoherent QMS defined as follows. Let $\Ncal$ be a $*$-subalgebra of $\Bcal(\Hcal)$ and let $E_\Ncal$ be any conditional expectation on it. Then the $\Ncal$-decoherent QMS is the one with Lindbladian defined by:
\[\Lcal_\Ncal(X)=E_\Ncal[X]-X\,,\qquad\forall X\in\Bcal(\Hcal)\,.\]
In particular, any conditional expectation commutes with the modular operator of $\sigma_\tr$, a well-known fact in operator algebra theory:
\begin{align}\label{commut}
	E_\cN\circ\Delta_{\sigma_\tr}=\Delta_{\sigma_\tr}\circ	E_\cN\,.
\end{align}	

\subsection{$\operatorname{DF}$-hypercontractivity and the log-Sobolev inequality}\label{sect24}
The main goal of this paper is to introduce a notion a hypercontractivity which is relevant to the study of decoherence rates. Indeed, for finite dimensional Hilbert spaces, hypercontractivity of the QMS with respect to the $\mathbb{L}_p(\sigma_\tr)$ norms is equivalent to the primitivity of the QMS. In order to deal with non-primitive QMS, a possible choice of norms are the so-called \textit{amalgamated norms} introduced in \cite{JP10}. These norms are defined as follows: for $1\leq q\le p\le+ \infty$ and $\frac{1}{r}=\frac{1}{q}-\frac{1}{p}$, define
\begin{align}
	&\norm{X}_{(q,p),\,\Ncal}:=\inf_{\substack{A,B\in\DF,\,Y\in\mathcal{B}(\cH)\\X=AYB}}\,\norm{A}_{2r,\,\sigma_\tr}\,\norm{B}_{2r,\,\sigma_\tr}\,\norm{Y}_{p,\,\sigma_\tr}\,,\label{eq111}\\	&\norm{Y}_{(p,q),\,\Ncal}:=\underset{A,B\in\DF}{\sup}\,\frac{\norm{A\,Y\,B}_{q,\sigma_{\Tr}}}{\norm{A}_{2r,\sigma_{\Tr}}\norm{B}_{2r,\sigma_{\Tr}}}\,.\label{eq222}
\end{align}
We shall prove that they are particularly well-suited to study the hypercontractivity of the QMS, namely:
\begin{itemize}
\item they reduce to the $\mathbb{L}_p(\sigma)$ norms when the QMS is primitive with unique invariant state $\sigma$;
\item they reduce to the $\mathbb{L}_q(\sigma_{\tr})$ norms when evaluated on $\DF$;
\item the QMS is contractive with respect to these norms for all $p,q\geq1$.
\end{itemize}
When differentiating this norm with respect to $p$, some natural quantities will appear that we will connect with entropic notions in \Cref{ent2p}. Similarly to \cite{OZ99}, we thus introduce a decoherence-free generalisation of the $\mathbb{L}_p$ relative entropies as follows: define the map
\begin{align*}
	S_p(X)=-p\,\partial_s I_{p+s,p}(X)|_{s=0}\,,
\end{align*}
referred to as \textit{operator valued relative entropy}, where $I_{q,p}$ is defined in \Cref{Ipq}. It can be computed explicitly: when $X\ge 0$, 
\begin{align*}
	S_p(X)=\Gamma_{\sigma_{\Tr}}^{-\frac{1}{p}}[\Gamma_{\sigma_{\Tr}}^{\frac{1}{p}}(X)\ln\Gamma_{\sigma_{\Tr}}^{\frac{1}{p}}(X)]-\frac{1}{2p}\{X,\ln{\sigma_{\Tr}}\}\,.
\end{align*}
We then define the \textit{$\operatorname{DF}$-{$\mathbb{L}_p$} relative entropy} associated with the algebra $\cN\equiv \cN(\cP)$ as follows: for $\frac{1}{p}+\frac{1}{q}=1$,
\begin{equation}\label{eq37}
	\operatorname{Ent}_{p,\,\cN}(X)
	:=\langle I_{q,p}(X),S_p(X)\rangle_{\sigma_\tr}-\frac{1}{p} \tr\left[(\Gamma_{\sigma_\tr}^{\frac{1}{p}}(X))^p\ln{E_\cN[\Gamma_{\sigma_\tr}^{-1}(\Gamma_{\sigma_\tr}^{\frac{1}{p}}(X))^p]}\right]\,.
\end{equation}
In the case of a primitive QMS where $\sigma_\tr$ is the unique invariant state of the evolution, $E_\cN[.]:=\tr(\sigma_\tr\,.)\,\mathbb{I}$, so that the last term in the above definition is null and we get back the original definition of \cite{kastoryano2013rapid}, which is denoted by $\operatorname{Ent}_{p,\,\sigma_\tr}(X)$. In general, this term is non-positive, so that $\operatorname{Ent}_{p,\,\cN}(X)\le \operatorname{Ent}_{p,\sigma_\tr}(X) $. In the important cases $p=1$ and $p=2$, \Cref{eq37} reduces to 
\begin{equation}\label{ent}
\begin{aligned}
	&\operatorname{Ent}_{1,\,\cN}(X):= \tr\left[\Gamma_{\sigma_\tr}(X)\left(\ln\frac{\Gamma_{\sigma_\tr}(X)}{\tr(\Gamma_{\sigma_\tr}(X))}-\ln\sigma_\tr\right)\right] -\tr\left[ \Gamma_{\sigma_\tr}(X)\ln\frac{E_\cN[X]}{\tr(\Gamma_{\sigma_\tr}(X))}\right]\,, \\
	&\operatorname{Ent}_{2,\,\cN}(X):=\tr\left(\left[ \Gamma_{\sigma_{\tr}}^{\frac{1}{2}}(X)\right]^2\left(  
	\ln \left[\Gamma_{\sigma_{\tr}}^{\frac{1}{2}}(X) \right]  -\frac{1}{2} \ln E_\cN\left[\Gamma_{\sigma_\tr}^{-1}\left( \Gamma_{\sigma_\tr}^{\frac{1}{2}}(X)\right)^2\right] -\frac{1}{2} \ln \sigma_\tr\right) \right)\,.
\end{aligned}
\end{equation}
We can now introduce the main definitions.
\begin{definition}\label{logsob}
We say that the QMS $(\cP_t)_{t\ge 0}$ of generator $\LL$:
\begin{itemize}
	\item[1)] satisfies a \emph{weak $\operatorname{DF}$-$q$-log-Sobolev inequality} with positive \textit{strong $\operatorname{DF}$-$q$-log-Sobolev constant} $c>0$ and \textit{weak $\operatorname{DF}$-$q$-log-Sobolev constant} $d\geq 0$, condition denoted by $\operatorname{LSI}_{q,\Ncal}(c,d)$, if for all  $X>0$,
	\begin{align}\label{logsob1}\tag{$\operatorname{LSI}_{q,\,\cN}(c,d)$}
		\operatorname{Ent}_{q,\Ncal}(X)\le c~\mathcal{E}_{q,\,\cL}(X)+\frac{2d}{q}\norm{X}_{q,\sigma_\tr}^q\,.
	\end{align}
	\item[2)] is \textit{weakly $q$-$\operatorname{DF}$-hypercontractive} for positive constants $c>0$ and $d\ge 0$, condition denoted by $\operatorname{HC}_{q,\,\cN}(c,d)$, if  
	\begin{align}\tag{$\operatorname{HC}_{q,\,\cN}(c,d)$}
		\|\cP_t(X)\|_{(q,p(t)),\,\cN}\le \exp\left\{ 2d \left( \frac{1}{q}-\frac{1}{p(t)}\right)\right\}\| X\|_{q,\sigma_\tr}\,,
	\end{align} 
	 for any function $p:[0,+\infty)\to\RR$ such that for any $t\ge 0$, $  q\le p(t)\le 1+ (q-1)~\e^{2t/c}$.
\end{itemize}
\end{definition}

The first main result of this article is the following generalisation of Gross' integration lemma that establishes the equivalence between hypercontractivity and the log-Sobolev inequality for a decohering QMS:

\begin{theorem}Let $(\mathcal{P}_t)_{t\ge 0}$ be a decohering QMS on $\cB(\cH)$ and let $q\ge 1$, $c>0$ and $d\ge0$. Then
	\begin{enumerate}\label{gross}\label{theogross}
		\item[(i)] If $\operatorname{HC}_{q,\,\cN}(c,d)$ holds, then $\operatorname{LSI}_{q,\,\cN}(c,d)$ holds.
	\item[(ii)] If $\operatorname{LSI}_{\tilde q,\,\Ncal}(c,d)$ holds for all $\tilde q\ge q$, then $\operatorname{HC}_{q,\,\Ncal}(c,d+\ln |I|)$ holds, where $|I|$ denotes the number of blocks of $\cN(\cP)$ in \Cref{eqtheostructlind2}. 
	\end{enumerate}
\end{theorem}

This theorem is quite surprising compared to the (classical and quantum) primitive case or the (classical) non-primitive case, where there is an exact equivalence between hypercontractivity and the logarithmic Sobolev inequality (i.e. with the same constant). In general, this $\ln|I|$ appearing here is not optimal (see \Cref{normestimate}). Even if we do not know if equivalence holds, we strongly believe that it is not the case in general.\\

The case where $\DF$ is a factor and where the QMS is unital and trace-preserving was proved in \cite{[BK16]}, but only in the case $d=0$. However, the authors failed to give an example where the constant $c$ is finite. We shall actually prove in \Cref{sec5bis} that this is impossible. More generally, we prove that, as soon as the QMS is truly non-primitive and non-invertible (that is, not a unitary evolution), necessarily $c<+\infty$ implies $d>0$.\\

Remark also that the last statement is weaker than in the classical case, when one only needs to assume that the weak LSI holds for $\tilde q=q$. This is due to the fact that the following regularity condition always holds in the commutative setting, which ensures that $\operatorname{LSI}_{ q,\,\Ncal}(c,d)$ implies $\operatorname{LSI}_{\tilde q,\,\Ncal}(c,d)$ for all $\tilde q\ge q$. This condition needs to be assumed in the general quantum setting, even in the primitive case. A generator $\LL$ of a QMS $(\cP_t)_{t\ge 0}$ is called \textit{weakly $\Lbb_p$-regular} if there exists $d_0\ge 0$ such that for all $p\ge 1$ and all $X\in\cB_{sa}(\cH)$, 
\begin{align}\tag{$\operatorname{w\,-}\mathbb{L}_p(d_0)$}\label{weaklp}
	\cE_{p,\Lcal}(X)\ge \left\{\begin{aligned}
		&\cE_{2,\,\Lcal}(I_{2,p}(X))-d_0\|X\|_{p,\sigma_\tr}^p, ~~~~~~~~~~~~~~~~~~~~~1\le p\le 2\,,\\
		&(p-1)\left(\cE_{2,\,\Lcal}(I_{2,p}(X))-d_0\|X\|_{p,\sigma_\tr}^p\right),~~~~~~\,~~p\ge 2\,.
	\end{aligned}\right.
\end{align}
Moreover, $\LL$ is said to be \textit{strongly $\Lbb_p$-regular} if there exists $d_0\ge 0$ such that for all $p\ge 1$ and all $X\in\cB_{sa}(\cH)$, 
\begin{align}\tag{$\operatorname{s\,-}\mathbb{L}_p(d_0)$}\label{stronglp}
	d_0\|X\|_{p,\sigma_\tr}^p+\frac{p}{2}\,\cE_{p,\Lcal}(X)\ge  \cE_{2,\,\Lcal}(I_{2,p}(X))\,.
\end{align}
With these definitions, we can prove the following theorem.

\begin{theorem}\label{thmgross}\label{lsi2top}
	Assume that $\operatorname{LSI}_{2,\,\cN}(c,d)$ holds. Then
	\begin{itemize}
		\item[(i)] If the generator $\LL$ is strongly $\Lbb_p$-regular for some $d_0\ge 0$, then $\operatorname{LSI}_{q,\,\cN}(c,d+c\,d_0)$ holds for all $q\geq1$, so that $\operatorname{HC}_{2,\,\cN}(c,d+\ln|I|+c\,d_0)$ holds.
		\item[(ii)] If the generator $\LL$ is only weakly $\Lbb_p$-regular for some $d_0\ge 0$, then $\operatorname{LSI}_{q,\,\cN}({2c}, d+c\,d_0)$ holds for all $q\geq1$, so that $\operatorname{HC}_{2,\,\cN}(2c,d+\ln|I|+c\,d_0)$ holds. 
	\end{itemize}	
\end{theorem}

The last two theorems generalise Theorem 3.8 of \cite{OZ99} as well as Theorem 15 of \cite{[KT13]}. Moreover, it was conjectured in \cite{[KT13]} that primitive QMS are weakly $\Lbb_p$-regular with $d_0=0$, and that reversible QMS are strongly $\Lbb_p$-regular, again with $d_0=0$. This second fact was recently shown to hold in \cite{BarEID17} under the condition of $\sigma_\tr$-DBC and without the primitive assumption. For reversible QMS, a straightforward extension of the proof of Proposition 5.2 of \cite{OZ99} implies that the strong regularity of $\LL$ always holds, with $d_0=\|\LL\|_{2\to 2,\,\sigma_\tr}+1:=\sup_{\|X\|_{2,\,\sigma_\tr}=1}\|\LL(X)\|_{2,\,\sigma_\tr}+1$. These remarks motivate the following corollary of \Cref{thmgross}:
\begin{corollary}\label{cor4}
Assume that $\operatorname{LSI}_{2,\,\cN}(c,d)$ holds. Then:
	\begin{itemize}
		\item[(i)]  If $\cL$ is reversible, then $\operatorname{HC}_{2,\,\cN}(c,d+\ln|I|+c\,(\|\LL\|_{2\to2,\,\sigma_\tr}+1))$ holds. 
		\item[(ii)] If $\LL$ satisfies $\sigma_\tr$-$\operatorname{DBC}$, then $\operatorname{HC}_{2,\,\cN}(c,d+\ln|I|)$ holds.
	\end{itemize}
\end{corollary}	
 We also prove that it is always possible to get a weak $\operatorname{DF}$-$2$-log-Sobolev inequality with a universal weak DF-$2$-log-Sobolev constant from any weak DF-$2$-log-Sobolev inequality, hence extending Theorem 4.2 of \cite{OZ99} to the non-primitive case. Recall that the spectral gap is defined as follows \cite{BarEID17}:
\begin{align*}
	\lambda(\LL):=\inf_{X>0}\frac{\cE_{2,\Lcal}(X) }{\|X-E_\cN[X]\|^2_{2,\sigma_\tr}}\,.
\end{align*}
\begin{theorem}\label{theo_wLSI}
	Assume that $\operatorname{LSI}_{2,\Ncal}(c,d)$ holds and denote by $\lambda(\Lcal)$ the spectral gap of $\Lcal$. Then $\operatorname{LSI}_{2,\Ncal}(c+\frac{d+1}{\lambda(\Lcal)},d'=\ln \sqrt{2})$ holds.
\end{theorem} 
FInally, using the DF-hypercontractivity and complex interpolation methods, we derive the following universal DF-$2$-log-Sobolev constants:
\begin{corollary}\label{coro_univconstants1}\label{coro_univconstants}
	Given a reversible QMS $(\cP_t)_{t\ge 0}$ with spectral gap $\lambda(\LL)$, $\operatorname{LSI}_{2,\Ncal}(c,\ln\sqrt{2})$ holds, with
	\begin{align*}
	{	c\leq\frac{\ln(\| \sigma_\tr^{-1} \|_\infty)+2}{2\,\lambda(\LL)} \,.}
	\end{align*}	
	\end{corollary}
	
\subsection{Application to decoherence rates}\label{sect25}

Given a QMS $(\cP_t)_{t\ge 0}$, its \textit{decoherence time} is defined as:
\begin{align*}
	\tau_{deco}(\eps):=\inf \left\{    t\ge 0:~\| \cP_{*t}\left(\rho-E_{\cN*}(\rho)\right)\|_1\le \eps,~\forall \rho\in\cD(\cH)\right\}\,.
\end{align*}	
The standard method to obtain estimates for $\tau(\eps)$ in the primitive case is to use Pinsker's inequality to upper bound the trace distance in terms of the relative entropy, which in the primitive case decay exponentially fast according to the 1-log-Sobolev constant \cite{Diaconis1996a,[KT13]}. The second step is to bound this constant by the strong 2-log-Sobolev constant, under the condition that the weak constant is null. However we prove in \Cref{sec5bis} that the weak constant is null only for primitive and unitary evolution. In the case when there is only access to a weak DF-log-Sobolev inequality, we can fortunately still derive bounds on the decoherence times by extending a technique already used in the classical case in \cite{zegarlinski1995ergodicity,Diaconis1996a}, by combining Poincar\'{e}'s inequality and the weak DF-hypercontractivity property of the semigroup.
\begin{proposition}\label{prop1bis}
	Assume that a QMS $(\cP_t)_{t\ge 0}$ satisfies $\operatorname{HC}_{2,\,\cN}(c,d)$, and that $\|\sigma_\tr^{-1}\|_\infty \ge \e$. Then, given $t=\frac{c}{2}\ln\ln\|\sigma_\tr^{-1}\|_\infty+\frac{\kappa}{\lambda(\cL)},~\kappa>0$:
	\begin{align}\label{bounddeco}
		\forall \rho\in\cD(\cH),~~~~~\|\cP_{*t}\left(\rho-E_{\Ncal*}[\rho]\right)\|_1\le \max_{i\in I} \sqrt{d_{\Hcal_i}}\,\e^{1+d-\kappa},
	\end{align}
	where the $d_{\Hcal_i}$ are the dimensions of the spaces $\Hcal_i$ occuring in the decomposition of $\DF$ given by \eqref{eqtheostructlind2}. The above inequality provides the following bound on the decoherence time of the QMS:
	\begin{align*}
		\tau_{\operatorname{deco}}(\eps)\le  \frac{\ln\left(\max_{i\in I}\,\sqrt{d_{\cH_i}}\,\eps^{-1}\right)+1+d}{\lambda(\cL)}+\frac{c}{2}\ln\ln\|\sigma_\tr^{-1}\|_\infty\,.
	\end{align*}	
\end{proposition}

Remark that the assumption on $\norm{\sigma_\tr^{-1}}_\infty$ is not restrictive: it means that the lowest eigenvalue of $\sigma_\tr$ has to be smaller than $1/\e$. In particular, it always holds when $d_\Hcal\geq3$.

We see that having a weak constant $d=\sqrt2$ has in practice no effect on the decoherence-time. Remark also that the constant $\max_{i\in I} \sqrt{d_{\Hcal_i}}$ is again a signature of the non-primitive case. We will see that in some interesting examples it is polynomial in $\ln (d_\Hcal)$ and therefore is dominated by the exponentially decaying term.

\subsection{CB hypercontractivity and the tensorization property}\label{sect26}
For two finite dimensional Hilbert spaces $\Hcal_A$ and $\Hcal_B$ and a full rank density matrix $\sigma$ on $\Hcal_B$, Pisier defined the $\mathbb{L}_q\left(\frac{{\mathbb I}_{\cH_A}}{d_A},\mathbb{L}_p(\sigma)\right)$ norm for $1\le q\le p\le +\infty$ as (\cite{[DJKR16],[P93]}): given $\frac{1}{r}=\left|\frac{1}{p}-\frac{1}{q}\right|$,
\begin{align*}
	& \|X\|_{\mathbb{L}_q\left(\frac{{\mathbb I}_{\cH_A}}{d_A},\,\mathbb{L}_p(\sigma)\right)}  \equiv   
	\inf_{\substack{\,A,B\in\cB(\cH_A),\,Y\in\mathcal{B}(\cH_B)\\X=(A\otimes \mathbb{I}_{\cH_B})Y(B\otimes \mathbb{I}_{\cH_B})
}}\,\norm{A}_{2r,\frac{{\mathbb I}_{\cH_A}}{d_A}}\,\norm{B}_{2r,\frac{{\mathbb I}_{\cH_A}}{d_A}}\,\norm{Y}_{p,\frac{\mathbb{I}_{\cH_A}}{d_{\cH_A}}\otimes\sigma}\,,\\
&	\|Y\|_{\mathbb{L}_p\left(\frac{{\mathbb I}_{\cH_A}}{d_A},\,\mathbb{L}_q(\sigma)\right)}\equiv
	\underset{A,B\in\cB(\cH_A)}{\sup}\,\frac{\norm{(A\otimes \mathbb{I}_{\cH_B})\,Y\,(B\otimes \mathbb{I}_{\cH_B})}_{q,\frac{{\mathbb I}_{\cH_A}}{d_A}\otimes\sigma}}{\norm{A}_{2r,\frac{{\mathbb I}_{\cH_A}}{d_A}}\norm{B}_{2r,\frac{{\mathbb I}_{\cH_A}}{d_A}}}\,.
\end{align*}
When $\sigma=\frac{{\mathbb I}_{\cH_B}}{d_B}$, these reduce to the norms introduced in \cite{[BK16]}. The norms defined in \Cref{eq111,eq222} reduce to the above norms in the particular situation where $\cH=\cH_A\otimes \cH_B$, $\sigma_{\tr}=\frac{{\mathbb I}_{\cH_A}}{d_A}\otimes\sigma$ and $\Ncal=\Bcal(\Hcal_A)\otimes \text{id}_{\cH_B}$. It is then immediate that for all $X\in\Bcal(\Hcal_A\otimes\Hcal_B)$ and all $p,q\geq1$:
\begin{align*}
	\norm{X}_{(q,p),\,\Ncal}=\|X\|_{\mathbb{L}_q\left(\frac{{\mathbb I}_{\cH_A}}{d_A},\,\mathbb{L}_p(\sigma)\right)}\,.
\end{align*}
This situation is particularly relevant when studying hypercontractivity for the CB-norms. For an operator $\Lambda: \cB(\cH_B)\to \cB(\cH_B) $, its \textit{weighted completely bounded norm} $\|\Lambda\|_{q\to p,\CB,\sigma}$ is defined as follows:
\begin{align}\label{CBnorm}
	\|\Lambda\|_{q\to p,\CB,\sigma}:= \sup_{d_{\cH_A}}\,\sup_{Y\in \cB( \cH)}\frac{\|  ( {\id}_{\cB(\cH_A)}\otimes\Lambda)(Y) \|_{(q,p),\,\cN}}{\| Y\|_{q, \frac{\mathbb{I}_{\cH_A}}{d_{\cH_A}}   \otimes \sigma }    }\,,
\end{align}
where the supremum in \reff{CBnorm} is over all dimensions $d_{\cH_A}$ of $\cH_A$ and all operators $Y\in \cB( \Hcal)$. 

These norms are known to be multiplicative, as proved in \cite{[DJKR16]}. As a result, in order to define a notion of hypercontractivity and its associated log-Sobolev inequality that satisfy the tensorization property, we embed a primitive QMS $(\cP_t)_{t\ge 0}$ on $\cB(\cH)$ into the QMS $( \id_{k}\otimes\cP_t)_{t\ge 0}$ on $\cB(\CC^k\otimes \cH)$, and study the latter's $\operatorname{DF}$-hypercontractivity properties, for each integer $k\ge 1$. Let $\sigma$ be the unique invariant state of $(\Pcal_t)_{t\geq0}$. Then $\cN_k:=\cN( \id_{k} \otimes \cP )=    \cB({\CC^{k}})\otimes \mathbb{I}_{\cH}$ and $\sigma_{\tr}= \frac{\mathbb{I}_{{k}}}{k}\otimes \sigma$. We are lead to the following definitions.

\begin{definition}\label{CBHCLSI}
We say that $(\cP_t)_{t\ge 0}$:
	\begin{itemize}
		\item[1)] satisfies a \textit{weak $\CB$-q-log-Sobolev inequality} with positive \textit{strong $\operatorname{CB}$-$q$-log-Sobolev constant} $c>0$ and \textit{weak $\operatorname{CB}$-$q$-log-Sobolev constant} $d\geq 0$, which we denote by $\operatorname{LSI}_{q,\CB}(c,d)$, if for all integer $k\geq1$, $\operatorname{LSI}_{q,\,\cN_k}(c,d)$ holds.
		\item[2)] is \textit{weakly $q$-$\CB$-hypercontractive} for positive constants $c>0$ and $d\ge 0$, condition denoted by $\operatorname{HC}_{q,\CB}(c,d)$, if for all $t\ge 0$,
		\begin{align*}
			\| \cP_t\|_{q\to p(t),\CB,\sigma}\le \exp\left(  2d\left(  \frac{1}{q}-\frac{1}{p(t)} \right)  \right)\,,
		\end{align*}
		for any function $p:[0,+\infty)\to\RR$ such that for any $t\ge 0$, $q\le p(t)\le 1+(q-1)\e^{2t/c}$.
	\end{itemize}	
\end{definition}

The above definitions extend the ones in \cite{[BK16]} to non-unital primitive QMS and to weak LSI and weak HC. In the next theorem, we establish the equivalence between the $\CB$-log-Sobolev inequality and $\CB$-hypercontractivity, hence extending Theorem 4 of \cite{[BK16]} to the cases mentioned above.

\begin{theorem}\label{grossCB}\label{grosscb}Let $(\mathcal{P}_t)_{t\ge 0}$ be a primitive QMS on $\cB(\cH)$ with associated generator $\LL$, and let $q\ge 1$, $d\ge0$ and $p(t)=1+(q-1)\e^{2t/c}$ for some constant $c>0$. Then
	\begin{enumerate}
		\item[(i)] If $\operatorname{HC}_{q,\CB}(c,d)$ holds, then $\operatorname{LSI}_{q,\CB}(c,d)$ holds.
		\item[(ii)] If $\operatorname{LSI}_{p(t),\CB}(c,d)$ holds for all $t\ge 0$, then $\operatorname{HC}_{q,\CB}(c,d)$ holds.
	\end{enumerate}
\end{theorem}
A direct application of the definitions for $\mathbb{L}_p$ regularity of Dirichlet forms then leads to the following:
\begin{theorem}
	Assume that $\operatorname{LSI}_{2,\,\CB}(c,d)$ holds. Then
	\begin{itemize}
		\item[(i)] If the generator $\LL$ is strongly $\Lbb_p$-regular for some $d_0\ge 0$, then $\operatorname{LSI}_{q,\,\CB}(c,d+c\,d_0)$ holds for all $q\geq1$, so that $\operatorname{HC}_{2,\,\CB}(c,d+c\,d_0)$ holds.
		\item[(ii)] If the generator $\LL$ is only weakly $\Lbb_p$-regular for some $d_0\ge 0$, then $\operatorname{LSI}_{q,\,\CB}({2c}, d+c\,d_0)$ holds for all $q\geq1$, so that $\operatorname{HC}_{2,\,\CB}(2c,d+c\,d_0)$ holds. 
	\end{itemize}	
\end{theorem}
As in the decoherence-free case, an application of Proposition 5.2 of \cite{OZ99} together with Theorem 4 of \cite{watrous2004notes} leads to the following corollary:

\begin{corollary}Assume that $\LL$ is the generator of a primitive QMS with unique invariant state $\sigma$, and that $\operatorname{LSI}_{2,\,\operatorname{CB}}(c,d)$ holds.
	\begin{itemize}
		\item[(i)] If $\cL$ is reversible, then $\operatorname{LSI}_{q,\,\operatorname{CB}}(c,d+c\,(\|\cL\|_{2\to 2,\,\sigma}+1))$ holds for all $q\geq1$ and consequently $\operatorname{HC}_{2,\,\CB}(c,d+c\,(\|\cL\|_{2\to 2,\,\sigma}+1))$ holds. 
		\item[(ii)] If $\LL$ satisfies $\sigma$-$\operatorname{DBC}$, then $\operatorname{LSI}_{q,\,\operatorname{CB}}(c,d)$ holds for all $q\geq1$ and consequently $\operatorname{HC}_{2,\,\CB}(c,d)$ holds.
	\end{itemize}
\end{corollary}
Moreover, we derive universal bounds on the CB-log-Sobolev constants:
	\begin{theorem}[Universal bounds on the $\CB$-log Sobolev constants]\label{coro_univconstantsCB1}\label{theo}
	Let $(\cP_t)_{t\ge 0}$ be a primitive reversible QMS, with unique invariant state $\sigma$ and spectral gap $\lambda(\LL)$. Then, $\operatorname{LSI}_{2,\,\CB}(c,\ln\sqrt{2})$ holds, with 
	\begin{align}
		c\leq \frac{  \ln \| \sigma^{-1}\|_{\infty} +2  }{2\lambda(\LL)}\,.
	\end{align}		
\end{theorem}

\section{The weighted $\mathbb{L}_q(\cN,\mathbb{L}_p(\sigma_\tr))$ norms}\label{norms}

Hypercontractivity is a statement concerning the contraction properties of a certain family of norms under the action of a QMS $(\cP_t)_{t\ge 0}$. Perhaps the main contribution of this article is the study and use of such a family, specific to the QMS and its decoherence-free algebra. The origin of these norms comes from operator space theory; they were recently defined by Junge and Parcet in \cite{JP10} and can be seen as a generalisation of the norms defined on non-commutative vector-valued $\mathbb{L}_p$ spaces by Pisier in \cite{[P93]}. For $1\leq q\le p\le+ \infty$, and $\frac{1}{r}=\frac{1}{q}-\frac{1}{p}$, define
\begin{align}\label{eq_def_DFnorms}
	&\norm{X}_{(q,p),\,\Ncal}:=\inf_{\substack{A,B\in\DF,\,Y\in\mathcal{B}(\cH)\\X=AYB}}\,\norm{A}_{2r,\,\sigma_\tr}\,\norm{B}_{2r,\,\sigma_\tr}\,\norm{Y}_{p,\,\sigma_\tr},\\
&	\norm{Y}_{(p,q),\,\Ncal}:=\underset{A,B\in\DF}{\sup}\,\frac{\norm{A\,Y\,B}_{q,\sigma_{\Tr}}}{\norm{A}_{2r,\sigma_{\Tr}}\norm{B}_{2r,\sigma_{\Tr}}}\,.\label{eq_def_norm_q}
\end{align}
For any $1\le q,p\le+ \infty$, we denote the space $\cB(\cH)$ endowed with the norms $\|X\|_{(q,p),\,\cN}$ by {$\mathbb{L}_q(\cN,\mathbb{L}_p(\sigma_\tr))$}. We refer the reader to \cite{JP10} for the proof that it defines an interpolating family of spaces. In \Cref{subsec1}, we study the properties of these norms and show the reason why they constitute good candidates for the study of hypercontractivity of decohering QMS. In \Cref{subsec2}, we state one of the main results of this article: a formulation of Gross' integration Lemma. We conclude in \Cref{unif} with a result on the almost convexity of the norm that will be essential in the next section.

\subsection{Some properties of the $\mathbb{L}_q(\cN, \mathbb{L}_p(\sigma_\tr))$ spaces}\label{subsec1}

We first gather some properties of the $\mathbb{L}_q(\cN,\,\mathbb{L}_p(\sigma_{\tr}))$ spaces. First, we will repeatedly use the crucial fact that they define a family of complex interpolating spaces \cite{JP10}. We refer to the latter citation for a proof of this and for more information about these norms.

\begin{proposition}\label{prop_duality}
Let $1\le q,p\le +\infty$ together with their H\"{o}lder conjugates $q',p'$, i.e. such that $\frac{1}{p}+\frac{1}{p'}=1$, and $\frac{1}{q}+\frac{1}{q'}=1$. Moreover, let $\cN$ be a subalgebra of $\cB(\cH)$ with corresponding conditional expectation $E_\cN$ and $\sigma_\tr:=E_{\cN*}(\mathbb{I}/d_\cH)$. Then the following holds:
\begin{itemize}
	\item[(i)] H\"{o}lder's inequality: For any $X\in\mathbb{L}_p(\cN,\mathbb{L}_q(\sigma_\tr))$ and $Y\in\mathbb{L}_{p'}(\cN,\mathbb{L}_{q'}(\sigma_\tr))$,
	\begin{align*}
		|\langle X,Y\rangle_{\sigma_\tr}|\le \|X\|_{(q,p),\,\cN}\|Y\|_{(q',p'),\,\cN}\,.
		\end{align*}
	\item[(ii)] Duality: For any $X\in\mathbb{L}_q(\cN, \mathbb{L}_{p}(\sigma_\tr))$, 
	\begin{align*}
		\|X\|_{(q,p),\,\cN}=\sup\left\{ |\langle X,Y\rangle_{\sigma_\tr}|:~\|Y\|_{(q',p'),\,\cN}= 1\right\} \end{align*}
	\item[(iii)] Relation with $\mathbb{L}_p(\sigma_\tr)$ norms: if $q\le p$, then for any $X\in\mathbb{L}_p(\sigma_\tr)$,
	\begin{align}\label{eq24}
		& \|X\|_{q,\sigma_\tr}\le\|X\|_{(q,p),\,\cN}\le \|X\|_{p,\sigma_\tr},\\
		& \|X\|_{q,\sigma_\tr}\le \|X\|_{(p,q),\,\cN}\le \|X\|_{p,\sigma_\tr}\,.\label{eq25}
	\end{align}	
and, in both cases, equality holds for all $X$ if $p=q$. This last statement is usually referred to as Fubini's Theorem.
\item[(iv)] The hierarchy of norms: for $1\le q_1\le q_2, p_1\le p_2\le +\infty$, and any $X\in\cB(\cH)$,
\begin{align*}
	\|X\|_{(q_1,p_1),\,\cN}\le \|X\|_{(q_2,p_2),\,\cN}\,.
\end{align*}	
	\item[(v)] When $1\leq q\leq p\le +\infty$, the $\sup$ on the right hand side of \Cref{eq_def_norm_q} may be restricted to the set of positive semidefinite operators $A,B\ge 0$. Furthermore, for all positive semidefinite $X$,
	\begin{equation}\label{eq_prop_normq_positive}
		\norm{X}_{(p,q),\,\Ncal}=\underset{\substack{
		A\in\cN,~A>0, ~\|A\|_{1,\sigma_\tr}=1}}{\sup}\,\norm{A^{1/2r}\,X\,A^{1/2r}}_{q,\sigma_{\Tr}}\,,
	\end{equation}
\item[(vi)] Similarly, the $\inf$ on the right hand side of \Cref{eq_def_DFnorms} may be restricted to the set of positive semidefinite operators $A,B\ge 0$. Furthermore, for all positive semidefinite $X$,
	\begin{equation}\label{eq_prop_normp_positive}
		\norm{X}_{(q,p),\,\Ncal}=\underset{A\in\cN,~A>0,~\|A\|_{1,\sigma_\tr}=1}{\inf}\,\norm{A^{-1/2r}\,X\,A^{-1/2r}}_{p,\sigma_{\Tr}}\,.
	\end{equation}
	\item[(vii)] For all $1\le q\le p\le +\infty$, $\|X\|_{(q,p),\,\cN}=\|X\|_{q,\,\sigma_\tr}$ whenever $X\in\cN$. 
\end{itemize}	
\end{proposition}

\begin{proof}
\begin{enumerate}
	\item[(i)] H\"{o}lder's inequality follows directly from H\"{o}lder's inequality in the case of the $\mathbb{L}_p(\sigma_\tr)$ norms (see \cite{OZ99}): without loss of generality, assume that $p\le q$, so that $q'\le p'$. Consider any decomposition of $Y$ of the form $Y=AZB$, with $A,B\in\cN$ and $Z\in\cB(\cH)$. Then,
	\begin{align*}
		|\langle X,Y\rangle_{\sigma_\tr}|=|\langle X,AZB\rangle_{\sigma_\tr}|&=|\langle A^*XB^*,Z\rangle_{\sigma_\tr}|\\
		&\le \|A^*XB^*\|_{p,\sigma_\tr}\|Z\|_{p',\sigma_\tr}\\
		&\le \|X\|_{(q,p),\,\cN}\|A\|_{2r,\sigma_\tr}\|B\|_{2r,\sigma_\tr}\|Z\|_{p',\sigma_\tr}\,.
	\end{align*}	 
We conclude by taking the infimum over the operators $A,B$ and $Z$.
\item[(ii)] Assume without loss of generality that $1\le q\le p\le +\infty$. Then
\begin{align*}
\|X\|_{(p,q),\, \cN}&=\sup_{A,B\in\cN}\{ \|AXB\|_{q,\sigma_\tr}:~\|A\|_{2r,\sigma_\tr}\|B\|_{2r,\sigma_\tr}\le 1\}\\
&=\sup_{A,B\in\cN,Z\in\cB(\cH)}\{ |\langle AXB,Z\rangle_{\sigma_\tr}|:~\|A\|_{2r,\sigma_\tr}\|B\|_{2r,\sigma_\tr}\le 1,\|Z\|_{q',\sigma_\tr}\le 1\}\\
&\le \sup_{A,B\in\cN,Z\in\cB(\cH)}\{ |\langle X,A^*ZB^*\rangle_{\sigma_\tr}|:~\|A\|_{2r,\sigma_\tr}\|B\|_{2r,\sigma_\tr} \|Z\|_{q',\sigma_\tr}\le 1\}\\
&=\sup_{A,B\in\cN,W,Z\in\cB(\cH)}\{ |\langle X,W\rangle_{\sigma_\tr}|:~W=A^*ZB^*,~\|A\|_{2r,\sigma_\tr}\|B\|_{2r,\sigma_\tr} \|Z\|_{q',\sigma_\tr}\le 1\}\\
&\le\sup_{W\in\cB(\cH)}\{ |\langle X,W\rangle_{\sigma_\tr}|:~\|W\|_{(p',q'),\,\cN}\le 1\}\,,
\end{align*}
where in the second line, we used the duality of $\mathbb{L}_p(\sigma_\tr)$ norms, in the third line we used that for $A,B\in\cN$, $[A,\sigma_\tr]=[B,\sigma_\tr]=0$, and in the last line we used that $\frac{1}{r}=\frac{1}{p'}-\frac{1}{q'}$. Using H\"{o}lder's inequality (i), the condition $\|W\|_{(p',q'),\,\cN}\le 1$ implies
\begin{align*}
|\langle X,W\rangle_{\sigma_\tr}|\le \|X\|_{(p,q),\,\cN}\|W\|_{(p',q'),\, \cN}\le\|X\|_{(p,q),\,\cN}\,.
\end{align*}
Therefore, the supremum is attained. This shows that the Banach space $\mathbb{L}_{p'}(\cN, \mathbb{L}_{q'}(\sigma_\tr))$ is the dual of $\mathbb{L}_p(\cN, \mathbb{L}_q(\sigma_\tr))$. As these spaces are finite dimensional, the converse holds.
\item[(iii)] The second inequality in \reff{eq24} and the first inequality in \reff{eq25} are obvious by definition. The second inequality in \reff{eq25} and the first inequality in \reff{eq24} are proved by a use of H\"{o}lder's inequality for the $\mathbb{L}_p(\sigma_\tr)$ norms.
\item[(iv)] By convexity of the inverse function, $\frac{1}{r_1}\equiv\frac{1}{q_1}-\frac{1}{p_1}\ge \frac{1}{q_2}-\frac{1}{p_2}\equiv\frac{1}{r_2}$, so that
\begin{align*}
	\|X\|_{(q_1,p_1),\,\cN}&=\inf_{A,B\in\cN, \,Y\in\cB(\cH),\, X=AYB}\,\norm{A}_{2r_1,\,\sigma_\tr}\,\norm{B}_{2r_1,\,\sigma_\tr}\,\norm{Y}_{p_1,\,\sigma_\tr}\\
	&\le\inf_{A,B\in\cN, \,Y\in\cB(\cH),\, X=AYB}\,\norm{A}_{2r_2,\,\sigma_\tr}\,\norm{B}_{2r_2,\,\sigma_\tr}\,\norm{Y}_{p_2,\,\sigma_\tr}\\
	&=	\|X\|_{(q_2,p_2),\,\cN}\,,
	\end{align*}
where in the second line we used the hierarchy of the $\|.\|_{p,\sigma_\tr}$ norms: for $p\le p'$, $\|X\|_{p,\sigma_\tr}\le \|X\|_{p',\sigma_\tr}$.
\item[(v)] The first claim follows directly from invariance of $\cN$ under $A\mapsto |A|\equiv \sqrt{A^*A}$, polar decomposition, as well as invariance of the $\mathbb{L}_p(\sigma_\tr)$ norms under unitary transformations $U\in\cN$. Assume now that $X\ge0$. Then, by H\"{o}lder's inequality for the Schatten norms,
\[\norm{AXB}_{q,\sigma_{\Tr}}\leq\sqrt{\norm{AXA^*}_{q,\sigma_{\Tr}}\norm{BXB^*}_{q,\sigma_{\Tr}}}\leq\max\left\{\norm{AXA^*}_{q,\sigma_{\Tr}}\,,\,\norm{BXB^*}_{q,\sigma_{\Tr}}\right\}\,,\]
where we also used that $\Gamma_{\sigma_{\Tr}}(AXB^*)=A\Gamma_{\sigma_{\Tr}}(X)B^*$. Moreover, equality holds when $A=B$. Since positive definite operators are dense in the set of positive semidefinite operators, we conclude that for all positive semidefinite $X$,
\[\norm{X}_{(p,q),\Ncal}=\underset{A\in\cN,~A>0,~\|A\|_{1,\sigma_\tr}=1}{\sup}\,\norm{A^{1/2r}\,X\,A^{1/2r}}_{q,\sigma_{\Tr}}\,.\]
\item[(vi)] This properties is more difficult to prove than the previous one. We refer to point (iv) of Proposition 4.1.5 in \cite{xu2007operator}. 
\item[(vii)] From the first inequality of \reff{eq24}, we only need to find $A,B\in\cN$ and $Y\in\cB(\cH)$ such that $X=AYB$, and $\|X\|_{q,\sigma_\tr}=\|A\|_{2r,\sigma_\tr}\|B\|_{2r,\sigma_\tr}\|Y\|_{p,\sigma_\tr}$. This works by taking $A=B=X^{\frac{q}{2r}}$ and $Y=X^{\frac{q}{p}}$. Indeed, in this case,
\begin{align*}
	\|A\|_{2r,\sigma_\tr}=	\|B\|_{2r,\sigma_\tr}=(\tr (\sigma_\tr |X|^{q}))^{\frac{1}{2r}}=\|X\|_{q,\sigma_\tr}^{\frac{q}{2r}},~~~~~\|Y\|_{p,\sigma_\tr}=\tr(\sigma_\tr |X|^{q})^{\frac{1}{p}}=\|X\|_{q,\sigma_\tr}^{\frac{q}{p}}\,,
\end{align*}	
and the claim follows from the fact that $\frac{1}{r}+\frac{1}{p}=\frac{1}{q}$.
\end{enumerate}
\end{proof}

In the following proposition, we gather properties of $\mathbb{L}_p(\cN,\,\mathbb{L}_q(\sigma_\tr))$, when $\cN\equiv\cN(\cP)$ is the decoherence-free algebra of a decohering QMS $(\cP_t)_{t\ge 0}$, that will be particularly useful throughout the paper: 
\begin{proposition}\label{theo_propr_norms}
Fix $1\leq q\leq p\leq +\infty$ and let $(\cP_t)_{t\ge 0}$ be a decohering QMS, with $\cN\equiv \cN(\cP)$. Then the following properties hold:
\begin{enumerate}
\item[(i)] $(\cP_t)_{t\ge 0}$ is contractive with respect to $\|.\|_{(q,p),\, \cN}$ for all $1\le q, p\le +\infty$.
\item[(ii)] For all $X\in\cN(\cP)$, $\|X\|_{(q,p),\, \cN}=\|X\|_{q,\sigma_\tr}$.
\item[(iii)] Ordering of the norms: for fixed $q\geq 1$ and for $q\leq p_1 \leq p_2$, $\|.\|_{(q,p_1),\, \cN}\leq\|.\|_{(q,p_2),\, \cN}$.
\item[(iv)] In the case when $\cN=\cN(\cP)\equiv \CC\,\mathbb{I}$ and $\sigma_\tr\equiv \sigma$ is its unique invariant state, equality holds in the second inequality of \reff{eq24} as well as the first inequality of \reff{eq25}.
\end{enumerate}
\end{proposition}

\begin{proof}
\begin{enumerate}
	
	\item[(i)] We first prove that $(\cP_t)_{t\ge 0}$ is contractive for the $\norm{\cdot}_{p,\sigma_\tr}$ norm for all $p\geq1$ and all $t\geq0$, i.e. for all $X\in\Bcal(\Hcal)$,
	\[\norm{\Pcal_t(X)}_{p,\sigma_\tr}\leq\norm{X}_{p,\sigma_\tr}\,.\]
	 We introduce the sandwiched R\'enyi entropy between two density matrices $\rho$ and $\sigma$ in $\Dcal(\Hcal)$, defined by:
	\[{D}_p(\rho||\sigma):=\frac{1}{p-1}\ln\left(\norm{\sigma^{-\frac12}\,\rho\,\sigma^{-\frac12}}_{p,\sigma}^p\right)\,.\]
	we thus have for all $X\in\Bcal(\Hcal)$, writing $\rho=\sigma_\tr^{\frac12}\,X\,\sigma_\tr^{\frac12}$:
	\[\norm{X}_{p,\sigma_\tr}=\exp\left(\left(1-\frac1p\right){D}_p(\rho||\sigma_\tr)\right)\,.\]
	Using the fact that $\Pcal_t(X)=\sigma_\tr^{-\frac12}\hat\Pcal_{t*}(\rho)\,\sigma_\tr^{-\frac12}$, where $\hat{\cP_t}$ denotes the adjoint of $\cP_t$ for the inner product $\langle .,.\rangle_{\sigma_\tr}$, and that $\sigma_\tr$ is an invariant state of $\hat\Pcal_{*t}$, the contractivity of $(\Pcal_t)_{t\geq0}$ with respect to the $p$-norm reduced to the datta-processing inequality for the sandwiched R\'enyi entropy for $p\geq1$, proved e.g.~in \cite{beigi2013sandwiched}: 
	\[{D}_p\left(\hat\Pcal_{*t}(\rho)||\hat\Pcal_{*t}(\sigma)\right)\leq {D}_p(\rho||\sigma_\tr)\,.\]
	Assume now that $1\le q\le p\le +\infty$. We first prove that $(\cP_t)_{t\ge 0}$ is contractive for the $\|.\|_{(p,q),\,\cN}$ norm. By definition,
		\begin{align*}
			\|\cP_t(X)\|_{(p,q),\,\cN}&= \sup_{\substack{A,B\in\cN(\cP),\\\|A\|_{2r,\sigma_{\Tr}},\|B\|_{2r,\sigma_{\Tr}}=1}  }\|\cP_t(A)\cP_t(X)\cP_t(B)\|_{q,\sigma_\tr}\\
			&= \sup_{\substack{A,B\in\cN(\cP),\\\|A\|_{2r,\sigma_{\Tr}},\|B\|_{2r,\sigma_{\Tr}}=1  }}\|\cP_t(AXB)\|_{q,\sigma_\tr}\\
			&\le \sup_{\substack{A,B\in\cN(\cP),\\\|A\|_{2r,\sigma_{\Tr}},\|B\|_{2r,\sigma_{\Tr}}=1  }}\|AXB\|_{q,\sigma_\tr}= \|X\|_{(p,q),\,\cN}\,,
		\end{align*}
		where $\frac{1}{r}=\frac{1}{q}-\frac{1}{p}$. Here the first line follows from the fact that $(\Pcal_t)_{t\geq0}$ acts unitarily on $\DF$, the second line follows from Proposition 1(2) of \cite{[DFSU14]}, and the third one from the contractivity of $\cP_t$ as a map from $\mathbb{L}_q(\sigma_\tr)$ to $\mathbb{L}_q(\sigma_\tr)$. The case of $\|.\|_{(q,p),\,\cN}$ follows by duality (\Cref{prop_duality}(ii)) and H\"{o}lder's inequality (Proposition \reff{prop_duality}(i)):
		\begin{align*}
			\|\cP_t(X)\|_{(q,p),\,\cN}=\sup_{\|Y\|_{(q',p'),\cN}\le 1}\langle Y,\cP_t(X)\rangle_{\sigma_\tr}&=\sup_{\|Y\|_{(q',p'),\cN}\le 1}\langle \hat\cP_t(Y),X\rangle_{\sigma_\tr}\\
			&\le\sup_{\|Y\|_{(q',p'),\cN}\le 1} \|\hat\cP_t(Y)\|_{(q',p'),\cN}\|X\|_{(q,p),\,\cN}\,,
		\end{align*}
		where $\frac{1}{p}+\frac{1}{p'}=1$ and $\frac{1}{q}+\frac{1}{q'}=1$. We conclude by using the above proof of $\operatorname{DF}$-contractivity for $1\le p'\le q'\le +\infty$, applied to the QMS $(\hat\Pcal_t)_{t\geq0}$.
\item[(ii)] This is point (vii) of \Cref{prop_duality} for $\cN\equiv \cN(\cP)$. 
\item[(iii)] This is point (iv) of \Cref{prop_duality} for $\cN\equiv \cN(\cP)$.
\item[(iv)] This is obvious since $\cN(\cP)\equiv \CC\mathbb{I}$.
\end{enumerate}
\end{proof}

\subsection{Differentiation of the decoherence-free norms}\label{subsec2}
As in the primitive case, the equivalence between hypercontractivity and the log-Sobolev inequality relies on a formula for the differentiation of the decoherence-free norms, commonly called Gross' integration Lemma. In the bipartite case, where $\cN(\cP)=\cB(\cH_A)\otimes I_{\cH_B}$ and the invariant state is the maximally mixed state, this differentiation was done in \cite{[BK16]}. Here we generalise this result to the case of the amalgamated $\mathbb{L}_p$ norms associated to a decohering QMS. The next lemma, which extends Lemma 5 of \cite{[KT13]}, provides a physical interpretation of the $\operatorname{DF}$-$\mathbb{L}_p$ relative entropies in terms of the \emph{quantum relative entropy} of a state and its projection onto the decoherence-free subalgebra. Recall that the quantum relative entropy $D(\rho\|\sigma)$ of two states $\rho,\sigma\in\Dcal(\Hcal)$ is given by
\begin{align}
&D(\rho\|\sigma):=\left\{
\begin{aligned}
&\tr\,(\rho\,(\ln\rho-\ln\sigma))~~~~~~~~\supp(\rho)\subset\supp(\sigma)\,,\\
&+\infty~~~~~~~~~~~~~~~~~~~~~~~~~~~~~~\text{otherwise}.
\end{aligned}
\right.\end{align}

\begin{lemma}\label{ent2p}
	Let $\rho\in\cD_+(\cH)$ and $X\in\cB(\cH)$ positive definite, then
	\begin{itemize}
		\item[(i)] $\operatorname{Ent}_{2,\,\cN}(\Gamma_{\sigma_\tr}^{-1/2}(\sqrt{\rho}))=\frac{1}{2}D(\rho\|\rho_\cN).$
		\item[(ii)] $\operatorname{Ent}_{1,\cN}(\Gamma_{\sigma_\tr}^{-1}({\rho}))=D(\rho\|\rho_\cN)$.
		\item[(iii)] More generally, $\operatorname{Ent}_{q,\,\cN}(\Gamma_{\sigma_\tr}^{-\frac{1}{q}}(\rho^{\frac{1}{q}}))=\frac{1}{q}D(\rho\|\rho_\cN)$ for any $q\ge 1$.
		\item[(iv)] If $X\in\cN(\cP)$ and any $q\ge 1$, $\operatorname{Ent}_{q,\,\cN}(X)=0$.
		\item[(v)] $\operatorname{Ent}_{p,\,\cN}(X)=\frac{2}{p}\operatorname{Ent}_{2,\,\cN}(I_{2,p}(X))$ for any $p\ge 1$.
	\end{itemize}	
\end{lemma}	

\begin{proof}
	\begin{itemize}
		\item[(i)] For $X=\Gamma_{\sigma_\tr}^{-1/2}(\sqrt{\rho})$, 
	\Cref{ent} reduces to
	\begin{align}\label{eq15a}
		\operatorname{Ent}_{2,\,\cN}(X)= \frac{1}{2}\tr(\rho\ln\rho)-\frac{1}{2}\tr\left(\rho\ln E_\cN\left[ \Gamma_{\sigma_\tr}^{-1}(\rho)\right]\right)-\frac{1}{2}\tr(\rho\ln\sigma_\tr)\,.
	\end{align}
	Now, $\tr\left(\rho\ln E_\cN\left[ \Gamma_{\sigma_\tr}^{-1}(\rho)\right]\right)=\tr\left(\rho_{\cN}\ln E_\cN\left[ \Gamma_{\sigma_\tr}^{-1}(\rho)\right]\right)
	$. Using \Cref{eq_com_cond_expt} together with $[\sigma_{\Tr},E_{\Ncal_*}(\rho)]=0$, we arrive at
	\[\tr\left(\rho_{\cN}\ln E_\cN\left[ \Gamma_{\sigma_\tr}^{-1}(\rho)\right]\right)=D(\rho_\cN\|\sigma_\tr)\,.\]
	Substituting the above right hand side into \reff{eq15a}, we finally arrive at (cf. \cite{BarEID17})
	\begin{align}
		\operatorname{Ent}_{2,\,\cN}(Y)=\frac{1}{2}D(\rho\|\sigma_\tr)-\frac{1}{2}D(\rho_\cN\|\sigma_\tr)=\frac{1}{2}D(\rho\|\rho_\cN)\,.
	\end{align}	
\item[(ii)] It is easy to verify that for $X=\Gamma_{\sigma_\tr}^{-1}(\rho)$:
\begin{align*}
	\langle I_{\infty,1}(X),S_1(X)\rangle_{\sigma_\tr}-\|X\|_{1,\sigma_\tr}\ln\|X\|_{1,\sigma_\tr}=D(\rho\|\sigma_\tr)\,.
	\end{align*}
Moreover, we proved in (i) that $\tr\left(\rho_{\cN}\ln E_\cN\left[ \Gamma_{\sigma_\tr}^{-1}(\rho)\right]\right)=D(\rho_\cN\|\sigma_\tr)$. We conclude from inserting the last two equations into the expression of $\operatorname{Ent}_{1,\cN}(\Gamma^{-1}_{\sigma_\tr}(\rho))$ and using once again that $D(\rho\|\rho_\cN)=D(\rho\|\sigma_\tr)-D(\rho_\cN\|\sigma_\tr)$.
\item[(iii)] Follows similarly.
\item[(iv)] This is a simple consequence of (iii) together with the fact that if $X\in\cN(\cP)$, 
\begin{align*}
\rho_\cN&:=E_{\cN*}[\Gamma_{\sigma_\tr}^{\frac{1}{q}}(X)^q]= E_{\cN*}[\Gamma_{\sigma_\tr}(X^q)]=\Gamma_{\sigma_\tr}(E_\cN[X^q])=\Gamma_{\sigma_\tr}(X^q)=(\Gamma_{\sigma_\tr}^{\frac{1}{q}}(X))^q\equiv\rho\,.
\end{align*}
\item[(v)] follows by direct computation.
	\end{itemize}
\end{proof}	
The proof of next theorem follows closely the one of Theorem 7 of \cite{[BK16]}, and is discussed in \Cref{diffnormapp} for sake of clarity. It can be seen as both a generalisation of the differentiation done in the primitive case in \cite{OZ99} to non-primitive QMS (see also Lemma 14 of \cite{[KT13]}), and the one carried out for the CB-norm in \cite{[BK16]} to the non-unital case.
	
\begin{theorem}\label{diffnorm}
	Let $t\mapsto p(t)$ be a twice continuously differentiable increasing function in a neighborhood of $0$, with $p(0)=q\ge 1$. Also let $t\mapsto Y(t)\in \cB(\cH)$ be an operator-valued twice continuously differentiable function, where $Y(t)$ is positive definite in a neighborhood of $0$, and define $Y:= Y(0)$. Then
	\begin{align*}
		\left.  \frac{d}{dt}\|Y(t)\|_{(q,p(t)),\,\cN}~\right|_{t=0}&=  \frac{p'(0)}{q\|Y\|^{q-1}_{q,\sigma_\tr}}\left(	\operatorname{Ent}_{q,\,\cN}(Y)+\frac{q}{p'(0)}\tr \left(\left[\Gamma_{\sigma_\tr}^{\frac{1}{q}}( Y)\right]^{q-1} \Gamma_{\sigma_\tr}^{\frac{1}{q}}(Y'(0))\right)\right)\,.
	\end{align*}
\end{theorem}	

We shall apply this theorem to different situations. Perhaps the most relevant one is when $Y(t)$ models the evolution of an observable $X\in\Bcal(\Hcal)$ under the QMS $(\cP_t)_{t\ge 0}$. We state it as a corollary.

\begin{corollary}\label{cor}
For any positive definite $X\in\Bcal(\Hcal)$,
\begin{align}\label{diff1}
			\left.  \frac{d}{dt}\|\cP_t(X)\|_{(q,p(t)),\,\cN}~\right|_{t=0}&=  \frac{ p'(0)}{q\|X\|^{q-1}_{q,\sigma_\tr}}\left(	\operatorname{Ent}_{q,\,\cN}(X)-\frac{2(q-1)}{p'(0)}\cE_{q,\,\cL}(X)\right)\,.
\end{align}	
\end{corollary}

\begin{remark}
The situation where $X(t)\equiv X$ for all $t$ and $p(t)=q+t$ provides an functional analytic justification of the term entropy, as it yields:
\begin{align}\label{diff2}
			\left.  \frac{d}{dp}\|X\|_{(q,p),\,\cN}~\right|_{p=q}&=  \frac{ 1}{q\|X\|^{q-1}_{q,\sigma_\tr}}\,\operatorname{Ent}_{q,\,\cN}(X)\,.
\end{align}
We see here the tight relationship between the amalgamated $\Lbb_p$ norms and entropic quantities that appear in quantum information theory. This link was recently exploited in \cite{gao2017strong} to prove a generalisation of the celebrated SSA inequality.
\end{remark}

\subsection{Almost uniform convexity}\label{unif}
In this subsection, we study an analogue of the well-known uniform convexity of the Schatten norms proved in \cite{Ball1994}. This analogue was proved in the context of weighted $\mathbb{L}_p(\sigma)$ norms in \cite{OZ99}. This will be an essential tool when proving universal lower bounds on the weak DF-log-Sobolev constants. This inequality states that for all $X$ positive semidefinite, any full-rank state $\sigma$, and all $p\in [1,2]$,
\begin{equation}\label{eq_unif_conv}
\norm{X}_{p,\sigma}^2\geq(p-1)\norm{X-\tr(\sigma X)}_{p,\sigma}^2+\Tr(\sigma\,X)^2\,.
\end{equation}
For the Shattern norms, this inequality can be seen as a consequence of Clarkson inequalities (see \cite{pisier2003non} for a discussion of this fact). It has many important applications in the theory of non-commutative $\Lbb_p$ spaces, such as yielding the optimal constant for Fermionic hypercontractivity~\cite{CL93}. We shall prove however in \Cref{sec5bis} that this inequality fails for the amalgamated $\Lbb_p$ spaces. Instead, in this section we prove a weak form of this inequality.\\

For $X\in\Bcal(\Hcal)$, $A\in\cN(\cP)\cap\cS^+_{\mathbb{L}_1(\sigma_\tr)}$ and $p\geq1$, we define
\begin{equation}\label{eq_almost_norm}
\Phi(X,A,p):=\norm{\Gamma_A^{-1/r}(X)}_{p}
=\Tr\left[\left|\,A^{-\frac{1}{2r}}\,X\,A^{-\frac{1}{2r}}\,\right|^p\right]^{\frac{1}{p}}\,,
\end{equation}
where we recall that $ 1/r=\left| 1/2- 1/p\right|$. Remark that for all positive semidefinite $X\in\Bcal(\Hcal)$ and all $A\in\cN(\cP)\cap \cS^+_{\mathbb{L}_1(\sigma_\tr)}$, $\Phi(\Gamma_{\sigma_\tr}^{\frac{1}{2}}(X),A,2)=\norm{X}_{2,\sigma_{\Tr}}\,$.
We shall prove that a similar result as \reff{eq_unif_conv} holds for $\Phi$, which we subsequently refer to as \emph{almost uniform convexity}.
\begin{lemma}\label{lem_unif_conv}
The two following properties hold:
\begin{enumerate}
\item[(i)] For all $X\in\cB_{sa}^+(\cH)$, $A\in\cN(\cP)\cap \cS^+_{\mathbb{L}_1(\sigma_\tr)}$ and all $1\leq p\leq 2$,
\begin{equation}\label{eq_lem_unif_conv1}
\Phi(\Gamma_{\sigma_\tr}^{\frac{1}{p}}(X),A,p)^2\ge (p-1)\Phi(\Gamma_{\sigma_\tr}^{\frac{1}{p}}(X-E_\cN[X]),A,p)^2+\Phi(\Gamma_{\sigma_\tr}^{\frac{1}{p}}(E_\cN[X]),A,p)^2\,.
\end{equation}
\item[(ii)] For all $X\in\cB_{sa}^+(\Hcal)$ and $A\in\cN(\cP)\cap \cS^+_{\mathbb{L}_1(\sigma_\tr)}$, 
\begin{equation}\label{eq_lem_unif_conv2}
\begin{aligned}
\left.\frac{\partial}{\partial p}\Phi(\Gamma_{\sigma_\tr}^{\frac{1}{p}}(X),A,p)^2\right|_{p=2}
\leq & \left.\frac{\partial}{\partial p}\Phi(\Gamma_{\sigma_\tr}^{\frac{1}{p}}(X-E_\Ncal[X]),A,p)^2\right|_{p=2}
+\left.\frac{\partial}{\partial p}\Phi(\Gamma_{\sigma_\tr}^{\frac{1}{p}}(E_\Ncal[X]),A,p)^2\right|_{p=2}   \\
 & +\norm{X-E_\Ncal[X]}_{2,\sigma_{\Tr}}^2\,. 
\end{aligned}
\end{equation}
\end{enumerate}
\end{lemma}

\begin{proof}
We follow the proof of Lemma 2.9 in \cite{OZ99} in order to prove the first claim. We adopt the following notations. For $0\leq t\leq 1$, define
\begin{align*}
& X(t)=E_\Ncal[X]+t\left(X-E_\Ncal[X]\right)\,,\\
& \varphi(t)=\Phi(\Gamma_{\sigma_\tr}^{\frac{1}{p}}(X(t)),A,p)^2\,,\\
& h=\Gamma_{\sigma_{\Tr}}^{\frac 1p}\circ\Gamma_A^{-\frac 1r}(X-E_\Ncal[X])\,.
\end{align*}
Then, \Cref{eq_lem_unif_conv1} reduces to:
\begin{align}\label{eq32}
\varphi(1)\geq(p-1)\|h\|^{2}_p+\varphi(0)\,.
\end{align}
This inequality follows directly from:
\begin{enumerate}
\item $\varphi'(0)=0$ ;
\item $\varphi''(t)\geq2(p-1)\norm{h}^2_{p}$ for all $0\leq t\leq 1$.
\end{enumerate}
We start by computing $\varphi'(t)$. Writting $Z(t)=\Gamma_{\sigma_{\Tr}}^{\frac 1p}\circ\Gamma_A^{-\frac 1r}(X(t))$, we have by integral representation that for all $0\leq t\leq 1$
\begin{align}
& \varphi'(t)=2\,\Tr[h\,Z(t)^{p-1}]\,\Tr[Z(t)^p]^{2/p-1}\,, \label{eq_proof_lem_unif_conv1}\\
& \varphi''(t)\geq 2\, \frac{\partial}{\partial t}\,\left(\Tr[h\,Z(t)^{p-1}]\right)\,\Tr[Z(t)^p]^{2/p-1}\,.  \label{eq_proof_lem_unif_conv2}
\end{align}
We start by proving claim 1. First remark that, since elements of $\DF$ commute with $\sigma_{\Tr}$,
\[Z(0)^{p-1}=\sigma_{\Tr}^{\frac12-\frac1{2p}}\,E_\Ncal\left[\Gamma_A^{-\frac 1r}(X)\right]^{p-1}\sigma_{\Tr}^{\frac12-\frac1{2p}}\,.\]
Therefore $\Tr[h\,Z(0)^{p-1}]=\sca{X-E_\Ncal[X]}{B}_{\sigma_{\Tr}}$ where $B=\Gamma_A^{-\frac 1r}\left(E_\Ncal\left[\Gamma_A^{-\frac 1r}(X)\right]^{p-1}\right)\in\DF$. By \Cref{En} we get that $\Tr[h\,Z^{p-1}]=0$ which results in $\varphi'(0)=0$. The proof of claim 2 is a direct copy of the proof of Lemma 2.9 in \cite{OZ99} and we omit it. Hence, \Cref{eq_lem_unif_conv1} holds.

In order to prove (ii), we rearrange the terms in \Cref{eq_lem_unif_conv1} to get
\begin{align*}
(2-p)\Phi(\Gamma_{\sigma_\tr}^{\frac{1}{p}}&(X-E_\Ncal[X]),A,p)^2\,\\
&\ge\left(\Phi(\Gamma_{\sigma_\tr}^{\frac{1}{2}}(X),A,2)+\Phi(\Gamma_{\sigma_\tr}^{\frac{1}{p}}(X),A,p)\right)\left(\Phi(\Gamma_{\sigma_\tr}^{\frac{1}{2}}(X),A,2)-\Phi(\Gamma_{\sigma_\tr}^{\frac{1}{p}}(X),A,p)\right) \\
&\qquad -\left(\Phi(\Gamma_{\sigma_\tr}^{\frac{1}{2}}(X-E_\Ncal[X]),A,2)+\Phi(\Gamma_{\sigma_\tr}^{\frac{1}{p}}(X-E_\Ncal[X]),A,p)\right)\,\times\\
&~~~~~~~~~\times\,\left(\Phi(\Gamma_{\sigma_\tr}^{\frac{1}{2}}(X-E_\Ncal[X]),A,2)-\Phi(\Gamma_{\sigma_\tr}^{\frac{1}{p}}(X-E_\Ncal[X]),A,p)\right) \\
&\qquad -\left(\Phi(\Gamma_{\sigma_\tr}^{\frac{1}{2}}(E_\Ncal[X]),A,2)+\Phi(\Gamma_{\sigma_\tr}^{\frac{1}{p}}(E_\Ncal[X]),A,p)\right)\,\times\\
&~~~~~~~~~\times\,\left(\Phi(\Gamma_{\sigma_\tr}^{\frac{1}{2}}(E_\Ncal[X]),A,2)-\Phi(\Gamma_{\sigma_\tr}^{\frac{1}{p}}(E_\Ncal[X]),A,p)\right)\,,
\end{align*}
where we used that
\[\Phi(\Gamma_{\sigma_\tr}^{\frac{1}{2}}(X),A,2)^2=\Phi(\Gamma_{\sigma_\tr}^{\frac{1}{2}}(X-E_\Ncal[X]),A,2)^2+\Phi(\Gamma_{\sigma_\tr}^{\frac{1}{2}}(E_\Ncal[X]),A,2)^2\,.\]
\reff{eq_lem_unif_conv2} follows by dividing this inequality by $2-p$ and taking the limit $p\to2$.

\end{proof}

\section{$\operatorname{DF}$-hypercontractivity and the log-Sobolev inequality}\label{logSob}

In this section we state and prove the main results of this article. In \Cref{sect41}, we prove the equivalence between hypercontractivity for the amalgamated norms and the DF-log-Sobolev inequality. In \Cref{logsobuniversal}, we prove that the weak constants in the DF-log-Sobolev inequality can always be upper bounded by a universal constant, namely $\ln \sqrt2$. In \Cref{univsec}, we show how to derive estimates on the log-Sobolev constants using interpolation techniques. Finally, we combine these two last results in order to obtain generic bounds on both constants.

\subsection{Fundamental equivalence between hypercontractivity and the log-Sobolev inequality}\label{sect41}
Here, we state and prove the main result of this section, that is, the equivalence between the DF-log-Sobolev inequality and $\operatorname{DF}$-hypercontractivity. 

\begin{theorem}\label{gross1}Let $(\mathcal{P}_t)_{t\ge 0}$ be a decohering QMS on $\cB(\cH)$ with associated generator $\LL$, and let $q\ge 1$, $d\ge0$ and $p(t)=1+(q-1)\e^{2t/c}$ for some constant $c>0$. Then
	\begin{enumerate}
		\item[(i)] If $\operatorname{HC}_{q,\,\cN}(c,d)$ holds, then $\operatorname{LSI}_{q,\,\cN}(c,d)$ holds.
		\item[(ii)] If $\operatorname{LSI}_{p(t),\,\cN}(c,d)$ holds for all $t\ge 0$, then $\operatorname{HC}_{q,\,\cN}(c,d+\ln|I|)$ holds, where $|I|$ denotes the number of blocks in the decomposition of $\DF$ as given in \Cref{eqtheostructlind2}.
	\end{enumerate}
\end{theorem}

\begin{remark}
For primitive evolution, $|I|=1$ and this theorem states the equivalence between hypercontractivity and the logarithmic Sobolev inequality. The equivalence is also achieved in the more general situation where $\DF$ is a factor, that is, in the situation of \Cref{sect26}. However, we will discuss in \Cref{normestimate} why the term $\ln |I|$ may not be optimal.
\end{remark}

\begin{proof}
	We first prove (i). For $X>0$, define the function 
	\begin{align*}
	F:[0,+\infty)\ni t\mapsto \exp\left\{   -2d \left(   \frac{1}{q}-\frac{1}{p(t)}\right) \right\}\|\cP_t(X)\|_{(q,p(t)),\,\cN}\,,
	\end{align*}
where $p(t):=1+(q-1)\e^{2t/c}$. $\operatorname{HC}_{q,\,\cN}(c,d)$ implies that $\ln F(t)\le \ln F(0)$ for all $t\ge 0$, with equality at $t=0$. Therefore, 
\begin{align*}
	\left.\frac{d\ln F(t)}{dt}\right|_{t=0^+}=-2d~\frac{{p}'(0)}{q^2}+\left.\frac{d}{dt}\ln\|\cP_t(X)\|_{(q,p(t)),\,\cN}\right|_{t=0^+}\le 0\,.
	\end{align*}
Using \Cref{diff1}, the above inequality reduces to
\begin{align*}
\frac{-2d}{q}+\frac{1}{\|X\|_{q,\sigma_\tr}^q}\left( \operatorname{Ent}_{q,\,\cN}(X)-c\, \cE_{q,\,\cL}(X)\right)\le 0\,,
\end{align*}
which yields $\operatorname{LSI}_{q,\,\cN}(c,d)$.

To prove (ii), we proceed by contradiction, similarly to \cite{[BK16]}. The main difference resides in the replacement of the norm by an auxillary quantity that allows to control a remainder term that does not appear in the case where $\DF$ is a factor. 
Assume that there exists an $X\in\Bcal(\Hcal)$ such that hypercontractivity fails for this $X$. Following the same proof as Theorem 12 of \cite{[DJKR16]}, we can show that it is sufficient to consider that $X$ is a positive definite operator. Indeed, for fixed $q\le p$, if there exists $C>0$ such that for any $X$ definite,
\begin{align*}
        \|\cP_t(X)\|_{(q,p),\,\cN}\le C\|X\|_{q,\sigma_\tr}\,,
\end{align*}
then the inequality remains true for any $X\in\cB(\cH)$. Without loss of generality, we also assume that $\|X\|_{q,\sigma_\tr}=1$. Then, suppose that there exists some time $t_0>0$ such that
\begin{align*}
\|\cP_{t_0}(X)\|_{(q,q(t_0)),\,\cN}>\exp\left\{ 2(d+\ln |I|)\left( \frac{1}{q}-\frac{1}{q(t_0)}\right)\right\}\,.
	\end{align*}
Define, for $\eps>0$, 
\begin{align*}
\tilde{\varphi}(t):=\vertiii{\cP_t(X)}_{(q,p(t)),\,\cN}\exp\left\{ -2(d+\ln|I|) \left( \frac{1}{q}-\frac{1}{p(t)}\right)\right\}-\eps t\,,
\end{align*}
where $\vertiii{.}_{(q,p),\,\cN}$ is defined in \Cref{miniproof1}. By definition, $\vertiii{\cP_{t_0}(X)}_{(q,q(t_0)),\,\cN}\ge	\|\cP_{t_0}(X)\|_{(q,q(t_0)),\,\cN}$ so that $\tilde{\varphi}(t_0)>1$ for $\eps$ small enough. Define the set $U:=\{ t\in [0,t_0]:~\tilde{\varphi}(t)\le 1\}$. Since $\cP_0=\id$ and $q(0)=q$, we have $\tilde{\varphi}(0)=\|X\|_{q,\sigma_\tr}=1$, so that $U\ne \emptyset$. Let $u$ be the supremum of the set $U$. By continuity of $t\mapsto \vertiii{\cP_t(X)}_{(q,p(t)),\,\cN}$ (cf. \Cref{mini2}), $\tilde{\varphi}$ is continuous and therefore $u\in U$ and $u<t_0$. Now, by definition of $u$, for all $t\in (u,t_0]$, $\tilde{\varphi}(t)>1= \tilde{\varphi}(u)$. For $t>0$, let $\tilde{A}(t)$ be the unique minimiser of
 \begin{align*}
 \cN(\cP)\cap \tilde{\mathcal{S}}_{\mathbb{L}_1(\sigma_\tr)}^+\ni	A\mapsto \|A^{-s(t)/2}\cP_t(X)A^{-s(t)/2}\|_{p(t),\sigma_\tr}\,,
 	\end{align*}
as characterised in \Cref{mini1}, where $\tilde{\mathcal{S}}_{\mathbb{L}_1(\sigma_\tr)}^+$ is defined in \Cref{miniproof1} and $s(t)=\frac{1}{q}-\frac{1}{p(t)}$. Define 
 \begin{align*}
 	\mu(t):= \| \tilde{A}(u)^{-s(t)/2}\cP_t(X)\tilde{A}(u)^{-s(t)/2} \|_{p(t),\sigma_\tr}\exp\left\{-2(d+\ln|I|)\left( \frac{1}{q}-\frac{1}{p(t)}\right) \right\}-\eps t\,.
 	\end{align*}
 Therefore, for all $t\ge u$, 
\begin{align*}
	\mu(t)\ge \inf_{A\in \cN(\cP)\cap \tilde{\mathcal{S}}_{\mathbb{L}_1(\sigma_\tr)}^+}  \| A^{-s(t)/2}\cP_t(X)A^{-s(t)/2} \|_{p(t),\sigma_\tr}\exp\left\{-2(d+\ln|I|)\left( \frac{1}{q}-\frac{1}{p(t)}\right) \right\}-\eps t=\tilde{\varphi}(t)\,
	\end{align*}
and $\tilde{\varphi}(u)=\mu(u)$. Now, the derivative of $\mu(t)$ at $t=u$ can be computed using \Cref{diff} with $X(t)=\Gamma_{\sigma_{\tr}}^{1/p(t)}\circ\cP_t(X)$ and $A=\tilde{A}(u)$. Given $M(t):= \tilde{A}(u)^{-s(t)/2}\cP_t(X)\tilde{A}(u)^{-s(t)/2}$, one finds
	\begin{align*}
&\left.	\frac{\partial}{\partial t}\right|_{t=u}\|\tilde{A}(u)^{-s(t)/2}\cP_t(X)\tilde{A}(u)^{-s(t)/2}\|_{p(t),\,\sigma_\tr}\\
&=\frac{p'(u)}{p(u)^2  	\|M(u)\|_{p(t),\sigma_\tr}^{p(u)-1}}\left( -\tr\,\left( \Gamma_{\sigma_\tr}^{\frac{1}{p(u)}}(M(u))^{p(u)}\right)\ln \tr\left[ \Gamma_{\sigma_\tr}^{\frac{1}{p(u)}}(M(u))^{p(u)}\right]\right.\nonumber\\
&+\tr \left[\Gamma_{\sigma_\tr}^{\frac{1}{p(u)}}(M(u))^{p(u)}\ln \Gamma_{\sigma_\tr}^{\frac{1}{p(u)}}(M(u))^{p(u)}\right]
- \left.\tr\left[\Gamma_{\sigma_\tr}^{\frac{1}{p(u)}}(M(u))^{p(u)}\ln \tilde{A}(u)	 \right]\right.\nonumber\\
&\left.+\frac{p(u)^2}{p'(u)}  \tr \left[\Gamma_{\sigma_\tr}^{\frac{1}{p(u)}}(M(u))^{p(u)-1} A^{-s(u)/2}\left\{\Gamma_{\sigma}^{\frac{1}{p(u)}}\cL(\cP_u(X))      -\frac{p'(u)}{2p(u)} \{\ln\sigma_\tr,\,\Gamma_{\sigma_\tr}^{\frac{1}{p(u)}}(\cP_u(X))\}    \right\}A^{-s(u)/2}\right]\right).
\end{align*}
Defining $\rho(u):=\Gamma_{\sigma_\tr}^{\frac{1}{p(u)}}(M(u))^{p(u)}$, the above simplifies into
\begin{align}
	&\left.	\frac{\partial}{\partial\label{final} t}\right|_{t=u}\|\tilde{A}(u)^{-s(t)/2}\cP_t(X)\tilde{A}(u)^{-s(t)/2}\|_{p(t),\,\sigma_\tr}\\
&=\nonumber
\frac{p'(u)}{p(u)^2  	\|M(u)\|_{p(t),\sigma_\tr}^{p(u)-1}}
\Big(D(\rho(u)\|E_{\cN*}(\rho(u)))-c\,p(u)\,\cE_{p(u),\,\cL}(M(u))\\
&\nonumber~~~~-\tr(\rho(u))\ln\tr\rho(u)+\tr\left(\rho(u)\ln E_{\cN*}(\rho(u))\right)-\tr\left(\rho(u)\ln \tilde{A}(u)\right) + \tr(   \rho(u)    \,\ln\sigma_\tr  )
\Big)\\
&=\nonumber
\frac{p'(u)}{p(u)	\|M(u)\|_{p(t),\sigma_\tr}^{p(u)-1}}
\Big\{\operatorname{Ent}_{p(u),\,\cN}(M(u))-c\,\cE_{p(u),\,\cL}(M(u))\\
&~~~~~\nonumber+\frac{1}{p(u)}\left(-\tr(\rho(u))\ln\tr\rho(u)+\tr\left(\rho(u)\ln E_{\cN*}(\rho(u))\right)-\tr\left(\rho(u)\ln \tilde{A}(u)\right) + \tr\left(   \rho(u)    \,\ln\sigma_\tr  \right)
\right)\Big\}\,.
\end{align}
Using the expression for $\tilde{A}(u)$ derived in \Cref{optim}, which we recall here:
\begin{align*}
		P_i	\tilde{A}P_i=\frac{\,P_iE_\cN\left[I_{1,p}(\tilde{A}^{-1/2r}Y\tilde{A}^{-1/2r})\right]    P_i}{|I|\,\tr\left[   P_i\,\left(\Gamma_{\sigma_\tr}^{\frac{1}{p}}(\tilde{A}^{-1/2r}\,Y\,\tilde{A}^{-1/2r})\right)^{p} \,P_i\right]}\,,
\end{align*}	
we have:
\begin{align*}
	\ln \tilde{A}(u)&=\ln\sum_{i\in I}\,P_i\tilde{A}(u)\,P_i\\
	&=\sum_{i\in I} \ln\,P_i\,\tilde{A}(u)\,P_i\\
	&=\sum_{i\in I}-\ln|I|\,P_i-\ln\tr(P_i\,\rho(u)\,P_i)\,-\ln P_i\sigma_\tr P_i+\ln P_i\,E_{\cN*}(\rho(u))P_i\\
	&=-\ln |I|-\ln\sigma_\tr-\ln\tr(\rho(u))+\ln\,E_{\cN*}(\rho(u))\,.
\end{align*}
Using this expression, \Cref{final} simplifies into:
\begin{align*}
&	\left.	\frac{\partial}{\partial t}\right|_{t=u}\!\!\!\!\!\!\!\|\tilde{A}(u)^{-s(t)/2}\cP_t(X)\tilde{A}(u)^{-s(t)/2}\|_{p(t),\,\sigma_\tr}\!\!\\
&~~~~~~~~~~~~~~~~~~~~~~=\nonumber
\frac{p'(u)/p(u)}{	\|M(u)\|_{p(t),\sigma_\tr}^{p(u)-1}}
\Big\{\operatorname{Ent}_{p(u),\,\cN}(M(u))-c\,\cE_{p(u),\,\cL}(M(u))+\frac{\ln|I|}{p(u)}\,\|M(u)\|_{p(u),\sigma_\tr}^{p(u)}\Big\}\,.
\end{align*}	
 Then, using the assumption that $\operatorname{LSI}_{p(u),\,\cN}(c,d)$ holds, we find that
\begin{align*}
	\mu'(u)\le -\eps\,. 
\end{align*}	
Therefore, there exists $\delta>0$ such that $u+\delta \le t_0$ and $\mu(u+\delta)\le \mu(u)$. We then have
\begin{align*}
	\tilde{\varphi}(u+\delta)\le \mu(u+\delta)\le \mu(u)=\tilde{\varphi}(u)\le 1\,,
\end{align*}	
which is in contradiction with the very definition of $u$. 

\end{proof}	
	
In the above theorem, one needs $\operatorname{LSI}_{\tilde{q},\,\cN}(c,d)$ to hold for any $\tilde{q}\ge q$ in order to conclude that $\operatorname{HC}_{q,\,\cN}(c,d)$ holds. Under the assumption of regularity of the Dirichlet forms, it is enough to assume that it holds for $q=2$ only.
\begin{theorem}
	Assume that $\operatorname{LSI}_{2,\,\cN}(c,d)$ holds. Then
	\begin{itemize}
		\item[(i)] If the generator $\LL$ is strongly $\Lbb_p$-regular for some $d_0\ge 0$, then $\operatorname{LSI}_{q,\,\cN}(c,d+c\,d_0)$ holds for all $q\geq1$, so that $\operatorname{HC}_{2,\,\cN}(c,d+\ln |I|+c\,d_0)$ holds.
		\item[(ii)] If the generator $\LL$ is only weakly $\Lbb_p$-regular for some $d_0\ge 0$, then $\operatorname{LSI}_{q,\,\cN}({2c}, d+c\,d_0)$ holds for all $q\geq1$, so that $\operatorname{HC}_{2,\,\cN}(2c,d+\ln|I|+c\,d_0)$ holds. 
	\end{itemize}	
\end{theorem}

\begin{proof}
		\begin{itemize}
		\item[(i)] From \Cref{ent2p}(vi),
\begin{align*}
	\operatorname{Ent}_{q,\,\cN}(X)&=\frac{2}{q}\operatorname{Ent}_{2,\,\cN}(I_{2,q}(X))\\
		&\le \frac{2}{q}~\left(c\,\cE_{2,\,\Lcal}(I_{2,q}(X))+ d\,\|I_{2,q}(X)\|_{2,\sigma_\tr}^2\right)\\
	&\le c\,\cE_{q,\,\cL}(X)+\frac{2}{q}\left(d+c\,d_0\right)\|X\|_{q,\sigma_\tr}^q\,,
	\end{align*}
where in the last line we used that $\cE_{2,\,\Lcal}(I_{2,q}(X))\le \frac{q}{2}\cE_{q,\,\cL}(X)+d_0\|X\|_{q,\sigma_\tr}^q$ by strong $\Lbb_p$-regularity.
\item[(ii)] Follows similarly.
\end{itemize}	
\end{proof}

It was shown in \cite{BarEID17} that any generator satisfying $\sigma_\tr$-DBC is strongly regular with constant $d_0=0$. Furthermore, in the case when \Cref{DBC} is satisfied, a straightforward extension of the proof of Proposition 5.2 of \cite{OZ99} to the case of a non-primitive QMS implies that the strong $\mathbb{L}_p$-regularity of $\LL$ always holds, with $d_0=\|\LL\|_{2\to 2,\,\sigma_\tr}+1$. The following corollary is a straightforward consequence of these two facts.

\begin{corollary}\label{cor4.4}Assume that $\operatorname{LSI}_{2,\,\cN}(c,d)$ holds. Then:
	\begin{itemize}
		\item[(i)]  If $\cL$ is reversible, then $\operatorname{HC}_{2,\,\cN}(c,d+\ln|I|+c\,(\|\LL\|_{2\to2,\,\sigma_\tr}+1))$ holds. 
		\item[(ii)] If $\LL$ satisfies $\sigma_\tr$-$\operatorname{DBC}$, then $\operatorname{HC}_{2,\,\cN}(c,d+\ln|I|)$ holds.
	\end{itemize}
\end{corollary}	

\subsection{A universal upper bound on the weak log-Sobolev constant}\label{logsobuniversal}
Here and in the next section, we show how to get a $\operatorname{DF}$ log-Sobolev inequality with universal constants in terms of the spectral gap of the QMS. Recall that the spectral gap $\lambda(\Lcal)$ is defined as the largest constant $\lambda>0$ such that the following \textit{$\operatorname{DF}$-Poincar\'{e} inequality} holds: for all $X\in\cB_{sa}(\cH)$:
\begin{align}\tag{$\operatorname{PI}(\lambda)$}\label{PI}
	\lambda\,\operatorname{Var}_{\cN}(X)\le\cE_{2,\Lcal}(X)\,.
\end{align}	
where $\operatorname{Var}_\cN(X):= \| X-E_\cN[X]\|_{2,\sigma_\tr}^2$ is the \textit{$\operatorname{DF}$-variance} of $X$. The first step is to prove that the weak log-Sobolev inequality together with Poincaré's inequality imply a universal weak log-Sobolev constant.
\begin{theorem}\label{thm4.5}
	Assume that $\operatorname{LSI}_{2,\Ncal}(c,d)$ holds and denote by $\lambda(\Lcal)$ the spectral gap of $\Lcal$. Then $\operatorname{LSI}_{2,\Ncal}(c+\frac{d+1}{\lambda(\Lcal)},d'=\ln \sqrt{2})$ holds.
\end{theorem} 
One can obtain from the uniform convexity \reff{eq_unif_conv} the following inequality
\[\operatorname{Ent}_{2,\,\sigma}(X)\le\operatorname{Ent}_{2,\,\sigma}\left(|X-\Tr[\sigma\,X]|_2\right)+\operatorname{Var}_{\sigma}(X)\,,\]
where $\operatorname{Var}_{\sigma}(X)=\|X-\tr(\sigma X)\|_{2,\sigma}^2$ is the variance of $X\in\cB_{sa}(\cH)$ under the state $\sigma$, and for any $Z\in\cB_{sa}(\cH)$, $$|Z|_2:=\Gamma_{\sigma}^{-\frac{1}{2}}|\Gamma_{\sigma}^{\frac{1}{2}}(Z)|\,.$$ From this we can derive the analogue result of \Cref{thm4.5} in the primitive case (see Theorem 4.2 of \cite{OZ99}). The extension of this result is the subject of the next proposition. 
\begin{proposition}\label{prop_domination_entropy}
For all $X\in\cB_{sa}^+(\cH)$,
	\begin{equation}\label{eq_prop_domination_entropy}
		\operatorname{Ent}_{2,\Ncal}(X)\leq\operatorname{Ent}_{2,\Ncal}(|X-E_\Ncal[X]|_2)+\operatorname{Var}_\Ncal(X)+\ln\sqrt2\,\norm{X}_{2,\sigma_{\Tr}}^2\,.
	\end{equation}
\end{proposition}
\begin{proof}
	We shall adopt the notations introduced in \Cref{unif} and write for $Z\in\cB_{sa}(\cH)$:
	\[ Z_\Ncal= \frac{E_\Ncal\left[I_{1,2}(Z)\right]}{\norm{Z}^2_{2,\sigma_{\Tr}}}\,.  \]
Using \Cref{eq300} with $q=2$ as well as \Cref{lem_GG} and \Cref{ent}, we find that
	\[\left.\frac{\partial}{\partial p}\Phi(\Gamma_{\sigma_\tr}^{\frac{1}{p}}(|Z|_2),A,p)^2\right|_{p=2}=\operatorname{Ent}_{2,\Ncal}(|Z|_2)+\frac{1}{2}\norm{Z}_{2,\sigma_{\Tr}}^2\,\Dent{\Gamma_{\sigma_\tr}(Z_\cN)}{\Gamma_{\sigma_\tr}(A)}\,,\]
	where $\Phi$ is defined in \Cref{eq_almost_norm}, and where we used $I_{1,2}(Z)=I_{1,2}(|Z|_2)$ and $\|\,|Z|_2\|_{2,\,\sigma}=\|Z\|_{2,\,\sigma}$. Consequently, by \Cref{eq_lem_unif_conv2} we get that for all $A\in \cN(\cP)\cap \cS^+_{\mathbb{L}_1(\sigma_\tr)}$ and for $p=\frac{\norm{E_\Ncal[X]}_{2,\sigma_{\Tr}}^2}{\norm{X}_{2,\sigma_{\Tr}}^2}$,  
	\begin{align}\label{refffff}
		\operatorname{Ent}_{2,\Ncal}(X)
		& \leq \operatorname{Ent}_{2,\Ncal}(|X-E_\Ncal[X]|_2)+\operatorname{Var}_\Ncal(X) \nonumber\\
		& \qquad +\frac{1}{2}\norm{X-E_\Ncal[X]}_{2,\sigma_{\Tr}}^2\,\Dent{\sigma_{\Tr}^{\frac12}\,(X-E_\Ncal[X])_\Ncal\,\sigma_{\Tr}^{\frac12}}{\sigma_{\Tr}^{\frac12}\,A\,\sigma_{\Tr}^{\frac12}} \nonumber \\
		& \qquad +\frac{1}{2}\norm{E_\Ncal[X]}_{2,\sigma_{\Tr}}^2\,\Dent{\sigma_{\Tr}^{\frac12}\,(E_\Ncal[X])_{\cN}\,\sigma_{\Tr}^{\frac12}}{\sigma_{\Tr}^{\frac12}\,A\,\sigma_{\Tr}^{\frac12}} \nonumber \\
		& \qquad -\frac{1}{2}\norm{X}_{2,\sigma_{\Tr}}^2\,\Dent{\sigma_{\Tr}^{\frac12}\,X_\Ncal\,\sigma_{\Tr}^{\frac12}}{\sigma_{\Tr}^{\frac12}\,A\,\sigma_{\Tr}^{\frac12}}  \nonumber \\
		& = \operatorname{Ent}_{2,\Ncal}(|X-E_\Ncal[X]|_2)+\operatorname{Var}_\Ncal(X) \nonumber\\
		& \qquad +\frac{1}{2}\norm{X}_{2,\sigma_{\Tr}}^2\,\times \nonumber\\
		& \qquad \times\,\Big\{\,p\,\Dent{\sigma_{\Tr}^{\frac12}\,(E_\Ncal[X])_\cN\,\sigma_{\Tr}^{\frac12}}{\sigma_{\Tr}^{\frac12}\,A\,\sigma_{\Tr}^{\frac12}}
		+(1-p)\,\Dent{\sigma_{\Tr}^{\frac12}\,(X-E_\Ncal[X])_\Ncal\,\sigma_{\Tr}^{\frac12}}{\sigma_{\Tr}^{\frac12}\,A\,\sigma_{\Tr}^{\frac12}} \nonumber\\
		& \qquad\qquad-\Dent{\sigma_{\Tr}^{\frac12}\,X_\Ncal\,\sigma_{\Tr}^{\frac12}}{\sigma_{\Tr}^{\frac12}\,A\,\sigma_{\Tr}^{\frac12}}\,,
		\Big\}
	\end{align}
since by definition $\operatorname{Var}_\Ncal(X)=\norm{X-E_\Ncal[X]}_{2,\Ncal}^2$ and $\operatorname{Ent}_{2,\,\cN}(E_\cN[X])=0$, and where we also used the fact that for any $Z\in\cB_{sa}(\cH)$, $Z_{\cN}=(|Z|_2)_\cN$. Remark that, since $(E_\cN[X])_\cN=E_\cN[X]^2/\|E_\cN[X]\|_{2,\sigma_\tr}^2$,
\begin{align}
	X_\Ncal&= \frac{E_\cN\left[\sigma_\tr^{-1/2}\,|\,\sigma_\tr^{1/4} X\sigma_\tr^{1/4}\,|^2\,\sigma_{\tr}^{-1/2} \right]}{\|X\|_{2,\sigma_\tr}^2} \nonumber\\
	 &=\frac{E_\cN\left[\sigma_\tr^{-1/2}(\sigma_\tr^{1/4} X\sigma_\tr^{1/4})^2\sigma_{\tr}^{-1/2} \right]}{\|X\|_{2,\sigma_\tr}^2} \nonumber\\
	&= \frac{E_\cN\left[\sigma_\tr^{-1/2}(\sigma_\tr^{1/4} (X-E_\cN[X]+E_\cN[X])\sigma_\tr^{1/4})^2\sigma_{\tr}^{-1/2} \right]}{\|X\|_{2,\sigma_\tr}^2} \nonumber\\
	&= \frac{E_\cN\left[\sigma_\tr^{-1/2}(\sigma_\tr^{1/4} E_\cN[X]\sigma_\tr^{1/4})^2\sigma_{\tr}^{-1/2} \right]}{\|X\|_{2,\sigma_\tr}^2}+ \frac{E_\cN\left[\sigma_\tr^{-1/2}(\sigma_\tr^{1/4} (X-E_\cN[X])\sigma_\tr^{1/4})^2\sigma_{\tr}^{-1/2} \right]}{\|X\|_{2,\sigma_\tr}^2}\nonumber\\ &+\frac{E_\cN\left[\sigma_\tr^{-1/4} (X-E_\cN[X])\sigma_\tr^{1/2}E_\cN[X]\sigma_\tr^{-1/4} \right]}{\|X\|_{2,\sigma_\tr}^2}+\frac{E_\cN\left[\sigma_\tr^{-1/4} E_\cN[X]\sigma_\tr^{1/2}(X-E_\cN[X])\sigma_\tr^{-1/4} \right]}{\|X\|_{2,\sigma_\tr}^2} \label{eq1}\\
	&=p\,(E_\Ncal[X])_\cN+(1-p)\,(X-E_\Ncal[X])_\Ncal\,,\nonumber
	\end{align}
where we used Pythagoras theorem $\|X\|_{2,\sigma_\tr}^2=\|X-E_\cN[X]\|_{2,\sigma_\tr}^2+\|E_\cN[X]\|_{2,\sigma_\tr}^{2}$, and where the last two terms in \Cref{eq1} can be shown to be equal to zero using \Cref{commut,condexp}, since
\begin{align*}
	E_\cN\left[\sigma_\tr^{-1/4} (X-E_\cN[X])\sigma_\tr^{1/2}E_\cN[X]\sigma_\tr^{-1/4} \right]&=E_\cN\left[\sigma_\tr^{-1/4} (X-E_\cN[X])\sigma_\tr^{1/2}E_\cN[X]\sigma_\tr^{-1/2}\sigma_\tr^{1/4} \right]\\
	&=E_\cN\left[\sigma_\tr^{-1/4} (X-E_\cN[X])E_\cN[\sigma_\tr^{1/2}X\sigma_\tr^{-1/2}]\sigma_\tr^{1/4} \right]\\
	&=\sigma_\tr^{-1/4}E_\cN\left[ (X-E_\cN[X])E_\cN[\sigma_\tr^{1/2}X\sigma_\tr^{-1/2}] \right]\sigma_\tr^{1/4}\\
	&=\sigma_\tr^{-1/4}E_\cN\left[ (X-E_\cN[X]) \right]\,E_\cN[\sigma_\tr^{1/2}X\sigma_\tr^{-1/2}]\sigma_\tr^{1/4}\\
	&=0\,,
\end{align*}	
and similarly for the second term. Consequently, since $(X-E_\Ncal[X])_\Ncal=(|X-E_\Ncal[X]|_{2})_\Ncal$, and by a use of the almost convexity of the von Neumann entropy (see Theorem 11.10 of \cite{[NC02]}), the term between brackets in \reff{refffff} can be upper bounded by $H((p,1-p))$, where $H$ denotes the binary Shannon entropy. This is itself upper bounded by $\ln2$, from which we get the result.
\end{proof}
We can now easily prove \Cref{thm4.5}.
\begin{proof}[Proof of \Cref{thm4.5}]
	This is a simple corollary of \Cref{prop_domination_entropy}. Indeed, $\operatorname{LSI}_{2,\Ncal}(c,d)$ applied to $|X-E_\Ncal[X]|_2$ gives
	\begin{align*}
		\operatorname{Ent}_{2,\Ncal}(|X-E_\Ncal[X]|_2)&\leq c~\Ecal_{2,\,\cL}(|X-E_\cN[X]|_2) + d\operatorname{Var}_\Ncal(X)\\
		&\le c~\Ecal_{2,\,\cL}(X) + d\operatorname{Var}_\Ncal(X)
		\end{align*}
	where we used that $\Ecal_{2,\,\cL}(|X-E_\Ncal[X]|_2)\le\Ecal_{2,\,\cL}(X-E_\Ncal[X])=\Ecal_{2,\,\cL}(X)$ (see Theorem 4.7 of \cite{CIPRIANI1997259}). Besides, the $\operatorname{DF}$-Poincaré inequality \ref{PI} implies $\lambda(\Lcal)\operatorname{Var}_\Ncal(X)\leq\Ecal_{2,\,\cL}(X)$. Consequently, we get by \reff{eq_prop_domination_entropy}:
	\[\operatorname{Ent}_{2,\Ncal}(X)\leq \left(c+\frac{d+1}{\lambda(\Lcal)}\right)\Ecal_{2,\,\cL}(X)+\ln\sqrt2\,\|X\|_{2,\,\sigma_\tr}^2,\]
	which is the desired result.	
\end{proof}

\subsection{Bounding log-Sobolev constants via interpolation}\label{univsec}

The idea to use interpolation in order to obtain estimates of the log-Sobolev constants goes back to Gross in \cite{G75b}. The strategy can be summarised as follows: assume a bound of the form $\norm{P_{t_p}}_{2\to p,\sigma}\leq M$ is known for some fixed $t_p\geq0$ and $p>2$, with $M\geq1$. Then can one show by extrapolation from this bound that hypercontravitity holds for all $t\geq0$? The answer is yes and its proof uses the crucial fact that the $\mathbb{L}_p$ norms used for the definition of hypercontractivity form an interpolating family of norms.

\begin{theorem}\label{theo_wHC}
Let $(\mathcal{P}_t)_{t\ge 0}$ be a reversible QMS on $\cB(\cH)$ and assume that for some $2<p\le +\infty$, there exist $t_p,M_p>0$ such that for all $X$ positive semidefinite, $\|\mathcal{P}_{t_p}(X)\|_{(2,p),\,\cN}\le M_p \|X\|_{2,\sigma_\tr}$. Then {$\operatorname{LSI}_{2,\,\cN}\left(\frac{p\,t_p}{p-2},\frac{p}{p-2}\,\ln M_p\right)$ holds.}
\end{theorem}

\begin{proof}
	The proof follows closely the analogous statement for classical Markov chains \cite{Diaconis1996a} and primitive QMS \cite{TPK}. The complex time semigroup 
	\begin{align*}
		\mathcal{P}_z:=\e^{z\mathcal{L}}=\sum_{n=0}^\infty \frac{z^n}{n!}\mathcal{L}^n,~~~~~~~~~~~~~~~~z\in\CC\,,
	\end{align*} 
	 defines an analytic family of operators. Define the time dilated complex semigroup $\tilde{\mathcal{P}}_z:={\mathcal{P}}_{t_pz}$. Since $\mathcal{P}_t$ is reversible, its spectral radius does not change upon the replacement $x\mapsto ix$, and therefore, for any $a>0$ and $X$ positive semidefinite:
	\begin{align*}
		\|	\tilde{\mathcal{P}}_{ia}(X) \|_{2,\sigma_\tr}\le \|X\|_{2,\sigma_\tr}\,.
	\end{align*}
	Therefore,
	\begin{align*}
		\| \tilde{\mathcal{P}}_{1+ia}(X)\|_{(2,p),\,\mathcal{N}}=	\| \tilde{\mathcal{P}}_{1}\circ \tilde{\mathcal{P}}_{ia}(X)\|_{(2,p),\,\mathcal{N}}\le M_p  \|\tilde{\mathcal{P}}_{ia}(X)\|_{2,\sigma_\tr}\le M_p \|X\|_{2,\sigma_\tr}\,.
	\end{align*}
	Hence, by Stein-Weiss' interpolation Theorem (\cite{bergh2012interpolation,SW16}), for all $0\le s \le1$, and any $X\in\Bcal(\Hcal)$:
	\begin{align*}
		\|	\tilde{\mathcal{P}}_{s}(X)\|_{(2,p_s),\,\mathcal{N}}\le M_{p}^s \|X\|_{2,\sigma_\tr}\,,
	\end{align*}
	for $p_s$ such that
	\begin{align*}
		\frac{1}{p_s}=\frac{s}{p}+\frac{1-s}{2}\,.
	\end{align*}
	Taking $t=st_p$ and $p(t):=p_s$, we get
	\begin{align}\label{eq333}
		\| {\mathcal{P}}_{t}(X)\|_{(2,p(t)),\,\mathcal{N}}\le \e^{\frac{t}{t_p}\ln M_p}\|X\|_{2,\sigma_\tr}\,,
	\end{align}
	with equality at $t=0$,	where
	\begin{align*}
		p(t)=\frac{2 p t_p}{(2-p) t +pt_p}\,.
	\end{align*}
	Taking derivatives on both sides of \reff{eq333} at $0$,
	\begin{align}\label{eq555}
		-\frac{\ln M_p}{t_p} \|X\|_{2,\sigma_\tr}+\left.\frac{d}{dt} 	\| {\mathcal{P}}_{t}(X)\|_{(2,p(t)),\,\cN}\right|_{t=0}\le 0\,.
	\end{align}
	Using \Cref{cor}, with $p(0)=2$ and ${p}'(0)= \frac{2(p-2)}{pt_p}$, 
	\begin{align*}
		\left.\frac{d}{dt} 	\| {\mathcal{P}}_{t}(X)\|_{(2,p(t)),\,\cN}\right|_{t=0}=\frac{p-2}{p\,t_p\|X\|_{2,\sigma_\tr}}\left[\operatorname{Ent}_{2,\,\cN}(X)-\frac{p\,t_p}{p-2}\mathcal{E}_{2,\,\cL}(X)\right]\,.
	\end{align*}
	Hence, \reff{eq555} can be rewritten as
	\begin{align}\label{eq29}
            \frac{p-2}{p\,t_p }\operatorname{Ent}_{2,\,\cN}(X)\le \mathcal{E}_{2,\,\cL}(X)+\frac{\ln M_p}{t_p}\|X\|_{2,\sigma_\tr}^2\,,
	\end{align}
which leads to the desired result.
\end{proof}	
In the following corollary, we combine \Cref{theo_wLSI} and \Cref{theo_wHC} to further provide upper bounds on the log-Sobolev constants in terms of the spectral gap of the QMS $(\mathcal{P}_t)_{t\ge 0}$. As such, it can be seen as an extension of Theorem 5 of \cite{TPK} to the case of decohering reversible QMS.
\begin{corollary}\label{coro_bound_constants}
	Given a reversible QMS $(\cP_t)_{t\ge 0}$ with spectral gap $\lambda(\LL)$, $\operatorname{LSI}_{2,\Ncal}(c,\ln\sqrt{2})$ holds, with
	\begin{align*}
		{	c\leq\frac{\ln(\| \sigma_\tr^{-1} \|_\infty)+2}{2\,\lambda(\LL)} \,.}
	\end{align*}	
\end{corollary}
\begin{proof}
	From \Cref{eq24}, we get that for any $X\ge 0$,
	\begin{align*}
		\|X\|_{(2,4),\,\cN}\le \|X\|_{4,\sigma_\tr}\le \|\sigma_\tr^{-1}\|^{1/4}_{\infty}\|X\|_{2,\sigma_\tr}\,,
	\end{align*}	
	where the last inequality is a well-known property of $\mathbb{L}_p$ norms. Together with the contractivity of $(\cP_t)_{t\ge 0}$ (cf. (i) of \Cref{theo_propr_norms}), we find
	\begin{align*}
		\|\cP_t(X)\|_{(2,4),\,\cN}\le \|X\|_{(2,4),\,\cN}\le \|\sigma_\tr^{-1}\|_\infty^{1/4}\|X\|_{2,\sigma_\tr}\,.
	\end{align*}	
We conclude with successive applications of \Cref{theo_wHC} and \Cref{theo_wLSI}, taking the limit $t_4\to0$ and $M_4=\norm{\sigma_\tr^{-1}}_\infty^{1/4}$.
\end{proof}

\section{Non-positivity of the strong LSI constant}\label{sec5bis}

In this section, we show that a strong DF-log-Sobolev inequality does not hold for a non-trivially decohering QMS, that is  a QMS that is neither primitive nor unitary. We deduce from this that the amalgamated $\mathbb{L}_p$ norms do not satisfy uniform convexity for $1\leq p\leq 2$ as soon as $\cN$ is non-trivial.\\

By comparison of Dirichlet forms, it is enough to consider the case of the $\Ncal$-decoherent QMS defined by $\cL_\cN:=E_\cN-\id$, where $\Ncal$ is any $*$-subalgebra of $\Bcal(\Hcal)$ and $E_\Ncal$ is a conditional expectation on it. Indeed, if $(\Pcal_t)_{t\geq0}$ is any decohering QMS with DF-algebra $\DF=\Ncal$ and the same conditional expectation $E_\Ncal$, then the following inequality holds \cite{MSFW}:
	\begin{align*}
		\lambda(\cL)\,\cE_{2,\,\cL_{\cN}}(X)\le\cE_{2,\,\cL}(X)\le \left\|\frac{\cL+\hat{\cL}}{2}\right\|_{\infty\to\infty} \,\cE_{2,\,\cL_{\cN}}(X)\,,
	\end{align*}
where $\hat{\cL}$ is the conjugate of $\cL$ with respect to $\langle .,.\rangle_{\sigma_\tr}$, and $\lambda(\LL)$ is the spectral gap of $(\Pcal_t)_{t\geq0}$. From this inequality we directly obtain that if $\operatorname{LSI}_{2,\,\cN}(c_\cN,0)$ holds for $\cL_\cN$, then $\operatorname{LSI}_{2,\,\cN}(c,0)$ holds for $\cL$ with:
	\begin{align*}
		0<\frac{\lambda(\cL)}{c_\cN} \le  \frac{1}{c} \le \left\|  \frac{\cL+\hat{\cL}}{2}\right\|_{\infty\to\infty}\frac{1}{c_\cN}\,.
	\end{align*}
Our goal is thus to show that if $\Ncal$ is non-trivial and $\operatorname{LSI}_{2,\,\cN}(c_\cN,0)$ holds for $\Lcal_\Ncal$, then $c_\cN=+\infty$.

\begin{theorem}
Let $\Ncal$ be any non-trivial $*$-subalgebra of $\Bcal(\Hcal)$ (that is, $\Ncal\ne\C \mathbb{I}$ and $\Ncal\ne\Bcal(\Hcal)$) and consider the Lindbladian $\cL_\cN:=E_\cN-\id$, where $E_\cN$ is any conditional expectatin on $\cN$. Define $\sigma_{\tr}:=E_{\cN*}(d_{\cH}^{-1}\,\mathbb{I}_\cH)$. Assume that there exists $\alpha\geq0$ such that for all positive semi-definite $X\in\Bcal(\Hcal)$, 
\begin{equation}\label{eq_theo_noLSI}
\alpha\,\operatorname{Ent}_{2,\,\cN}(X)\leq \Ecal_{2,\Lcal_\Ncal}(X)\,.
\end{equation}
Then $\alpha=0$.
\end{theorem}

\begin{proof}
Let $\alpha\geq0$ be such that inequality \eqref{eq_theo_noLSI} holds for all positive semi-definite $X\in\Bcal(\Hcal)$. We shall construct a sequence $(Z_k)_{k\in\NN}$ such that
		\begin{align*}
			\frac{\mathcal{E}_{2,\,\cL_\cN}(Z_k)}{\operatorname{Ent}_{2,\,\cN}(Z_k)}\underset{k\to\infty}\rightarrow0\,,
		\end{align*}	
which directly implies that $\alpha=0$. More precisely, we shall construct a sequence of density matrices $(\rho_k)_{k\geq1}$ such that $Z_k= \Gamma_{\sigma_\tr}^{-\frac{1}{2}}(\sqrt{\rho_k})$ and 
       	\begin{align}\label{eq_proof_noLSI}
		\frac{\cE_{2,\,\cL_\cN}(\sigma_\tr^{-1/4}\sqrt{\rho_k}\,\sigma_\tr^{-1/4})}{D(\rho_k\|\rho_{\cN,k})}\underset{k\to\infty}\rightarrow0\,,
	\end{align}
where $\rho_{\cN,k}:=E_{\cN*}(\rho_k)$. Now assume that $\Hcal$ and $\Ncal$ admit the decomposition given by \Cref{eqtheostructlind1} and \Cref{eqtheostructlind2}. As $\Ncal$ is non-trivial, we can assume that either there exists $i\in I$ such that $\dim\,\Hcal_i>1$ and $\dim\,\Kcal_i>1$, or $|I|>1$. We shall construct a sequence $(\rho_k)_{k\geq1}$ in each case and then treat them simultaneously to prove the limit in \eqref{eq_proof_noLSI}.\\
We start by considering the first case and, without loss of generality, we assume that $\Hcal=\Hcal_A\otimes\Hcal_B$ and that $\Ncal=\Bcal(\Hcal_A)\otimes\id_B$, with $\dim\,\Hcal_B:=d_B>1$ and $\dim\,\Hcal_A=2$. One can recover the general case by adding zeros in the corresponding entries of $\rho_k$. Then, it means that there exists a density matrix $\tau\in\Dcal(\Hcal_B)$ such that for all $\omega\in\mathcal{S}_1(\Hcal)$,
\[E_{\Ncal*}(\omega)=\Tr_{\Hcal_B}(\omega)\otimes \tau\,.\]
We define, in an orthonormal basis in which $\tau$ is diagonal and in any orthonormal basis of $\Hcal_A$,
\begin{align*}
	\Delta:=  \,
	 \begin{pmatrix}
		0&1     \\
		1 & 0 
	\end{pmatrix}\otimes	
	\underbrace{ \begin{pmatrix}
0 & 1 & 0 & \cdots & 0   \\
1 & 0 && \cdots & 0 \\
0 &   & \ddots && \vdots \\
\vdots& &&&\\
0 & &\cdots& &0
\end{pmatrix}}_{d_{B}}\,,\qquad \rho_{\Ncal,k}=\begin{pmatrix}
		\frac1k&0     \\
		0 & 1-\frac1k 
	\end{pmatrix}\otimes\tau\,.
\end{align*}
It is clear that $E_{\Ncal*}(\Delta)=0$. Next, define
\[e_1=\begin{pmatrix}
1 \\ 0
\end{pmatrix}\otimes\begin{pmatrix}
1 \\ 0 \\ \vdots \\ 0
\end{pmatrix}
\,,\qquad e_2=\begin{pmatrix}
0 \\ 1 
\end{pmatrix}\otimes\begin{pmatrix}
0 \\ 1 \\ 0 \\ \vdots \\ 0
\end{pmatrix}\,,\]
so that $\sca{e_i}{\Delta\, e_j}=1-\delta_{ij}$. We also define $\lambda_1:=k\,\sca{e_1}{\rho_{\Ncal,k}\,e_1}$ and $\lambda_2:=\frac{k}{k-1}\sca{e_2}{\rho_{\Ncal_k}\,e_2}$, which clearly do not depend on $k$. We now set, for $\eps\geq0$,
  \begin{align*}
	\rho_{k,\eps}:=\rho_{\Ncal,k}+\eps\,\Delta\,,
\end{align*}          
so that $E_{\Ncal*}(\rho_{k,\eps})=\rho_{\Ncal,k}$. Since the $\rho_{\Ncal,k}$ are full-rank, the $\rho_{k,\eps}$ are well-defined density matrices for $\eps$ small enough.\\
We now turn to the case where $|I|>1$. Up to adding zero entries in the matrices defining $\rho_k$, we can assume that $|I|=2$. Denote by $P_i$ the orthogonal projection on $\Hcal_i\otimes\Kcal_i$ for $i\in I$, and consider $\eta_i=\frac{\II_{\Hcal_i}}{\dim \Hcal_i}\otimes\tau_i$. We also denote by $e_i\in\Hcal_i\otimes\Kcal_i$ an eigenvector of $\eta_i$ of associated eigenvalue $\lambda_i>0$. We then set
\[\Delta=\outerp{e_1}{e_2}+\outerp{e_2}{e_1}\,,\qquad \rho_{\Ncal,k}=\frac1k\,\eta_1 + \left(1-\frac1k\right)\,\eta_2\,,\]
so that again $E_{\Ncal*}(\Delta)=0$ and $\sca{e_i}{\Delta\, e_j}=1-\delta_{ij}$. As before, we define $\rho_{k,\eps}:=\rho_{\Ncal,k}+\eps\,\Delta$.\\
Remark that in both cases, we have $E_{\Ncal*}(\Delta)=0$ and
\begin{equation}\label{eq_proof_noLSI2}
\begin{aligned}
& \Delta=\outerp{e_1}{e_2}+\outerp{e_2}{e_1},\qquad \sca{e_i}{\Delta\, e_j}=1-\delta_{ij}\,,\\
& \lambda_1:=k\,\sca{e_1}{\rho_{\Ncal,k}\,e_1}\,,\qquad \lambda_2:=\frac{k}{k-1}\sca{e_2}{\rho_{\Ncal_k}\,e_2}\,.
\end{aligned}
\end{equation}
This will be enough to treat both cases simultaneously. We shall now prove that the limit in \eqref{eq_proof_noLSI} holds with $\rho_k=\lim_{\eps\to0}\rho_{k,\eps}$. The first step is to obtain a limit for a fixed $k\geq1$ and $\eps\to0$, that is, to obtain a continuous extension of the quotient appearing in the limit at $\rho_{\Ncal,k}$. For this purpose, we compute the Taylor expansion of both the numerator and the denomitator. A simple calculation using the integral representations of the logarithm and of the square root functions \cite{[HMPB11]} shows that (see also the proofs of Theorem 16 in \cite{[KT13]} and Lemma 3.5 in \cite{BarEID17}).
	\begin{align*}
		&D(\rho_{k,\eps}\|\rho_{\cN,k})=\eps^2\,\int_0^\infty\tr\left[ \Delta\, \frac{1}{t-\rho_{\cN,k}}\,\Delta\,\frac{1}{t-\rho_{\cN,k}}\right]\,dt+\mathcal{O}(\eps^3)\\
		&\mathcal{E}_{2,\,\cN}(\sigma_\tr^{-1/4}\sqrt{\rho_{k,\eps}}\,\sigma_\tr^{-1/4})= \pi^2\eps^2\,\iint_{[0,\infty)^2}\,\sqrt{s\,t} \, \tr\left[ \frac{1}{t+\rho_{\cN,k}}\, \Delta\, \frac1{t+\rho_{\cN,k}}\,\Delta    \right] \\
		&~~~~~~~~~~~~~~~~~~~~~-\pi^2\eps^2\,\iint_{[0,\infty)^2}\sqrt{s\,t}\,\tr\left[   \frac{\sigma_\tr^{1/4}}{t+\rho_{\cN,k}}\, \Delta   \frac{ \sigma_\tr^{1/4}}{t+\rho_{\cN,k}}  \,E_\cN\left[\frac{\sigma_\tr^{-1/4}}{ s+\rho_{\cN,k}}\,\Delta \frac{\sigma_\tr^{-1/4}}{s+\rho_{\cN,k}}\right]  \right] +\mathcal{O}(\eps^3)\,.
	\end{align*}
Using \Cref{eq_proof_noLSI2} we can compute explicitely these integrales. For instance, the second integral in the second equation is null, since $E_\cN\left[\frac{\sigma_\tr^{-1/4}}{ s+\rho_{\cN,k}}\,\Delta \frac{\sigma_\tr^{-1/4}}{s+\rho_{\cN,k}}\right]=0$. This can be checked directly using the fact that both $e_1$ and $e_2$ are eigenvectors of $\frac{\sigma_\tr^{-1/4}}{ s+\rho_{\cN,k}}$ and that $E_\Ncal[\outerp{e_1}{e_2}]=E_\Ncal[\outerp{e_2}{e_1}]=0$. We thus obtain:
\begin{align}
	&D(\rho_{k,\eps}\|\rho_{\cN,k})= \eps^2\,   g\left(\frac1k\,\lambda_1,\left(1-\frac1k\right)\lambda_2\right)\,  |\langle e_1|\Delta|e_2\rangle|^2+\mathcal{O}(\eps^3),\label{D2}\\
	&\mathcal{E}_{2,\,\cN}(\sigma_\tr^{-1/4}\sqrt{\rho_{k,\eps}}\,\sigma_\tr^{-1/4})=2\pi^2\,\eps^2 \, f\left(\frac1k\,\lambda_1,\left(1-\frac1k\right)\lambda_2\right)\,|\langle e_1|\Delta|e_2\rangle|^2+\mathcal{O}(\eps^3),\label{E2}
\end{align}	
		where
		\begin{align}\label{eq3}
			f(x,y):=\left\{  \begin{aligned}
				&		\frac{(\sqrt{x}-\sqrt{y})^2}{(x-y)^2}~~~~\text{ if }x\ne y\\
				&\frac{1}{4\,x}~~~~~~~~~~~~~~~~~~\text{else}
			\end{aligned}	  \right. ,~~~~~~~~~~~~~	g(x,y):=\left\{  \begin{aligned}
				&		\frac{\log(x)-\log(y)}{x-y}~~~~\text{ if }x\ne y\\
				&\frac{1}{x}~~~~~~~~~~~~~~~~~~~~~~~\text{else\,.}
			\end{aligned}\right.
		\end{align}
For a fixed $k\geq1$, we thus obtain that
\[\frac{\cE_{2,\,\cL_\cN}(\sigma_\tr^{-1/4}\sqrt{\rho_{k,\eps}}\,\sigma_\tr^{-1/4})}{D(\rho_{k,\eps}\|\rho_{\cN,k})}\underset{\eps\to0}
\longrightarrow\,2\pi^2\,\frac{f\left(\frac1k\,\lambda_1,(1-\frac1k)\lambda_2\right)}{g\left(\frac1k\,\lambda_1,(1-\frac1k)\lambda_2\right)}\,.\]
We just have to take the limit $k\to+\infty$ to conclude. Indeed,
\begin{align*}
& f\left(\frac1k\,\lambda_1,(1-\frac1k)\lambda_2\right)\underset{k\to+\infty}\longrightarrow 1/{\lambda_2}   \,,\\
& g\left(\frac1k\,\lambda_1,(1-\frac1k)\lambda_2\right)\underset{k\to+\infty}\longrightarrow +\infty    \,.
\end{align*}
\end{proof}
The above result implies the following straightforward corollary:
\begin{corollary}\label{coro_noUconv}
The $\mathbb{L}_2\left(\Ncal,\mathbb{L}_p(\sigma_\tr)\right)$ spaces do not satisfy the uniform convexity property.
\end{corollary}

\section{Application to decoherence times}\label{sec6}

In this section, we apply the framework of $\operatorname{DF}$-log-Sobolev inequalities in order to find bounds on the decoherence rates of a non-primitive quantum Markov semigroup. We recall that, for $0<\eps<1$, the decoherence time of a reversible QMS $(\cP_t)_{t\ge 0}$ is defined as
\begin{align*}
\tau_{\text{deco}}(\eps):=\inf\left\{ t\ge 0:~ \|\rho_t-\rho_\cN\|_1\le \eps\right\}\,,
\end{align*}
where $\rho_\Ncal\equiv E_{\Ncal*}(\rho)$. A classical technique to get rapid decoherence for all times comes from looking at the spectral gap of a reversible QMS: 
\begin{align*}
\|\cP_t\left(X-E_\cN[X]\right)\|_{\infty}&\le \|{\sigma_\tr}^{-1}\|_{\infty}^{\frac{1}{2}} \|\cP_t\left(X-E_\cN[X]\right)\|_{2,\sigma_\tr}\nonumber\\
&\le \|{\sigma_\tr}^{-1}\|_{\infty}^{\frac{1}{2}}\, \e^{-\lambda(\LL) t}\|X-E_\Ncal(X)\|_{2,\sigma_\tr}.
\end{align*}
In the dual Schr\"{o}dinger picture, such a bound translates into 
\begin{align}\label{spectral}
\|\cP_{*t}(\rho)-\rho_\cN\|_1\le \|\sigma_\tr^{-1}\|_\infty^{1/2}\e^{-\lambda(\cL)t}.
\end{align}
However, already in the classical case, the spectral gap does not usually provide tight enough bounds on the decoherence time of a Markov semigroup \cite{Diaconis1996a}. Moreover, in practice, the coefficient $ \|{\sigma_\tr}^{-1}\|_{\infty}^{\frac{1}{2}}$ explodes exponentially fast as the dimension of the system grows. If LSI$_{2,\,\cN}(c,0)$ held with $c<\infty$, the original techniques of \cite{TPK} could be adapted to yield
\begin{align*}
\|\cP_{*t}(\rho-\rho_\cN)\|_1\le(2\, \log\|\sigma_\tr^{-1}\|_\infty)^{1/2}\,\e^{-\frac{t}{c}}\,,
\end{align*}
improving significantly the bound \reff{spectral} derived from the spectral gap method. However, as discussed in the last section, a strong LSI never holds for non-primitive QMS. This motivates the search for a technique that would deal with the weak version of the log-Sobolev inequality. Fortunately, such a technique already exists in the classical literature \cite{zegarlinski1995ergodicity,Martinelli1999,Diaconis1996a}: it consists in combining hypercontractivity bounds at short times with the spectral gap at long times. Using such a method, we can prove the exponential convergence in terms on the $\infty$-norm.

\begin{proposition}\label{prop_deco_time}
	Assume that a QMS $(\cP_t)_{t\ge 0}$ satisfies $\operatorname{HC}_{2,\,\cN}(c,d)$, and that $\|\sigma_\tr^{-1}\|_\infty \ge \e$. Then for $t=\frac{c}{2}\ln\ln\|\sigma_\tr^{-1}\|_\infty+\frac{\kappa}{\lambda(\cL)},~\kappa>0\,$, and all $X\in\cB(\cH)$
	\begin{align}
	\|\cP_{t}\left(X-\Ebb_\Ncal[X]\right)\|_{\infty}\le \left(\max_{i\in I} \sqrt{d_{\Hcal_i}}\right)\,\e^{1+d-\kappa}\,\|X\|_\infty\,,
	\end{align}
	where $d_{\Hcal_i}$ denote the dimensions of the spaces $\Hcal_i$ appearing in the decomposition of $\DF$ given by \eqref{eqtheostructlind2}. By duality, we get the following similar bound:
	\begin{align}
	\forall \rho\in\cD(\cH),~~~~~\|\cP_{*t}\left(\rho-E_{\Ncal*}[\rho]\right)\|_1\le \max_{i\in I} \sqrt{d_{\Hcal_i}}\,\e^{1+d-\kappa}\,.
\end{align}
	The above inequality provides a bound on the decoherence time of the QMS:
	\begin{align*}
		\tau_{\operatorname{deco}}(\eps)\le  \frac{\ln\left(\max_{i\in I}\,\sqrt{d_{\cH_i}}\,\eps^{-1}\right)+1+d}{\lambda(\cL)}+\frac{c}{2}\ln\ln\|\sigma_\tr^{-1}\|_\infty\,.
	\end{align*}	
\end{proposition}

\begin{proof}
  Let $t,s>0$. Then:
	\begin{align*}
		\|\cP_{t+s}\left(X-E_\cN[X]\right)\|_{(2,\infty),\,\cN}&\le \|{\sigma_\tr}^{-1}\|_{\infty}^{\frac{1}{p}}\|\cP_{t+s}(X-E_\cN[X])\|_{(2,p),\,\cN}\\
		&\le  \|{\sigma_\tr}^{-1}\|_{\infty}^{\frac{1}{p}}\exp\left(  2d\left(\frac{1}{2}-\frac{1}{p}\right)\right)\|\cP_{ t}(X-E_\cN[X])\|_{2,\sigma_\tr}\\
		&\le \|{\sigma_\tr}^{-1}\|_{\infty}^{\frac{1}{p}}\exp\left(  2d\left(\frac{1}{2}-\frac{1}{p}\right)\right)\,\|X-E_\Ncal[X]\|_{2,\sigma_\tr}\,\e^{-\lambda(\cL)t}\\
		& \le\e^{d}\,\|\sigma_\tr^{-1}\|_\infty^{\frac{1}{p}}\,\|X-E_\Ncal[X]\|_{2,\sigma_\tr}\,\e^{-\lambda(\cL)t}\,,
	\end{align*}
where the first inequality follows from \eqref{eq_normestimate2} in \Cref{normestimate} applied to $\cP_{t+s}(X-E_\cN[X])$, the second inequality from $\operatorname{HC}_{2,\,\cN}(c,d)$, and the last one by definition of the spectral gap. 
Since $\|\sigma_\tr^{-1}\|_\infty\ge \e$, one can choose $s:= \frac{c}{2}\log\log \|\sigma_\tr^{-1}\|_\infty$, and $p\equiv p(s)=1+\log \|\sigma_\tr^{-1}\|_\infty$, so that
\begin{align*}
\|\cP_{t+s}(X-E_\cN[X])\|_{(2,\infty)\,,\,\Ncal}
& \le \|X\|_{\infty}\,\e^{1+d-\lambda\,t}\,,
\end{align*}
where we use that $\|X-E_\Ncal[X]\|_{2,\sigma_\tr}\leq \norm{X}_{2,\sigma_\tr}\leq\|X\|_\infty$. The result follows by applying the following norm estimate proved in \Cref{prop_norm_estim2}:
\[\norm{\id}_{(2,\infty)\to(\infty,\infty)\,,\,\Ncal}\leq\max_{i\in I} \sqrt{d_{\Hcal_i}}\,.\]
By duality, we get,
\begin{align*}
	\norm{\rho_{t+s}-\sigma}_{1} &= \underset{\norm{X}_\infty\leq1}{\sup}\,     \tr\left(\cP_{(t+s)*} (\rho-E_{\cN*}(\rho))\,{X} \right)\\
	&=\underset{\norm{X}_\infty\leq1}{\sup}\,     \tr\left(\rho\,\cP_{(t+s)}( {X}-E_\cN[X]) \right)\\
&\le \max_{i\in I} \sqrt{d_{\Hcal_i}}\, \e^{1+d-\lambda\,t} \,.
\end{align*}
	
	\end{proof}

\section{Illustration on a class of non-primitive QMS}\label{sect_example}

All along this article, we highlighted key differences between the primitive and the non-primitive settings. In particular, these differences appear as coefficients in the hypercontractive constants. In this section, we wish to illustrate on a class of non-primitive QMS how these coefficients naturally emerge from the representation theory representation of Lie subgroups of the unitary group $\mathbb U_n(\C)$ on $\C^n$. Some particular instances of this class appear in the physical litterature under the name \emph{weakly and strongly collective decoherent QMS}. They can also be seen as particular cases of QMS having an \emph{essentially commutative dilation} in terms of Brownian noises (see \cite{K-M4}).\\

Let $G$ be a sub-Lie group of dimension $m\geq1$ of $\mathbb U_n(\C)$ for some positive integer $n\geq1$ and let $(\tilde L_1,...,\tilde L_m)$ be a basis of the corresponding Lie-algebra, where $\tilde L_1,...,\tilde L_m$ are viewed as (anti-selfadjoint) operators on $M_n(\C)$. We write $L_k=i \tilde L_k$, which is a selfadjoint operator. Let $(g_t)_{t\geq0}$ be the solution of the following stochastic differential equation on $\mathbb U_n(\C)$:

\begin{equation}\label{eq_QLangevin}
 dg_t=-\frac12\,\sum_{k=1}^m\,L_k^2\,g_t\,dt+\sum_{k=1}^m\,L_k\,g_t\,dB_t^k\,,
\end{equation}
where $\bold B_t=(B^1_t,...,B^m_t)$ is an $m$-dimensional Brownian process (we refer to \cite{partha12introduction,Mey2} for the technical details, such as existence and unicity of the solution of this equation). A simple It\^o computation shows that $g_t$ is indeed unitary almost surely for all $t\geq0$ and that $g_t\in G$ for all $t\geq0$, since the generators $L_k$ belong to the Lie algebra of $G$. Thus, $(g_t)_{t\geq0}$ is a stochastic process in $G$.
We now embed this stochastic process in the unitary group of a finite dimensional Hilbert space $\Hcal$ by considering the unitary representation $\pi:G\to\mathbb U(\Hcal)$ of $G$ on $\Hcal$ and write $U_t=\pi(g_t)$. Next, we define a QMS on $\Bcal(\Hcal)$ as:
\begin{equation}\label{eq_def_exQMS}
 \Pcal_t(X)=\Ebb[U_t^*\,X\,U_t]\,,\qquad X\in\Bcal(\Hcal)\,,
\end{equation}
where the expectation is taken with respect to the probability law of the stochastic process $(g_t)_{t\geq0}$. By the Hudson-Parthasarathy theory on quantum stochastic calculus \cite{partha12introduction,Mey2} and by a result by Frigerio \cite{frigerio85}, $(U_t)_{t\geq0}$ is a cocycle so that \Cref{eq_def_exQMS} defines a proper QMS on $\Bcal(\Hcal)$.

\Cref{eq_QLangevin} is a special instance of quantum Langevin Equation where the only quantum noises occuring in the equation are classical noises. In the general theory of quantum Langevin Equations developed by Hudson and Parthasarathy \cite{H-P}, more complex noises can occur which leads to a complete dilation theory of QMS on finite dimensional Hilbert spaces. In \cite{K-M4}, the authors completely characterized the QMS that admit \emph{essentially commutative} dilations, which is equivalent to having a dilation in terms of the solution of a quantum Langevin Equation with Brownian noises, as in \Cref{eq_QLangevin}, or Poisson noises (see \cite{attal2016classical} for a discussion on this point).\\

Going back to the analysis of the QMS defined by \Cref{eq_def_exQMS}, by \cite{FR08} and since the operators $\tilde L_k$ form a basis of the Lie algebra of $G$, we have
\begin{equation}\label{eq_asympt_ex}
\Pcal_t(X)\underset{t\to+\infty}{\longrightarrow}\int_{G}\, \pi(g)^*\,X\,\pi(g)\,\mu_{\operatorname{Haar}}(dg)\,,
\end{equation}
where $\mu_{\operatorname{Haar}}$, the Haar measure on $G$, is the unique probability measure on $G$ which is left and right translation invariant. Otherwise states,
\[E_\Ncal[X]=\int_{G}\, \pi(g)^*\,X\,\pi(g)\,\mu_{\operatorname{Haar}}(dg)\,.\]
In particular, the maximally mixed density matrix $\frac{\II_\Hcal}{d_\Hcal}$ is an invariant state of $(\Pcal_t)_{t\geq0}$ which is reversible with respect to it. Moreover, the DF algebra can be easily identified using the decomposition of $\pi$ in irreducible representations (irreps). Write $\Hcal  =\oplus_{\gamma\in\Gamma}\,E_\gamma\otimes F_\gamma$, where $E_\gamma$ are irreps of $G$. Then $\DF$ is the commutant of the $*$-algebra generated by $\pi$, i.e.
 \begin{align*}
\DF & =*-\operatorname{alg}\{\pi(g)\}'   \\
& =\bigoplus_{\gamma\in\Gamma}\, \II_{E_\gamma}\otimes \Bcal(F_\gamma)  
 \end{align*}
We see that the QMS is primitive if and only if the representation is irreducible and if the representation is trivial, then the QMS is trivial, i.e. $\Pcal_t=\II_{\Bcal(\Hcal)}$ for all $t\geq0$.\\

We can now summarize our results for this class of QMS.
 
\begin{proposition}\label{prop_example}
Let $(\Pcal_t=\e^{t\Lcal})_{t\geq0}$ be a decohering QMS defined as above. Then:
 \begin{enumerate}
  \item $\operatorname{HC}_{q,\,\cN}(c,\ln\,(|\Gamma|\,\sqrt{2}))$ holds where
  \[c\leq\frac{\ln d_\Hcal+2}{2\,\lambda(\LL)}\]
and where the number $|\Gamma|$ of block in the decomposition \eqref{eqtheostructlind2} is the number of irreducible sub-representations in $\pi$.
  \item For all $\rho\in\Scal(\Hcal)$, 
  \begin{equation}\label{eq_prop_example}
  \|\cP_{*t}(\rho)-\rho_\cN\|_1\le \max_{\gamma\in \Gamma} \sqrt{|\gamma|}\,\e^{1+\ln(|\Gamma\,\sqrt2)-\kappa}~~~~~\text{for}~t=\frac{c}{2}\ln\ln d_\cH+\frac{\kappa}{\lambda(\cL)},~\kappa>0\,,
  \end{equation}
  where $|\gamma|$ is the multiplicity of the irreducible representation $E_\gamma,~\gamma\in\Gamma$ (that is, the dimension of $F_\gamma$).
 \end{enumerate} 
\end{proposition}

In this article, the only estimates we obtained on the hypercontractive constants $c$ and $d$ are universal and in a sense reflect the properties of the amalgamated $\mathbb{L}_p$ norms. The above proposition shows that these constants also naturally appear in some construction of decohering QMS based on representation theory. We leave to future works the study of the precise hypercontractive constants for such QMS, as well as the study of their spectral gap.

We finish this section by focusing on the simplest case where the Lie algebra of $G$ is one dimensional.

\begin{example}
The \textit{weakly collective decoherence} (WCD) semigroup provides a simple example of such decohering QMS. This evolution has been already extensively studied. For example, it was shown in \cite{kempe2001} that it is a good candidate for fault-tolerant universal quantum computation. 
In this case, take  $G=\{\e^{i\theta\,\sigma_z}\,;\,\theta\in\R\}$ where $\sigma_z$ denotes the Pauli matrix on $\CC^2$:
\begin{align*}
\sigma_z:=\left( {\begin{array}{cc}
	1 &0 \\
	0 &-1 \\
	\end{array} } \right).
\end{align*}
We then take $L=\sigma_z$ in \Cref{eq_QLangevin} and consider the representation of $G$ on $(\CC^2)^{\otimes n}$, $n\geq1$, given by:
\[G\ni g\mapsto\pi(g)=g^{\otimes n}\,.\]
\Cref{eq_def_exQMS} then defines the WCD semigroup on $n$ qubits, denoted by $(\cP_t^{W,n})_{t\ge0}$ and of associated generator $\LL^{W,n}$ given by
\begin{align*}
L^{W,n}:= \sum_{i=1}^n \sigma^{(i)}_z,~~~~~~~\text{where }~~~~~~~\sigma^{(i)}_z:=\mathbb{I}^{\otimes(i-1)}\otimes \sigma_z\otimes \mathbb{I}^{\otimes(n-i)},
\end{align*}	
and trivial Hamiltonian ($H^{W,n}=0$). In this case, one can easily show that the completely mixed state $\sigma_\tr:=\mathbb{I}/2^n$ is invariant, since $\LL^{W,n}_*(\mathbb{I})=0$. Moreover, since $L^{W,n}$ is self-adjoint, $\LL^{W,n}$ satisfies $\sigma_\tr$-DBC with respect to $\mathbb{I}/2^n$. By \Cref{theo_deco}, $(\cP_t)_{t\ge 0}$ displays decoherence. In this simple situation, the group $G$ being abelian, each representation is trivial and it is easy to compute the constants $|\Gamma|$ and $|\gamma|$ of \Cref{prop_example}:
\begin{align*}
 & \Gamma=\{-n,-n+2,...,n-2,n\}\,,\qquad|\Gamma|=n\,, \\
 & |\gamma|=\begin{pmatrix}n\\k\end{pmatrix}\,,\qquad \gamma=n-2k\in\Gamma\,.
\end{align*}
Besides, we can compute the spectral gap of $\LL^{W,n}$:
\begin{proposition}\label{prop_spectral_Wn}
	For any $n\ge 2$, $\lambda(\LL^{W,n})=2$.
\end{proposition}	

\begin{proof}
	In view of \Cref{coro_univconstants}, it suffices to find the spectral gap of the generator $\LL^{W,n}$. This is equivalent to finding the spectral gap of its matrix representation $\tilde{\LL}^{W,n}$ (see e.g. \cite{wolftour}):
	\begin{align*}
	\tilde{\LL}^{W,n}:= L^{W,n}\otimes L^{W,n}-\frac{1}{2}\left( (L^{W,n})^2\otimes\mathbb{I}+\mathbb{I}\otimes (L^{W,n})^2\right).
	\end{align*}	
	One can easily check that $L^{W,n}|i_1,...,i_n\rangle=\sum_{j=1}^n (-1)^{i_j}|i_1...i_n\rangle$ for any $(i_1,...,i_n)\in\{0,1\}^n$, so that 
	\begin{align*}
	\tilde{L}^{W,n}|i_1...i_n\rangle\otimes|j_1...j_n\rangle&=\left[   \sum_{k=1}^n(-1)^{i_k} \sum_{k=1}^n(-1)^{j_k}-\frac{1}{2}\sum_{k,l=1}^n (-1)^{i_k+i_l}-\frac{1}{2}\sum_{k,l=1}^n(-1)^{j_k+j_l}  \right]|i_1...i_n\rangle\otimes |j_1...j_n\rangle\\
	&=-2(|\mathbf{i}|-|\mathbf{j}|)^2|i_1...i_n\rangle\otimes|j_1...j_n\rangle,
	\end{align*}	
	where $|\mathbf{i}|$, resp. $|\mathbf{j}|$, denotes the number of $1$'s in the string $(i_1,...,i_n)$, resp. $(j_1,...,j_n)$. Therefore, the spectral gap of $\LL^{W,n}$ is equal to $2$.
\end{proof}	
Looking at \Cref{eq_prop_example} and assuming that the logarithmic constant $c$ is of order $\log n$, we see that the dominating term in $n$ in the decoherence-time comes from that constant $|\gamma|$, that is,
\[\tau_{\text{deco}}(\eps)=\mathcal O (n)\,.\]
This can be computed using the Stirling formula and the fact that the maximum of $\begin{pmatrix}n\\k\end{pmatrix}$ is acheived for $k\approx n/2$.
\end{example}

\section{CB-log-Sobolev inequality and hypercontractivity}\label{CBlogsob}
In the classical setting, log-Sobolev inequalities satisfy the very useful \textit{tensorization property}, that is, given $n$ primitive Markov semigroups $(P^{(i)}_t)_{t\ge 0}$ with generators $L_i$, $i=1,...,n$, if for each $i$, the semigroup $(P_t^{(i)})_{t\ge 0}$ satisfies the log-Sobolev inequality $\operatorname{LSI}_{2}(c_i,d_i)$, then the product semigroup $(P_t)_{t\ge 0}$ with $P_t=P_t^{1}\otimes\cdots\otimes P_t^{n}$, satisfies the log-Sobolev inequality $\operatorname{LSI}_{2}(\max_i c_i,\sum_i d_i)$. This can be seen as a consequence of the multiplicativity of the classical weighted $\mathbb{L}_p$ norms. It is strongly believed that this latter property no longer holds true in the quantum case, since quantum weighted $\mathbb{L}_p$ norms are not multiplicative. In \cite{[BK16]}, the authors proposed to define the hypercontractivity property with respect to the CB-norm, which is known to be multiplicative even in the noncommutative framework, and proved that it is equivalent to the so-called notion of a CB-log-Sobolev inequality for primitive QMS with invariant state $\mathbb{I}/d_{\cH}$. This provides a way to recover the tensorisation property in the noncommutative framework. Here, we generalize their theory to the case of any primitive QMS. In the next theorem, we establish the equivalence between $\CB$-log-Sobolev inequalities and $\CB$ hypercontractivity, hence extending Theorem 4 of \cite{[BK16]} to any primitive QMS.
\begin{theorem}\label{theo8.2}Let $(\mathcal{P}_t)_{t\ge 0}$ be a primitive QMS on $\cB(\cH)$ with associated generator $\LL$, and let $q\ge 1$, $d\ge0$ and $p(t)=1+(q-1)\e^{\frac{2}{c}t}$ for some constant $c>0$. Then
	\begin{enumerate}
		\item[(i)] If $\operatorname{HC}_{q,\CB}(c,d)$ holds, then $\operatorname{LSI}_{q,\CB}(c,d)$ holds.
		\item[(ii)] If $\operatorname{LSI}_{p(t),\CB}(c,d)$ holds for all $t\ge 0$, then $\operatorname{HC}_{q,\CB}(c,d)$ holds.
	\end{enumerate}
\end{theorem}
\begin{proof}
	We first prove (i). If $\operatorname{HC}_{q,\CB}(c,d)$ holds, then for any $k$ and any $X\in \cB(\CC^{k}\otimes \cH )$, 
	\begin{align*}
		\|  \id_{k}\otimes \cP_t (X)\|_{(q,p(t)),\,\cN_k}\le  \exp\left(  2d\left(  \frac{1}{q}-\frac{1}{p(t)} \right)  \right) \|X\|_{q, \frac{\mathbb{I}_{k}}{k}\otimes \sigma}\,,
	\end{align*}	
	that is $\operatorname{HC}_{q,\,\cN}(c,d)$ holds for the QMS $( \id_{k}\otimes \cP_t)_{t\ge 0}$, for which $\cN( \id_{k}\otimes\cP_t)=\cB(\CC^{k})\otimes \mathbb{I}_{\cH}$ and $\sigma_\tr= \frac{\mathbb{I}_{k}}{k}\otimes \sigma $. The result then follows from a direct application of \Cref{gross1}(i). (ii) follows similarly from \Cref{gross1}(ii).
\end{proof}	
A direct application of the definitions for $\mathbb{L}_p$ regularity of Dirichlet forms then leads to the following:
\begin{theorem}
	Assume that $\operatorname{LSI}_{2,\,\CB}(c,d)$ holds. Then
	\begin{itemize}
		\item[(i)] If the generator $\LL$ is strongly $\Lbb_p$-regular for some $d_0\ge 0$, then $\operatorname{LSI}_{q,\,\CB}(c,d+c\,d_0)$ holds for all $q\geq1$, so that $\operatorname{HC}_{2,\,\CB}(c,d+c\,d_0)$ holds.
		\item[(ii)] If the generator $\LL$ is only weakly $\Lbb_p$-regular for some $d_0\ge 0$, then $\operatorname{LSI}_{q,\,\CB}({2c}, d+c\,d_0)$ holds for all $q\geq1$, so that $\operatorname{HC}_{2,\,\CB}(2c,d+c\,d_0)$ holds. 
	\end{itemize}	
\end{theorem}
As in the decoherence-free case, an application of Proposition 5.2 of \cite{OZ99} together with Theorem 4 of \cite{watrous2004notes} leads to the following corollary:

\begin{corollary}Assume that $\operatorname{LSI}_{2,\,\operatorname{CB}}(c,d)$ holds.
	\begin{itemize}
		\item[(i)] If $\cL$ is reversible, then $\operatorname{LSI}_{q,\,\operatorname{CB}}(c,d+c\,(\|\cL\|_{2\to 2,\,\sigma}+1))$ holds for all $q\geq1$ and consequently $\operatorname{HC}_{2,\,\CB}(c,d+c\,(\|\cL\|_{2\to 2,\,\sigma}+1))$ holds. 
		\item[(ii)] If $\LL$ satisfies $\sigma$-$\operatorname{DBC}$, then $\operatorname{LSI}_{q,\,\operatorname{CB}}(c,d)$ holds for all $q\geq1$ and consequently $\operatorname{HC}_{2,\,\CB}(c,d)$ holds.
	\end{itemize}
\end{corollary}

\begin{proof}
	The result follows directly from the fact that reversibility of $\LL$ w.r.t. $\sigma$ implies reversibility of $  {\id}_{k}\otimes \LL$ w.r.t. $\sigma_\tr$, for any $k\in\NN$, so that \Cref{cor4.4} applies. We conclude by noticing that for any $k\in\NN$,
	\begin{align}
		\|\id_k\otimes \LL\|_{2\to 2,\,\frac{ \mathbb{I}_{k}}{k}\otimes \sigma}&=     
		\|  \id_{k}\otimes (\Gamma^{\frac{1}{2}}_{ \sigma}\circ\LL\circ \Gamma^{-\frac{1}{2}}_{ \sigma })\|_{2\to 2}\nonumber\\
		&=\|\Gamma^{\frac{1}{2}}_{\sigma }\circ \LL\circ \Gamma^{-\frac{1}{2}}_{\sigma }\|_{2\to 2}\label{eq50}\\
		&=\|\LL\|_{2\to 2,\,\sigma}\,,\nonumber
	\end{align}	
where for any super-operator $\Lambda:\cB(\cH)\to \cB(\cH)$, $\|\Lambda\|_{2\to 2}:=\sup_{\|X\|_2=1}\|\Lambda(X)\|_{2}$ denotes the usual super-operator norm induced by the Schatten norm $\|.\|_2$, and where we used Theorem 4 of \cite{watrous2004notes} in \Cref{eq50}. The second part follows similarly to the one of \Cref{cor4.4}. In both cases, hypercontractivity follows from \Cref{theo8.2}.
\end{proof}

	\begin{theorem}[Universal bounds on the $\CB$-log Sobolev constants]
	Let $(\cP_t)_{t\ge 0}$ be a primitive reversible QMS, with unique invariant state $\sigma$ and spectral gap $\lambda(\LL)$. Then, $\operatorname{LSI}_{2,\,\CB}(c,\ln\sqrt{2})$ holds, with 
	\begin{align}
		c\leq \frac{  \ln \| \sigma^{-1}\|_{\infty} +2  }{2\,\lambda(\LL)}\,.
	\end{align}		
\end{theorem}

\begin{proof}
	First notice that for all $k\in\NN$, and any $X\in\cB(\CC^{k}\otimes \cH)$,
	\begin{align}
		\|\id_k\otimes \cP_t(X)\|_{(2,4),\,\cN_k}\le \|X\|_{(2,4),\,\cN_k}&\le \|\id\|_{2\to 4,\,\CB,\,\sigma}\,    \|X\|_{2,\,\frac{\mathbb{I}_{k}}{k}\otimes \sigma}
	\end{align}
	where the first inequality follows from (i) of \Cref{theo_propr_norms}. Then, by \Cref{eq_normestimate3}:
	\begin{align*}
		\|\id\|_{2\to 4,\,\CB,\,\sigma}\le\|\sigma^{-1}\|_\infty^{\frac{1}{4}}\,.
	\end{align*}	
Now, an application of \Cref{theo_wHC} and \Cref{thm4.5} to the QMS $\id_k\otimes \cP_t$ together with the fact that $\lambda(\LL)=\lambda( \id_k\otimes \LL)$ for any $k\in\NN$ allow us to conclude.
\end{proof}
Using the multiplicativity of $\CB$ norms, we directly get the tensorisation property of the $\CB$-log-Sobolev inequality, hence extending Theorem 6 of \cite{[BK16]} to any primitive QMS.
\begin{theorem}\label{theorem7.1}
	Suppose that for all $i=1,...,n$ the primitive QMS $(\cP^{(i)})_{t\ge 0}$ on $\cB(\cH_i)$ generated by $\LL_i$ with invariant state $\sigma_i$ satisfies $\operatorname{LSI}_{q,\CB}(c_i,d_i)$. Then the QMS $(\cP_t)_{t\ge 0}$ on $\cB(\otimes_{i=1}^n\cH_i)$ generated by $\LL:=\sum_{i=1}^n \bigotimes_{k=1}^{i-1}\id_{\cH_k}\otimes \LL_i\otimes \bigotimes_{k=i+1}^n\id_{\cH_k}$ with invariant state $\bigotimes_{i=1}^n \sigma_i$ satisfies $\operatorname{LSI}_{q,\CB}(c,d)$ with $c:=\max_ic_i$ and $d=\sum_{i=1}^n d_i$.
\end{theorem}	

\begin{remark}
	The additivity of the weak $\CB$-log Sobolev constant prevents one from obtaining relevant estimates for a large number of tensorized primitive QMS. In particular, estimating both constants separately as in \Cref{theorem7.1} leads to weaker bounds than the ones found in \cite{TPK,MSFW}. This however does not exclude the possibility of better controlling both constants simultaneously when considering tensor products of QMS. 
\end{remark}

\section{Conclusion and open questions}

In this paper we defined and studied a new notion of hypercontractivity with respect the amalgamated $\Lbb_p$ norms, and the related notion of logarithmic Sobolev inequality, in the setting of non primitive QMS. The amalgamated norms appear as appropriate weighted norms depending on the semigroup via its algebra of effective observables $\cN(\cP)$ as well as a the invariant state $\sigma_\tr$ which acts as a trace on $\cN(\cP)$. We extended some of the important results known in the case of primitive semigroups to the decohering case, namely Gross' integration lemma, as well as multiple bounds on the log-Sobolev constants. This allowed us to derive bounds on decoherence rates from the framework previously developed. Finally, we used these results to extend the recently defined framework of CB-log-Sobolev inequalities for unital QMS \cite{[BK16]} to the general case of a primitive QMS. 

In the decohering case, we proved that a weak log-Sobolev inequality always holds in finite dimensions. We also showed that there is no way of recovering a strong notion of LSI for non-primitive QMS. This is different from \cite{BarEID17} where the DF-modified log-Sobolev inequality was proved to hold in some cases. Since this inequality can be interpreted as the limit $p\to 1$ of the family of LSI$_{p\,\cN}(c,0)$, this raises the question of finding the range of $p$'s for which one can find a non-primitive semigroup for which LSI$_{p,\,\cN}(c,0)$ holds for some $c>0$. Moreover, such a no-go result implies the impossibility for any primitive QMS to satisfy LSI$_{2,\,\CB}(c,0)$, as opposed to LSI$_{1,\,\CB}(c,0)$.

All these results rely heavily on the structure and the properties of the amalgamated $\Lbb_p$ spaces. Further development will require better understanding of these spaces, in particular as interpolating Banach spaces.

\paragraph{Acknowledgements:}

The authors would like to thank Eric Hanson, Mithuna Yoganathan and \'{A}ngela Capel Cuevas, as well as Marius Junge, Carlos Palazuelos and Javier Parcet, for very helpful conversations. We are also indebted to Alexander Müller-Hermes and Daniel Sticlk Fran\c{c}a for pointing us towards Beigi and King's article.

Ivan Bardet is supported by the ANR project StoQ ANR-14-CE25-0003-01.
\newpage

\bibliographystyle{abbrv}
\bibliography{biblio}

\appendix

\section{Proof of \Cref{diffnorm}}\label{diffnormapp}
Let $t\mapsto X(t)\in \cB(\cH)$ be an operator-valued twice continuously differentiable function, where $X(t)>0$ for all $t\in [-\eta,\eta]$, for some $\eta>0$, as well as an increasing twice continuously differentiable function $\RR \ni t\mapsto p(t)$ with $p(0)=q\ge 1$. Define
\begin{align*}
	s(t):=\frac{1}{q}-\frac{1}{p(t)}\,,
\end{align*}	
and for a positive definite operator $A\in \DF$, such that $\|A\|_{1,\sigma_\tr}=1$,
\begin{align*}
	M(t,A):=A^{-s(t)/2}X(t) A^{-s(t)/2}\,.
\end{align*}
Thus $M(t,A)$ is positive definite for any $t\in [-\eta,\eta]$. Define moreover
\begin{align}\label{eq41}
	\Phi(X(t),A,p(t)):=\| M(t,A)\|_{p(t)}\,.
\end{align}
The following proposition gathers straightforward generalization of results proved in \cite{[BK16]} which were used to prove the relation between hypercontractivity and the log-Sobolev inequality for the completely bounded norm. (cf. lemmas 8, 9 of \cite{[BK16]}). We recall that $\cS^+_{\mathbb{L}_1(\sigma_\tr)}$ denotes the set of positive definite operators on the sphere of radius one in $\Lbb_1(\sigma_\tr)$.

\begin{proposition}\label{lemma8}
	For a fixed $t\in(-\eta,\eta)$, $A\mapsto 	\Phi(X(t),A,p(t))^{p(t)}$ is convex for $1\le q\le p(t)\le 2q$ and concave for $1\le  p(t)\le q$. Moreover, the following assertions hold true:
\begin{itemize}
\item[1] The function $(t,A)\mapsto\frac{\partial^2}{\partial t^2}	\Phi(X(t),A,p(t))$ is continuous on $(-\eta,\eta)\times \cN(\cP)\cap\mathcal{S}^+_{\mathbb{L}_1(\sigma_\tr)} $.
\item[2] The function $A\mapsto  	\Phi(X(t),A,p(t))$ is continuously differentiable for all $A\in \cN(\cP)\cap\mathcal{S}^+_{\mathbb{L}_1(\sigma_\tr)} $.
\item[3] For all $A\in\cN(\cP)\cap \mathcal{S}^+_{\mathbb{L}_1(\sigma_\tr)} $ and $t\in (-\eta,\eta)$,
	\begin{align}\label{diff}
		\frac{\partial}{\partial t}	\Phi(X(t),A,p(t))=&\frac{p'(t)	\Phi(X(t),A,p(t))}{p(t)^2 \tr \left[M(t,A)^{p(t)}\right]}\left( -\tr \left[M(t,A)^{p(t)} \right]\ln \tr\left[ M(t,A)^{p(t)}\right]\right.\nonumber\\
	&+\tr \left[M(t,A)^{p(t)}\ln M(t,A)^{p(t)}\right]
	- \left.\tr\left[ M(t,A)^{p(t)}\ln A	 \right]\right.\nonumber\\&\left.+\frac{p(t)^2}{p'(t)}  \tr \left[M(t,A)^{p(t)-1} A^{-s(t)/2}X'(t) A^{-s(t)/2}\right]\right)\,.
	\end{align}
\end{itemize}
\end{proposition}
In what follows, we fix a positive definite $Y\in\Bcal(\Hcal)$ and  set $X(t)=\Gamma_{\sigma_\tr}^{\frac{1}{p(t)}}(Y(t))$, where $t\mapsto Y(t)$ is some twice continuously differentiable matrix-valued function with $Y(0)=Y$. Therefore,
\begin{align*}
	\left.\frac{d}{dt}X(t)\right|_{t=0}= \left.\frac{d}{dt}\Gamma_{\sigma_\tr}^{\frac{1}{p(t)}}(Y(t))\right|_{t=0}=-\frac{p'(0)}{2q^2}\big\{\ln \sigma_\tr,\Gamma_{\sigma_\tr}^{\frac{1}{q}}(Y(0))\big\}+ \Gamma_{\sigma_\tr}^{\frac{1}{q}}(Y'(0))\,,
\end{align*}
where we used that $p(0)=q$ and where $\{\cdot,\cdot\}$ is the anticommutator. Thus, using that $M(0,A)=\Gamma_{\sigma_\tr}^{\frac{1}{q}}(Y(0))$ and that $\Phi(X(0),A,q)=\| Y\|_{q,\sigma_\tr}$, \Cref{diff} reduces to 
\begin{align}\label{diff2}
	&\left.	\frac{\partial}{\partial t}	\Phi(X(t),A,p(t))\right|_{t=0}=\frac{p'(0) }{q^2 \|Y\|^{q-1}_{q,\sigma_\tr} }\left[ -\| Y\|^q_{q,\sigma_\tr}        \ln \|Y\|^q_{q,\sigma_\tr}     +\tr\left(\left[ \Gamma_{\sigma_{\tr}}^{\frac{1}{q}}(Y)\right]^q  
	\ln \left[\Gamma_{\sigma_{\tr}}^{\frac{1}{q}}(Y) \right]^q  \right)\right.\nonumber\\
&	\left. -\tr \left(\left[\Gamma_{\sigma_{\tr}}^{\frac{1}{q}}(Y)\right]^q \ln A\right)-\tr\left( \left[\Gamma_{\sigma_{\tr}}^{\frac{1}{q}}(Y)\right]^q \ln \sigma_\tr\right) +\frac{q^2}{p'(0)}\tr \left(\left[\Gamma_{\sigma_\tr}^{\frac{1}{q}}( Y)\right]^{q-1} \Gamma_{\sigma_\tr}^{\frac{1}{q}}(Y'(0))\right)
	\right].\nonumber\\
	\hphantom{.}
\end{align}
In fact, in the case when $Y(t)=Y\in\cB_{sa}(\cH)$, and $p(t)=q+t$, one can similarly show the following
\begin{align}\label{eq300}
	&\left.	\frac{\partial}{\partial p}	\Phi(\Gamma_{\sigma_\tr}^{\frac{1}{p}}(Y),A,p)\right|_{p=q}=\frac{1 }{q^2 \|Y\|^{q-1}_{q,\sigma_\tr} }\left[ -\| Y\|^q_{q,\sigma_\tr}        \ln \|Y\|^q_{q,\sigma_\tr}     +\tr\left(\left|\Gamma_{\sigma_{\tr}}^{\frac{1}{q}}(Y)\right|^q  
	\ln \left|\Gamma_{\sigma_{\tr}}^{\frac{1}{q}}(Y) \right|^q  \right)\right.\nonumber\\
	&~~~~~~~~~~~~~~~~~~~~~~~~~~~~~~~~~~~~~~~~~~~~~~~~~~\,	\left. -\tr \left(\left|\Gamma_{\sigma_{\tr}}^{\frac{1}{q}}(Y)\right|^q \ln A\right)-\tr\left( \left|\Gamma_{\sigma_{\tr}}^{\frac{1}{q}}(Y)\right|^q \ln \sigma_\tr\right) 
	\right].\nonumber\\
	\hphantom{.}
\end{align}	
Now, define $G(A)$ as the part in the parenthesis:
\begin{align}\label{G}
	G(A)&:=-\| Y\|^q_{q,\sigma_\tr}        \ln \|Y\|^q_{q,\sigma_\tr}     +\tr\left(\left[ \Gamma_{\sigma_{\tr}}^{\frac{1}{q}}(Y)\right]^q  
	\ln \left[\Gamma_{\sigma_{\tr}}^{\frac{1}{q}}(Y) \right]^q  \right)\nonumber\\
	&-\tr \left(\left[\Gamma_{\sigma_{\tr}}^{\frac{1}{q}}(Y)\right]^q \ln A\right)-\tr\left( \left[\Gamma_{\sigma_{\tr}}^{\frac{1}{q}}(Y)\right]^q \ln \sigma_\tr\right) \nonumber\\
	&+\frac{q^2}{p'(0)}\tr \left(\left[\Gamma_{\sigma_\tr}^{\frac{1}{q}}( Y)\right]^{q-1} \Gamma_{\sigma_\tr}^{\frac{1}{q}}(Y'(0))\right)\,,
\end{align}
 and let, for a given $Y\in\cB_{sa}(\cH)$,
\begin{align}\label{mindiff}
	Y_{\cN}:=\frac{E_\cN\left[I_{1,{q}}(Y)\right]    }{\|Y\|_{q,\sigma_\tr}^q}\,.
\end{align}
Next, we derive a formula that  will be useful in what follows.

\begin{lemma}\label{lem_GG}
	With the above notations and for positive semidefinite $Y\in\Bcal(\Hcal)$,
\begin{align}\label{GG}
	G(A)-G(Y_\cN)=\norm{Y}_{q,\sigma_\tr}^q~D(\Gamma_{\sigma_\tr}(Y_\cN)\| \Gamma_{\sigma_\tr}(A) )\,.
\end{align}
\end{lemma}
Remark that $G(A)-G(Y_\cN)$ does not depend on $Y'(0)$ and therefore one can check that the same result holds for $Y\in\cB_{sa}(\cH)$.
\begin{proof}
	First note that
	\begin{align*}\label{eq9}
		G(A)-G(Y_\cN)=\tr \left(\Gamma_{\sigma_\tr}^{\frac1q}(Y)(\log Y_\Ncal-\log A)\right)\,.
		\end{align*}
        As $Y_\Ncal$ and $A$ are in $\DF$, they commute with $\sigma_\tr$ and therefore we get
        \[G(A)-G(Y_\cN)=\norm{Y}_{q,\sigma_\tr}^q\,\tr\left(\frac{\Gamma_{\sigma_\tr}^{\frac1q}(Y)^q}{\norm{Y}_{q,\sigma_\tr}^q}\,\big(\log\Gamma_{\sigma_\tr}(Y_\Ncal)-\log\Gamma_{\sigma_\tr}(A)\big)\right)\,.\]
	Now, as again $Y_\cN,A\in \cN(\cP)$, $\ln Y_\cN$ and $\log A$ also belong to $\cN(\cP)$ and we get
	\begin{align*}
	G(A)-G(Y_\cN)
	& =\norm{Y}_{q,\sigma_\tr}^q\,\tr\left(\frac{\Gamma_{\sigma_\tr}^{\frac1q}(Y)^q}{\norm{Y}_{q,\sigma_\tr}^q}\,E_\Ncal\left[\log\Gamma_{\sigma_\tr}(Y_\Ncal)-\log\Gamma_{\sigma_\tr}(A))\right]\right)   \\
	& = \norm{Y}_{q,\sigma_\tr}^q\,\tr\left(\,E_{\Ncal*}\left(\frac{\Gamma_{\sigma_\tr}^{\frac1q}(Y)^q}{\norm{Y}_{q,\sigma_\tr}^q}\right)\,\left(\log\Gamma_{\sigma_\tr}(Y_\Ncal)-\log\Gamma_{\sigma_\tr}(A)\right)\right) \\
	& = \norm{Y}_{q,\sigma_\tr}^q\,\,\tr\left(\Gamma_{\sigma_\tr}(Y_\Ncal)\left(\log\Gamma_{\sigma_\tr}(Y_\Ncal)-\log\Gamma_{\sigma_\tr}(A)\right)\right)\,,
	\end{align*}
	which is the desired result.
\end{proof}	
\Cref{diffnorm} follows from a direct adaptation of the proof of Theorem 7 of \cite{[BK16]}. In a nutshell, all the lemmas used in \cite{[BK16]} to prove it can be generalized to our framework, when replacing the equation (25) of \cite{[BK16]} by \Cref{GG}. In particular, one can prove that
	\begin{align*}
	\Delta(t):=\frac{1}{t}\left(  \|Y(t)\|_{(q,p(t)),\,\cN}- \|Y\|_{q,\sigma_\tr}\right)-\frac{G(Y_\cN)p'(0)}{q^2\|Y\|_{q,\sigma_\tr}^{q-1}} .
\end{align*}
converges to $0$, which leads to the desired result. The details are provided for sake of completeness.
\begin{lemma}\label{lemma10}
	There exist $\kappa >0$ and $K<\infty$, such that for all $t\in [-\eta/2,\eta/2]$ and $A\in \mathcal{S}(\kappa):=\cN(\cP)\cap\left\{   B>0,~\|B\|_{1,\sigma_\tr}=1,~ \|B-Y_\cN\|_{1,\sigma_\tr}\le\kappa  \right\}$,
	\begin{align*}
		\left|  \Phi(X(t),A,p(t))-\|Y\|_{q,\sigma_\tr}-t~\frac{p'(0)G(A)}{q^2\|Y\|^{q-1}_{q,\sigma_\tr}}\right|\le K t^2\,,
	\end{align*}
where $X(t)=\Gamma_{\sigma_\tr}^{\frac{1}{p(t)}}(Y(t))$.
\end{lemma}
\begin{proof}
	The proof is similar to the one of Lemma 10 of \cite{[BK16]}. Let $t\in [-\eta/2,\eta/2]$. Since the set $\cN(\cP)\cap\left\{B>0,~ \|B\|_{1,\sigma_\tr}=1 \right\}$ is open, there exists $\kappa>0$ such that $\Scal(\kappa)$ is a compact and a subset of $\cN(\cP)\cap\left\{B>0,~ \|B\|_{1,\sigma_\tr}=1 \right\}$. By \Cref{lemma8}, the function $\frac{\partial^2}{\partial t^2}\Phi(X(t),A,p(t))$ is continuous on $(-\eta,\eta)\times\cN(\cP)\cap\mathcal{S}^+_{\mathbb{L}_1(\sigma_\tr)} $. Hence, there exists $K>\infty$ such that
	\begin{align*}
		-2K\le \frac{\partial^2 \Phi(X(t),A,p(t))}{\partial t^2}\le 2K\,,
	\end{align*}
	for all $t\in [-\eta/2,\eta/2]$ and all $A\in 	\mathcal{S}(\kappa)$. Therefore, for any $t\in [-\eta/2,\eta/2]$ and $A\in \mathcal{S}(
	\kappa)$:
	\begin{align*}
		&\left| \Phi(X(t),A,p(t))-\Phi(X(0),A,q)-t \left.\frac{\partial \Phi(X(u),A,p(u))}{\partial u} \right|_{u=0}\right|\\
		&~~~~~~~~~~~~~~~~~~~~~~~~~~~~~~~~~~~~~~~~~~~~~~~~~~~~~~~~~~~~~~~~=\left| \int_0^t (t-v) \frac{\partial^2\Phi(X(v),A,p(v))}{\partial v^2}dv\right|\le Kt^2\,.
	\end{align*}
	Noting that $\Phi(X(0),A,q)=\|Y\|_{q,\sigma_\tr}$ and using the definition of $G(A)$, we find that
	\begin{align*}
		\left| \Phi(X(t),A,p(t))-\|Y\|_{q,\sigma_\tr} -t~\frac{p'(0)G(A)}{q^2 \|Y\|_{q,\sigma_\tr}^{q-1}}\right|\le Kt^2\,,
	\end{align*}
	for all $t\in [-\eta/2,\eta/2]$ and $A\in 		\mathcal{S}(\kappa)$.
	
\end{proof}

\bigskip
\begin{lemma}\label{lemma11}
	With the notations of \Cref{lemma10}, for any $0<\eps\le \kappa$, there exists $\delta>0$ such that for all $t\in [-\delta,\delta]$ there is $A(t) \in \cN(\cP)\cap\mathcal{S}^+_{\mathbb{L}_1(\sigma_\tr)} $ satisfying 
	\begin{align*}
		\|Y(t)\|_{(q,p(t)),\,\cN}= \Phi(X(t),A(t),p(t)),~~~~~~~~~~~~ \| Y_\cN- A(t)\|_{1,\sigma_\tr}\le \eps\,.
	\end{align*}
\end{lemma}
\begin{proof}
	The proof is similar to the one of Lemma 11 of \cite{[BK16]}. Given $\eps\le \kappa$, choose $\delta'>0$ satisfying
	\begin{align*}
		\delta'< \min \left\{\frac{\eta}{2}, \frac{\eps^2\,p'(0)\, \|Y\|_{q,\sigma_\tr} }{4 Kq^2} \right\}
	\end{align*}
	where $K$ is defined in \Cref{lemma10}. We have
	\begin{align*}
		\mathcal{S}(\eps)\subset 	\mathcal{S}(\kappa)\subset  \cN(\cP)\cap\mathcal{S}^+_{\mathbb{L}_1(\sigma_\tr)} 
	\end{align*}
	and so the boundary of $	\mathcal{S}(\eps)$ is contained in $ \cN(\cP)\cap\mathcal{S}^+_{\mathbb{L}_1(\sigma_\tr)} $. Suppose that $A$ is on the boundary of $\mathcal{S}(
	\eps)$, so that
	\begin{align*}
		\|Y_\cN-A\|_{1,\sigma_\tr}=\eps\,.
	\end{align*}
	By the quantum Pinsker inequality,
	\begin{align*}
		D(\Gamma_{\sigma_\tr}(Y_\cN)\|\Gamma_{\sigma_\tr}(A))\ge \frac{1}{2}\|Y_\cN -A\|_{1,\sigma_\tr}^2= \frac{\eps^2}{2}\,.
	\end{align*}
	From \Cref{GG} we deduce
	\begin{align}\label{28}
		G(A)\ge G(Y_\cN)+\frac{\eps^2\, \|Y\|_{q,\sigma_\tr}^q}{2}\,.
	\end{align}
	Let us first consider the case where $t\ge 0$. From \Cref{lemma10}, we deduce that 
	\begin{align}\label{eq133}
	 \Phi(X(t),A,p(t))\ge \|Y\|_{q,\sigma_\tr}+t~\frac{p'(0)\,G(A)}{q^2\,\|Y\|^{q-1}_{q,\sigma_\tr}}-Kt^2\,.
	\end{align}
	Our choice of $\delta'$ implies that for all $0\le t\le \delta'$,
	\begin{align*}
		t~\frac{\eps^2\,p'(0)\,\|Y\|_{q,\sigma_\tr}}{2\,q^2}-Kt^2>Kt^2,
	\end{align*}
	and hence, combining this with \reff{28} and \reff{eq133}, 
	\begin{align}\label{29}
		 \Phi(X(t),A,p(t))>\|Y\|_{q,\sigma_\tr}+t~\frac{p'(0)\,G(Y_\cN)}{q^2\,\|Y\|^{q-1}_{q,\sigma_\tr}}+Kt^2\,.
	\end{align}
	Furthermore, from \Cref{lemma10}, we also deduce that
	\begin{align}\label{30}
	 \Phi(X(t),Y_\cN,p(t))\le \|Y\|_{q,\sigma_\tr}+t~\frac{p'(0)\,G(Y_\cN)}{q^2\,\|Y\|^{q-1}_{q,\sigma_\tr}}+Kt^2\,.
	\end{align}
	Combining \reff{29} and \reff{30}, we find that 
	\begin{align*}
		 \Phi(X(t),Y_\cN,p(t))< \Phi(X(t),A,p(t))\,.
	\end{align*}
	Since this inequality holds for any $A$ on the boundary of $\mathcal{S}(\eps)$, we conclude that for all $0\le t\le \delta'$, the function $A\mapsto  \Phi(X(t),A,p(t))$ has a local minimum $A(t)$ in the interior of $\mathcal{S}(\eps)$. We now choose $0<\delta_+ \le \delta'$ so that $q\le p(t) \le 2q$ for all $0\le t\le \delta_+$ (the existence of $\delta_+>0$ is guaranteed by the assumptions that $p(0)=q\ge 1$ and that $t\mapsto p(t)$ is increasing and differentiable). Applying \Cref{lemma8}, we conclude that, for all $0\le t\le \delta_+$, the local minimum of the convex function $A\mapsto  \Phi(X(t),A,p(t))^{p(t)}$ in the interior of $\mathcal{S}(\eps)$ is in fact a global minimum. Since $A\mapsto  \Phi(X(t),A,p(t))$ and $A\mapsto  \Phi(X(t),A,p(t))^{p(t)}$ share the same minimum $A(t)\in \mathcal{S}(\eps)  $, we conclude that 
	\begin{align*}
		\|Y\|_{(q,p(t)),\,\cN}= \Phi(X(t),A(t),p(t)),~~~~~~~~~~~~~~~~	\| Y_\cN-A(t)\|_{1,\sigma_\tr}\le \eps\,\,.
	\end{align*}
	We consider now the case $t\le 0$. Using \Cref{lemma10} as well as inequality \reff{28}, we deduce that
	\begin{align*}
	 \Phi(X(t),A,p(t))&\le \|Y\|_{q,\sigma_\tr}+t~\frac{p'(0) G(A)}{q^2\,\|Y\|^{q-1}_{q,\sigma_\tr}}+Kt^2\\
		&\le \|Y\|_{q,\sigma_\tr}+t~\frac{p'(0)\,G(Y_\cN)}{q^2\, \|Y\|^{q-1}_{q,\sigma_\tr}}+t~\frac{ \eps^2\,p'(0)\, \|Y\|_{q,\sigma_\tr}}{2\,q^2}+Kt^2\,,
	\end{align*}
	where the second inequality follows from the fact that $t\le 0$. Now, for $-\delta'\le t\le 0$,
	\begin{align*}
		t~\frac{\eps^2\,p'(0)\, \|Y\|_{q,\sigma_\tr}}{2\,q^2}+Kt^2<-Kt^2\,,
	\end{align*}
	and thus
	\begin{align}
		 \Phi(X(t),A,p(t))<\|Y\|_{q,\sigma_\tr} +t\frac{G(Y_\cN)p'(0)}{q^2\,\|Y\|^{q-1}_{q,\sigma_\tr}}-Kt^2\,.
	\end{align}
	Combining with the lower bound for $F$ obtained from \Cref{lemma10} we deduce that
	\begin{align*}
		 \Phi(X(t),Y_\cN,p(t))> \Phi(X(t),A,p(t))\,,
	\end{align*}
	for all $A$ on the boundary of $\mathcal{S}(\eps) $. Thus, we conclude that for all $\delta'\le t\le 0$, the function $A\mapsto  \Phi(X(t),A,p(t))$ has a local maximum in the interior of $\mathcal{S}(\eps)$. Choose $0<\delta_-\le \delta'$ so that $1\le p(t)\le 2$ for all $-\delta_-\le t \le 0$. Applying \Cref{lemma8}, we conclude that the local maximum of the concave function $A\mapsto  \Phi(X(t),A,p(t))^{p(t)}$ in the interior of $\mathcal{S}(\eps)  $ is in fact a global maximum for all $-\delta_-\le t\le 0$. Finally, take $\delta:=\min\{\delta_+,\delta_-\}$ to deduce that for all $t\in [-\delta,\delta]$ there exists $A(t)\in \cN(\cP)\cap  \mathcal{S}^+_{\mathbb{L}_1(\sigma_\tr)}$ satisfying:
	\begin{align*}
		\|Y\|_{(q,p(t)) ,\,\cN}=\Phi(X(t),A(t),p(t)),~~~~~~~~~~~~~~  \|Y_\cN-A(t)\|_{1,\sigma_\tr}\le \eps\,.
	\end{align*}
\end{proof}
We are finally ready to state and prove \Cref{diffnorm}.
\begin{proof}[Proof of \Cref{diffnorm}]
	Recall from \reff{G} that
	\begin{align}\label{GY}
	G(Y_\cN)&=	-\| Y\|^q_{q,\sigma_\tr}        \ln \|Y\|^q_{q,\sigma_\tr}     +\tr\left(\left[ \Gamma_{\sigma_{\tr}}^{\frac{1}{q}}(Y)\right]^q  
		\ln \left[\Gamma_{\sigma_{\tr}}^{\frac{1}{q}}(Y) \right]^q  \right)\nonumber -\tr \left(\left[\Gamma_{\sigma_{\tr}}^{\frac{1}{q}}(Y)\right]^q \ln Y_\Ncal\right)\\
		&~~~~~~~~~~~~~~~~~~~~~~-\tr\left( \left[\Gamma_{\sigma_{\tr}}^{\frac{1}{q}}(Y)\right]^q \ln \sigma_\tr\right) +\frac{q^2}{p'(0)}\tr \left(\left[\Gamma_{\sigma_\tr}^{\frac{1}{q}}( Y)\right]^{q-1} \Gamma_{\sigma_\tr}^{\frac{1}{q}}(Y'(0))\right)\,.
		\end{align}
	Using the expression \reff{mindiff} for $Y_\cN$, \reff{GY} reduces to
	\begin{align*}
		G(Y_\cN)&=  \tr\left(\left[ \Gamma_{\sigma_{\tr}}^{\frac{1}{q}}(Y)\right]^q
		\ln \left[\Gamma_{\sigma_{\tr}}^{\frac{1}{q}}(Y) \right]^q  \right)
		-\tr \left(\left[\Gamma_{\sigma_{\tr}}^{\frac{1}{q}}(Y)\right]^q \ln E_\cN\left[\Gamma_{\sigma_\tr}^{-1}\left( \Gamma_{\sigma_\tr}^{\frac{1}{q}}(Y)\right)^q\right]  \right)\nonumber\\
		&~~~~~~~~~~~~~~~~~~~~~~~~~~~~~ -\tr\left( \left[\Gamma_{\sigma_{\tr}}^{\frac{1}{q}}(Y)\right]^q \ln \sigma_\tr\right) +\frac{q^2}{p'(0)}\tr \left(\left[\Gamma_{\sigma_\tr}^{\frac{1}{q}}( Y)\right]^{q-1} \Gamma_{\sigma_\tr}^{\frac{1}{q}}(Y'(0))\right)\,.
	\end{align*}
	Define now 
	\begin{align}
		\Delta(t):=\frac{1}{t}\left(  \|Y(t)\|_{(q,p(t)),\,\cN}- \|Y\|_{q,\sigma_\tr}\right)-\frac{p'(0)\,G(Y_\cN)}{q^2\,\|Y\|_{q,\sigma_\tr}^{q-1}} \,.
	\end{align}
	We next prove that $\Delta(t)\to 0$ as $t\to 0$. Let $\eps>0$ be such that
	\begin{align}\label{epss}
		0< \eps < \min\big\{\kappa\,,\,\eta\,,\,\frac{{\lambda}_{\min}(\Gamma_{\sigma_\tr}(Y_\cN))}{2}\big\}
	\end{align}
	where $\kappa$ is the parameter introduced in \Cref{lemma10} and ${\lambda}_{\min}(\Gamma_{\sigma_\tr}(Y_\cN))$ is the minimum eigenvalue of $\Gamma_{\sigma_\tr}(Y_\cN)$. According to \Cref{lemma11}, there exists $\delta>0$ such that for every $0<t<\delta$ there is an operator $A(t)\in \cN(\cP)\cap\mathcal{S}^+_{\mathbb{L}_1(\sigma_\tr)}$ such that
	\begin{align*}
		\|A(t)-Y_\cN\|_{1,\sigma_\tr}\le \eps\le \kappa,~~~~~~~~~~~~~		\|Y(t)\|_{(q,p(t)),\,\cN}=\Phi(X(t),A(t),p(t))\,.
	\end{align*}
	Then
	\begin{align*}
		\Delta(t)&=\frac{1}{t}\left(\Phi(X(t),A(t),p(t))- \|Y\|_{q,\sigma_\tr}\right)-\frac{p'(0)\,G(Y_\cN)}{q^2\,\|Y\|_{q,\sigma_\tr}^{q-1}} \nonumber\\
		&=\frac{1}{t}\left(\Phi(X(t),A(t),p(t))- \|Y\|_{q,\sigma_\tr}-t~\frac{p'(0)\,G(A(t))}{q^2\,\|Y\|_{q,\sigma_\tr}^{q-1}}\right)+ \frac{p'(0)\,(G(A(t))-G(Y_\cN))}{q^2\,\|Y\|_{q,\sigma_\tr}^{q-1}}\,.\\
	\end{align*}
	Since $A(t)\in \mathcal{S}(\eps)$, \Cref{lemma10} implies that
	\begin{align}\label{47}
			\left|  \Phi(X(t),A(t),p(t))-\|Y\|_{q,\sigma_\tr}-t~\frac{p'(0)\,G(A(t))}{q^2\,\|Y\|^{q-1}_{q,\sigma_\tr}}\right|\le K t^2\,.
	\end{align}
	Furthermore, from \Cref{GG} and using Lemma 14 of \cite{[BK16]}:
	\begin{align}\label{48}
		|G(A(t))-G(Y_\cN)|= \|Y\|_{q,\sigma_\tr}^q  D(\Gamma_{\sigma_\tr}(Y_\cN)\|\Gamma_{\sigma_\tr}(A(t)))&\le \frac{2\,   \|Y\|_{q,\sigma_\tr}^q }{{\lambda}_{\min}(\Gamma_{\sigma_\tr}(Y_\cN))}\|Y_\cN-A(t)\|_{1,\sigma_\tr}\nonumber\\
		&\le \frac{2\,\|Y\|_{q,\sigma_\tr}^q }{{\lambda}_{\min}(\Gamma_{\sigma_\tr}(Y_\cN))}\,\eps\,.
	\end{align}
	Using \reff{47} and \reff{48}, we obtain the bound
	\begin{align}
		|\Delta(t)|\le K t+\frac{2\,p'(0)\,\|Y\|_{q,\sigma_\tr}}{{\lambda}_{\min}(\Gamma_{\sigma_\tr}(Y_\cN))\,q^2 }\,\eps\,,
	\end{align}
	for all $\eps$ satisfying \reff{epss} and all $0<t<\delta$. Therefore,
	\begin{align}
		\limsup_{t\to 0} |\Delta(t)|\le \frac{2\,p'(0)\,\|Y\|_{q,\sigma_\tr}}{\lambda_{\min}(\Gamma_{\sigma_\tr}(Y_\cN))q^2 }\,\eps\,,
	\end{align} 
	and since $\eps$ may be arbitrarily small, we deduce that
	\begin{align*}
		\limsup_{t\to 0} |\Delta(t)|=\lim_{t\to 0}|\Delta(t)|=0\,.
	\end{align*}
	
\end{proof}	

\section{Towards the proof of \Cref{gross1}(ii)}\label{miniproof1}

In this appendix, we define and study the properties of an object that turns out to be useful in the derivation of \Cref{gross1}(ii): first define the following norm on operators $A\in\cB(\cH)$:
\begin{align*}
	\vertiii{A}_{1,\,\sigma_\tr}:=|I|\,\max_{i\in I}\|P_i\,A\,P_i\|_{1,\,\sigma_\tr}\,.
\end{align*}
In what follows, we also denote by ${\tilde{\mathcal{S}}}^+_{\mathbb{L}_1(\sigma_\tr)}$ the set of positive definite operators $A$ of norm $\vertiii{A}_{1,\,\sigma_\tr}=1$. Now, given a positive semidefinite operator $X$ and $1\le q< p\le \infty$, let
\begin{align}\label{|||}
	\vertiii X _{(q,p),\,\cN}:=\inf_{A\in\cN(\cP)\cap {\tilde{\mathcal{S}}}^+_{\mathbb{L}_1(\sigma_\tr)}}\left\| A^{-1/2r}XA^{-1/2r}\right\|_{p,\,\sigma_\tr}\,.
\end{align}	
The following lemma is straightforward:
\begin{lemma}\label{reduceopt}
For all $X$ positive semidefinite, and any $1\le q< p\le \infty$, $\vertiii{X}_{(q,p),\,\cN}\ge \|X\|_{(q,p),\,\cN}$, and equality holds whenever $|I|=1$. Moreover, the optimum in \Cref{|||} in attained on the subset of positive definite operators $A\in\cN(\cP)$ such that $\|P_i\,A\,P_i\|_{1,\sigma_\tr}=\frac{1}{|I|}$ for all $i\in I$.
\end{lemma}
\begin{proof}
The second part of the lemma follows from the observation that for any two positive semidefinite operators $A\in\cN(\cP)$ and $X\in\cB(\cH)$, 
\begin{align*}
	\|A^{-1/2r}\,X\,A^{-1/2r}\|_{p,\,\sigma_\tr}=\|A^{-1/2r}\,\Gamma_{\sigma_\tr}^{1/p}(X)\,A^{-1/2r}\|_p=\|[\Gamma_{\sigma_\tr}^{1/p}(X)]^{1/2}\,A^{-1/r}\,[\Gamma_{\sigma_\tr}^{1/p}(X)]^{1/2}\|_p\,.
\end{align*}	
Since $\frac{1}{r}=\frac{1}{q}-\frac{1}{p}\le 1$, $x\mapsto x^{1/r}$ is operators monotone and therefore the optimization in \Cref{|||} occurs at the boundary of ${\tilde{\mathcal{S}}}^+_{\mathbb{L}_1(\sigma_\tr)}$, that is for $\|P_i\,A\,P_i\|_{1,\sigma_\tr}=\frac{1}{|I|}$ for all $i\in I$. The first part follows directly form the latter fact, since it implies that $\|A\|_{1,\,\sigma_\tr}=1$.
	\end{proof}	

\Cref{gross1}(ii) relies crucially on the below \Cref{mini1,mini2}, which respectively generalize Lemmas 12 and 13 of \cite{[BK16]} to the non unital case and for $|I|\ge1$. In order to prove these results, we first need to extend \Cref{lemma10,lemma11} to the quantity defined in \Cref{|||}.
\begin{proposition}\label{newlemmaA23}
	 Let $q\ge 1$, $[0,\infty)\ni t\mapsto p(t)$ by a twice continuously differentiable increasing function with $p(0)=q$ and $[0,\infty)\ni Y(t)$ be a twice continuous differentiable positive semidefinite matrix-valued function with $Y(0)=Y$, and for any $\kappa>0$, define $\tilde{\mathcal{S}}(\kappa):=\cN(\cP)\cap \{ B>0,\,\vertiii{B}_{1,\,\sigma_\tr}=1,\,\vertiii{B-\tilde{Y}_\cN}_{1,\,\sigma_\tr}\le \kappa  \}$, where
	\begin{align*}
		\tilde{Y}_\cN:=\sum_{i\in I}\frac{P_i\,E_\cN[I_{1,q}(Y)]\,P_i}{|I|\,\tr[ P_i\,(\Gamma_{\sigma_\tr}^{\frac{1}{q}}(Y))^q\,P_i]}\,.
	\end{align*}
Then, there exists $\tilde{\kappa}>0$ and $\tilde{K}>0$ such that for all $t\ge 0$ and $A\in\tilde{\mathcal{S}}(\tilde{\kappa})$,
	\begin{align}\label{lemmaA2prime}
	\left|  \Phi(\Gamma_{\sigma_\tr}^{\frac{1}{p(t)}}(Y(t)),A,p(t))-\|Y\|_{q,\sigma_\tr}-t~\frac{p'(0)G(A)}{q^2\|Y\|^{q-1}_{q,\sigma_\tr}}\right|\le \tilde{K} t^2\,.
\end{align}
Moreover, for any $\tilde{\eps}\le \tilde{\kappa}$, there exists $\tilde{\delta}>0$ such that for all $t\in [0,\tilde{\delta}]$ there is $A(t)\in\cN(\cP)\cap \tilde{\mathcal{S}}^+_{\mathbb{L}_1(\sigma_\tr)}$ satisfying 
\begin{align}\label{lemmaA3prime}
	\vertiii{Y(t)}_{(q,p),\,\cN}=\Phi(\Gamma_{\sigma_\tr}^{\frac{1}{p(t)}}(Y(t)),A(t),p(t)),~~~~~~~~\vertiii{\tilde{Y}_\cN-A(t)}_{1,\,\sigma_\tr}\le \tilde{\eps}\,.
	\end{align}
	\end{proposition}	
\begin{proof}
	The proof of \reff{lemmaA2prime} follows the exact same lines as the proof of \Cref{lemma10}. Now, let $X(t):=\Gamma_{\sigma_\tr}^{\frac{1}{p(t)}}(Y(t))^{p(t)}$ and given $\tilde{\eps}\le \tilde{\kappa}$, choose $\tilde{\delta}'>0$ satisfying 
	\begin{align*}
		\tilde{\delta}'<  \frac{\tilde{\eps}^2\, \,\min_{j\in I} \tr(P_j\,\Gamma_{\sigma_\tr}^{\frac{1}{q}}(Y)^q\,P_j ) \,p'(0)}{4\,\tilde{K}\,q^2\,\|Y\|_{q,\sigma_\tr}^{q-1}}\,.
		\end{align*}
	Then, we have
	\begin{align*}
		\tilde{\mathcal{S}}(\tilde{\eps})\subset		\tilde{\mathcal{S}}(\tilde{\kappa})\subset\cN(\cP)\cap\tilde{\mathcal{S}}^+_{\mathbb{L}_1(\sigma_\tr)}\,.
		\end{align*}
Suppose that $A$ belongs to the boundary of $		\tilde{\mathcal{S}}(\tilde{\eps})$, so that
\begin{align*}
	\vertiii{\tilde{Y}_\cN-A}_{1,\,\sigma_\tr}=\tilde{\eps}.
	\end{align*}
Hence, as in the proof of \Cref{lem_GG}, we can show that
\begin{align*}
	G(A)-G(\tilde{Y}_\cN)&=\tr\left(\, \Gamma_{\sigma_\tr}^{\frac{1}{q}}(Y)^q\,(\ln \tilde{Y}_\cN-\ln A)\right)\\
	&=\sum_{i\in I}\tr\left( E_{\cN*}[P_i\,\Gamma_{\sigma_\tr}^{\frac{1}{q}}(Y)^q\,P_i](\ln \tilde{Y}_\cN-\ln A)\right)\,.
\end{align*}	
	Now, define for any $i\in I$ the states $\sigma_i:=|I|\,P_i\,\sigma_\tr^{1/2}A\sigma_\tr^{1/2}P_i$
		and $\eta_i:=|I|\,P_i\,\sigma_\tr^{1/2}\tilde{Y}_\cN\sigma_\tr^{1/2}P_i$, one can easily verify that $E_{\cN*}[P_i\,\Gamma_{\sigma_\tr}^{\frac{1}{q}}(Y)^q\,P_i]=\tr(P_i\,\Gamma_{\sigma_\tr}^{\frac{1}{q}}(Y)^q\,P_i)\,\eta_i$, so that
		\begin{align*}
			G(A)-G(\tilde{Y}_\cN)&=\sum_{i\in I}\tr( E_{\cN*}[P_i\,\Gamma_{\sigma_\tr}^{\frac{1}{q}}(Y)^q\,P_i] (\ln\eta_i-\ln\sigma_i) )\\
			&=\sum_{i\in I}\tr(P_i\,\Gamma_{\sigma_\tr}^{\frac{1}{q}}(Y)^q\,P_i)D(\eta_i\|\sigma_i)\\
			&\ge\frac{1}{2}\sum_{i\in I}\tr(P_i\,\Gamma_{\sigma_\tr}^{\frac{1}{q}}(Y)^q\,P_i)\,\|\eta_i-\sigma_i\|_1^2\\
			&=\frac{|I|^2}{2}\sum_{i\in I}\tr(P_i\,\Gamma_{\sigma_\tr}^{\frac{1}{q}}(Y)^q\,P_i)\,\|P_i\,A\,P_i-P_i\,\tilde{Y}_\cN\,P_i\|_{1,\,\sigma_\tr}^2\\
			&\ge \frac{|I|^2}{2} \min_{j\in I}\,\tr(P_j\,\Gamma_{\sigma_\tr}^{\frac{1}{q}}(Y)^q\,P_j)\,
			\sum_{i\in I}\,\|P_i\,A\,P_i-P_i\,\tilde{Y}_\cN\,P_i\|_{1,\,\sigma_\tr}^2
			\\
			& \ge \frac{1}{2} \min_{j\in I}\,\tr(P_j\,\Gamma_{\sigma_\tr}^{\frac{1}{q}}(Y)^q\,P_j)\,
			\vertiii{ A-\tilde{Y}_\cN}_{1,\,\sigma_\tr}^2  \\
			&\ge\frac{\tilde{\eps}^2}{2}\min_{j\in {I}}\tr(P_j\,\Gamma_{\sigma_\tr}^{\frac{1}{q}}(Y)^q\,P_j)\,,
			\end{align*}
	where we used Pinsker's inequality on the third line above. Following the steps of the proof of \reff{29}, we can show from \reff{lemmaA2prime} that for all $0\le t\le \tilde{\delta}'$,
	\begin{align*}
		\Phi(X(t),\,A,\,p(t))>\|Y\|_{q,\,\sigma_\tr}+t\,\frac{p'(0)\,G(\tilde{Y}_\cN)}{q^2\,\|Y\|_{q,\,\sigma_\tr}^{q-1}}+\tilde{K}\,t^2\,.
	\end{align*}	
	This, together with another use ot \reff{lemmaA2prime} applied to $A=\tilde{Y}_\cN$ implies that
	\begin{align*}
		\Phi(X(t),\tilde{Y}_\cN,p(t))<		\Phi(X(t),A,p(t))\,.
	\end{align*}
The rest of the proof follows similarly to the proof of \Cref{lemma11}.
\end{proof}

\begin{lemma}\label{mini1}
	Let $Y\in\cB(\cH)$ positive definite and for $1\le q<p\le \infty$, let $\frac{1}{r}:=\frac{1}{q}-\frac{1}{p}$. Then the function $ \Psi(Y,\,.\,,\,p):\,A \mapsto  \|A^{-1/2r}YA^{-1/2r}\|_{p,\,\sigma_\tr}$ is strictly convex. Moreover, there exists a unique $\tilde{A}\in \cN(\cP)\cap \tilde{\mathcal{S}}^+_{\mathbb{L}_1(\sigma_\tr)}$ such that
	\begin{align}\label{eq28}
		\Psi(Y,\tilde{A},p)=\vertiii{Y}_{(q,p),\,\cN}\,.
	\end{align}
	Moreover, the optimizer $\tilde{A}$ of \Cref{eq28} satisfies the following constraint
	\begin{align}\label{optim}
	P_i	\tilde{A}P_i=\frac{\,P_iE_\cN\left[I_{1,p}(\tilde{A}^{-1/2r}Y\tilde{A}^{-1/2r})\right]    P_i}{|I|\,\tr\left[   P_i\,\left(\Gamma_{\sigma_\tr}^{\frac{1}{p}}(\tilde{A}^{-1/2r}\,Y\,\tilde{A}^{-1/2r})\right)^{p} \,P_i\right]}\,.
	\end{align}
\end{lemma}

\begin{proof}
	Following the exact same steps as in the proof of Lemma 12 of \cite{[BK16]}, one can show that the function 
	\begin{align*}
		\Phi(X,\,.\,,\,p):A\mapsto \|A^{-1/2r}XA^{-1/2r}\|_{p}\,
	\end{align*}	
is strictly convex. Let $X=\Gamma_{\sigma_\tr}^{\frac{1}{p}}(Y)$. The first point then follows from the observation that $[A,\sigma_\tr]=0$ for $A\in\DF$, so that $\Psi(Y,A,p)=\Phi(X,A,p)$.
	The fact that the infimum is achieved at a unique point $\tilde{A}$ also follows from the same lemma. Now, we prove \Cref{optim}. Let $A\in\DF$ such that for all $i\in I$, $\tr(P_i\,\tilde A\,P_i)=\frac{1}{|I|}$. Moreover, let $D\in\cN(\cP)$ be a self-adjoint operator such that $\tr(\sigma_\tr P_iDP_i)=0$ for all $i\in I$. Then, it follows that for any $x\in \RR$ sufficiently small, $A(x):=A+x\,D$ satisfies the same constraints as $A$. Let $B(x):= X^{\frac{1}{2}}A(x)^{-s/2}$ and $C(x):=A(x)^{s/2}\frac{d}{dx}A(x)^{-s/2}\in\cN(\cP)$, where $s=1/r$. Then the minimum is achieved at $A$ if for any such $D$,
	\begin{align}
0=\left.		\frac{d}{dx}\right|_{x=0}\Phi(X,A(x),p)^p&=\left.\frac{d}{dx}\right|_{x=0}\tr\left[(B(x)^*B(x))^p\right]\nonumber\\
		&=p\,\tr\left[  (B(0)^*B(0))^{p-1}  (B(0)^*B(0)C(0)+C(0)^*B(0)^*B(0))  \right]\nonumber\\
		&=p\,\langle \,\Gamma_{\sigma_\tr}^{-1}\big((B(0)^*B(0))^p\big), \,\big(C(0)+C(0)^*\big) \rangle_{\sigma_{\tr}}\nonumber\\
		&=p\,\langle  E_\cN\big[ \Gamma_{\sigma_\tr}^{-1}\big((B(0)^*B(0))^p\big)\big], \,\big(C(0)+C(0)^*\big) \rangle_{\sigma_\tr}\nonumber\\
		&=p\,\langle  A^{-1/2}\, E_\cN\big[\Gamma_{\sigma_\tr}^{-1}\big((B(0)^*B(0))^p\big) \big]  \,A^{-1/2},\,\Lambda_A(D) \rangle_{\sigma_\tr},\,\label{Atilde}
	\end{align}	
where $D\mapsto\Lambda_A(D):=A^{\frac{1}{2}}(C(0)+C(0)^*)A^{\frac{1}{2}}$ maps the space of Hermitian operators $D$ in $\cN(\cP)$ such that $\tr [\sigma_\tr P_iDP_i]=0$ for all block $i\in I$ onto itself. Indeed, for any such $D$,
\begin{align*}
	\tr [\sigma_\tr\,P_i\Lambda_A(D)P_i]&=2\,\tr\,\sigma_{\tr}\,P_i A^{s/2+1}\left.\frac{d}{dx}\right|_{x=0}(A(x)^{-s/2})P_i\\
	&=2\,\tr \,\sigma_{\tr}\,P_iA^{s/2+1} \left(-s/2\right)A^{-s/2-1}DP_i\\
	&=-s\,\tr\,[\sigma_\tr P_iDP_i]=0\,.
\end{align*}
Moreover, the map $D\mapsto\Lambda_A(D)$ is onto. To show this, we extend the definition of this map to a linear operator $\tilde{\Lambda}_A$ on the whole space of self-adjoint operators in $\cN(\cP)$ and prove that $\tilde{\Lambda}_A$ is onto. First, notice that $D\mapsto D^{-s/2}$ is one-to-one on the set of positive definite matrices in $\cN(\cP)$, and hence its differential at $A$
\begin{align}\label{eq39}
	D\mapsto \left.\frac{d}{dx}\right|_{x=0}(A+x\,D)^{-s/2}
\end{align}	
is onto on $\cN(\cP)\cap\cB_{sa}(\cH)$. This directly implies that $\tilde{\Lambda}_A$ is onto, since it derives from the map defined in \Cref{eq39} by multiplication with positive definite operators. Hence, $\Lambda_A$ is onto, which together with \Cref{Atilde} implies that for any $D\in \cN(\cP)\cap \cB_{sa}(\cH)$ satisfying $\tr[ \,\sigma_\tr\,P_iDP_i\,]=0$ for all $i\in I$,
\begin{align}
	\langle  A^{-1/2} E_\cN[\Gamma_{\sigma_\tr}^{-1} ((B(0)^*B(0))^p) ] A^{-1/2},\,D \rangle_{\sigma_\tr}=0\,.
\end{align}	
Thus, in each block $i\in I$, $P_i\,A^{-1/2} E_\cN[\Gamma_{\sigma_\tr}^{-1} ((B(0)^*B(0))^p) ] A^{-1/2}\,P_i$ is a multiple of the identity:
\begin{align*}
	P_i\,E_\cN[\Gamma_{\sigma_\tr}^{-1} ((B(0)^*B(0))^p) ] \,P_i=c_i~P_iAP_i,~~~~~~~~~~c_i\in\RR\,.
\end{align*}	
Replacing $B(0)$ by its definition we find
\begin{align}\label{eq40}
	P_i\,E_\cN[I_{1,p}(\tilde{A}^{-1/2r}Y\tilde{A}^{-1/2r}) ] \,P_i=c_i~P_iAP_i\,.
\end{align}	
Finally, the multiplicative factors $c_i$ are found after tracing \Cref{eq40} against $\sigma_\tr$, using the fact that $\tr(\sigma_\tr P_iAP_i)=\frac{1}{|I|}$ for all $i\in I$, and \Cref{optim} follows after rearranging the terms in \Cref{eq40}.
\end{proof}	
\begin{lemma}\label{mini2}
	Given $X\in\cB(\cH)$ positive definite and $q\ge 1$, the function
	\begin{align*}
		[0,\infty)	\ni t\mapsto\varphi(t):= \|\cP_t(Y)\|_{(q,p(t)),\,\cN}\equiv \Phi(X(t),\tilde{A}(t),p(t))
	\end{align*}	
	is continuous on $[0,\infty)$, for $p(t):=1+(q-1)\,\e^{2t/c}$, where $\Phi$ is the map defined in \Cref{eq41}, $X(t)\equiv \Gamma_{\sigma_\tr}^{\frac{1}{p(t)}}(\cP_t(Y))$ and $\tilde{A}(t)$ is the optimizer obtained in \Cref{mini1}.
\end{lemma}	
\begin{proof}
From \reff{lemmaA2prime}, there exist $\tilde{\kappa} >0$ and $\tilde{K}<\infty$, such that for all $t\in [0,\infty)$ and $A\in \tilde{\mathcal{S}}(\tilde{\kappa})$,
\begin{align*}
	\left|  \Phi(X(t),A,p(t))-\|Y\|_{q,\sigma_\tr}\right|\le t~\frac{p'(0)G(A)}{q^2\|Y\|^{q-1}_{q,\sigma_\tr}}+ \tilde{K} t^2\,.
\end{align*}
Moreover, from the second part of \Cref{newlemmaA23} we know that, for sufficiently small $t$, the optimizer $\tilde{A}(t)$ is in $\tilde{\mathcal{S}}(\tilde{\kappa})$. Since $\varphi(0)=\|Y\|_{q,\sigma_\tr}$, the above inequality implies
\begin{align*}
	\left| \varphi(t)-\varphi(0)\right|\le t~\frac{p'(0)G(\tilde{A}(t))}{q^2\|Y\|^{q-1}_{q,\sigma_\tr}}+ K t^2\,.
\end{align*}	
By definition, the map $A\mapsto G(A)$ defined in \Cref{G} is continuous, and hence uniformly bounded on $\tilde{\mathcal{S}}(\tilde{\kappa})$. Hence, the continuity of $\varphi$ at $0$ follows. We now prove the continuity of $\varphi$ at $t_0>0$. For any $0<a<t_0<b$, $t\in[a,b]$ and $s(t)=\frac{1}{q}-\frac{1}{p(t)}$,
\begin{align*}
	\varphi(t)&=\Phi(X(t),\tilde{A}(t),p(t))\\
	&= \|\tilde{A}(t)^{-s(t)/2} X(t)\tilde{A}(t)^{-s(t)/2}\|_{p(t)}\\
	&\ge \|\tilde{A}(t)^{-s(t)}\|_{p(t),\sigma_\tr}~\|\cP_t(Y)^{-1/2}\|_\infty^{-2}\\
	&\ge \lambda_{\min}(\sigma_\tr)\lambda_{\min}(\tilde{A}(t))^{-s(t)}~\|\cP_t(Y)^{-1/2}\|_\infty^{-2}\,,
	\end{align*}
where $\lambda_{\min}(\tilde{A}(t))$ is the minimum eigenvalue of $\tilde{A}(t)$. On the other hand,
\begin{align*}
	\varphi(t)=\inf_{A}   \Phi(X(t),A,p(t))\le\Phi(X(t),\mathbb{I},p(t))=\|X(t)\|_{p(t)}\le \|\cP_t(Y)\|_{p(t),\sigma_\tr}\le \|\cP_t(Y)\|_\infty\,.
\end{align*}	
Together with the previous bound, we arrive at
\begin{align*}
	 \lambda_{\min}(\tilde{A}(t))^{-s(a)}\le  \lambda_{\min}(\tilde{A}(t))^{-s(t)}\le \lambda_{\min}(\sigma_\tr)\|\cP_t(Y)^{-1/2}\|_\infty^2~\|\cP_t(Y)\|_{\infty}\,.
\end{align*}	
Above, we used that $t\mapsto s(t)$ increases, as well as the fact that $  \lambda_{\min}(\tilde{A}(t))\le 1$, since $\|\tilde{A}(t)\|_{1,\sigma_\tr}=1$. By continuity of $t\mapsto \cP_t(Y)$, the right hand side of the above chain of inequalities is uniformly bounded by some positive constant $C>0$ over the interval $[a,b]$. Therefore, $\tilde{A}(t)$ belongs to the compact set $\mathcal{R}:=\cN(\cP)\cap\{B>0,~\|B\|_{1,\sigma_\tr}=1,~\lambda_{\min}(B)\ge C^{-1/s(a)}\}$. The function $(t,A)\mapsto \Phi(X(t),A,p(t))$ restricted to the compact set $[a,b]\times \mathcal{R}$ is uniformly continuous, which means that for any $\eps>0$, there exists $\delta>0$ such that for all $t,t'\in[a,b]$ such that $|t-t'|\le \delta$, and any $A\in\mathcal{R}$, 
\begin{align*}
	|  \Phi(X(t),A,p(t))-\Phi(t',A,p(t')) |\le \eps\,.
\end{align*}	
Therefore,
\begin{align*}
	\varphi(t)=\Phi(X(t),\tilde{A}(t),p(t))\le \Phi(X(t), \tilde{A}(t'),p(t))\le \Phi(X(t'), \tilde{A}(t'),p(t'))+\eps=\varphi(t')+\eps\,.
\end{align*}
Conversely, $\varphi(t')\le \varphi(t)+\eps$. Thus, $|\varphi(t)-\varphi(t')|\le \eps$ for all $|t-t'|\le\delta$. We established the continuity of $\varphi$ on the interval $[a,b]$, and hence at the point $t=t_0\in[a,b]$.

\end{proof}

\section{Some norm estimates}\label{normestimate}

In this appendix we discuss how our results, in particular \Cref{gross1}, can be applied to obtain some estimations on the amalgamated $\mathbb{L}_p$ norms. Consider a subalgebra $\Ncal$ of $\Bcal(\Hcal)$ for some finite dimensional Hilbert space $\Hcal$ and let $E_\Ncal$ be a conditional expectation from $\Bcal(\Hcal)$ to $\Ncal$. We can define $\sigma_\tr$ by \Cref{eq_reference_state} and subsequently the norms $\norm{\cdot}_{(p,q),\,\Ncal}$ as in \Cref{eq111} and \Cref{eq222}. In \Cref{coro_bound_constants}, we use that for any $p\ge2$:
\begin{equation}\label{eq_normestimate1}
\norm{\id}_{(2,2)\to(2,p),\,\Ncal}\leq\norm{\id}_{2\to p,\,\sigma_\tr}=\norm{\sigma_\tr^{-1}}_{\infty}^{\frac12-\frac1p}\,.
\end{equation}
Using the Riesz-Thorin interpolation Theorem and that $\id$ is contractive for $\norm{\cdot}_{(2,q)\,,\Ncal}$, this implies the following estimate that we used in the proof of \Cref{prop_deco_time}:
\begin{lemma}
For any	$2\leq p\le q\le\infty$,
\begin{equation}\label{eq_normestimate2}
\norm{\id}_{(2,p)\to(2,q),\,\Ncal}\leq \norm{\sigma_\tr^{-1}}_{\infty}^{\frac1p-\frac1q}\,.
\end{equation}
\end{lemma}
 In general, the bound given by \eqref{eq_normestimate1} can be very bad. In the bipartite scenario where $\Hcal=\Hcal_A\otimes\Hcal_B$, $\Ncal=\Bcal(\Hcal_A)\otimes \II_{\cH_B}$ and $\sigma_\tr=\frac{\mathbb{I}_{\cH_A}}{d_{\cH_A}}\otimes \sigma$ for some full-rank density matrix $\sigma$, one can get the better bound
\begin{lemma}
	For any $2\le p\le \infty$,
\begin{align}\label{eq_normestimate3}
\norm{\id_{\cB(\cH)}}_{(2,2)\to(2,p),\,\Ncal}\leq\norm{\id_{\cB(\cH_B)}}_{2\to p,\,\CB,\,\sigma}\leq\norm{\sigma^{-1}}_{\infty}^{\frac12-\frac1p}\,.
\end{align}
In particular, the outer bound does not depend on $\Hcal_A$. 
\end{lemma}
\begin{proof}
	The first inequality is obvious by definition of the weighted $\CB$ norms. For the second inequality, it is enough to prove that for all $\cH_A$,
	\begin{align}\label{eq10000}
\norm{\id_{\cB(\cH)}}_{(2,2)\to(2,p),\,\Ncal}\leq\norm{\sigma^{-1}}_{\infty}^{\frac12-\frac1p}
	\end{align}
Now, for any fixed $\cH_A$ and $\frac{1}{r}=\frac{1}{2}-\frac{1}{p}$,
\begin{align*}
	\norm{\id_{\cB(\cH)}}_{(2,2)\to(2,p),\,\Ncal}&=\sup_{X\in\cB(\cH)}\frac{\|X\|_{(2,2)\to (2,p),\,\cN}}{\|X\|_{2,\,\frac{\mathbb{I}_{\cH_A}}{d_{\cH_A}}\otimes\sigma}}\\
	&=\sup_{X\in\cB(\cH)}\inf_{A\in\cB(\cH_A)}\frac{\|   (A\otimes\mathbb{I}_{\cH_B})^{-1}\,X\, (A\otimes\mathbb{I}_{\cH_B})^{-1}        \|_{p,\,\frac{\mathbb{I}_{\cH_A}}{d_{\cH_A}}\otimes\sigma}\,\|A\|^2_{2r,\,\frac{\mathbb{I}_{\cH_A}}{d_{\cH_A}}} }{\|X\|_{2,\,\frac{\mathbb{I}_{\cH_A}}{d_{\cH_A}}\otimes\sigma}}\\
	&=\sup_{X\in\cB(\cH)}\inf_{A\in\cB(\cH_A)}\frac{\|   (A\otimes\mathbb{I}_{\cH_B})^{-1}\,X\, (A\otimes\mathbb{I}_{\cH_B})^{-1}        \|_{p,\,\mathbb{I}_{\cH_A}\otimes\sigma}\,\|A\|^2_{2r} }{\|X\|_{2,\,\mathbb{I}_{\cH_A}\otimes\sigma}}\,.
\end{align*}
Assuming $p=\infty$, the above right hand side is bounded by
\begin{align*}
\sup_{X\in\cB(\cH)}\frac{\|X\|_{(2,2)\to (2,\infty),\,\cN}}{\|X\|_{2,\,\frac{\mathbb{I}_{\cH_A}}{d_{\cH_A}}\otimes\sigma}}&=	\sigma_{\min}^{-1/2}\sup_X\frac{1}{\|X\|_{2}}\,\inf_{A\in\cB(\cH_A)}\|  (A\otimes\mathbb{I}_{\cH_B})^{-1}\,X\, (A\otimes\mathbb{I}_{\cH_B})^{-1}    \|_{\infty}\,\|A\|_{2r}\\
	&=\sigma_{\min}^{-1/2}\sup_{X\in\cB(\cH)}\frac{1}{\|X\|_{2}}\,\|  X    \|_{(2,\infty)}\\
	&=\sigma_{\min}^{-1/2}\,\|\id\|_{2\to \infty,\,\CB}\\
	&\le \sigma_{\min}^{-1/2}
\end{align*}	
where $\|X\|_{(2,\infty)}$ denotes the (unnormalized) $(2,\infty)$ norm of Pisier \cite{Pis93}, and $\|.\|_{2\to\infty,\,\CB}$ the corresponding $\CB$ norm. We conclude by interpolating for fixed $\cH_A$ at the level of \Cref{eq10000}, since $\|\id_{\cB(\cH)}\|_{2\to 2,\sigma_{\tr}}=1$.
\end{proof}	
One could hope to improve this bound by applying \Cref{gross1} to the trivial QMS $(\Pcal_t)_{t\geq0}=\id$:
\begin{proposition}\label{prop_norm_estimate}\ 
\begin{enumerate}
\item[(i)]Assume that $\norm{\id}_{(2,2)\to(2,p),\,\Ncal}\leq C^{\frac12-\frac1p}$ for some $C>0$ and for all $p\geq2$. Then $\Dent{\rho}{E_{\Ncal*}[\rho]}\leq\log C$ for any density matrix $\rho\in\Dcal(\Hcal)$.
\item[(ii)]Conversely, assume that there exists a $C>0$ such that $\Dent{\rho}{E_{\Ncal*}[\rho]}\leq\log C$ for all density matrix $\rho\in\Dcal(\Hcal)$. Then for any $p\geq2$
\[\norm{\id}_{(2,2)\to(2,p),\,\Ncal}\leq (|I|\,C)^{\frac12-\frac1p}\,\]
where $|I|$ is the number of blocks in the decomposition \ref{eqtheostructlind2} of $\Ncal$.
\end{enumerate}
\end{proposition}

\begin{remark}
In the proposition we ask that $\norm{\id}_{(2,2)\to(2,p),\,\Ncal}\leq C^{\frac12-\frac1p}$ for all $p\geq2$. This is actually not needed, as by the Riesz-Thorin interpolation Theorem this is equivalent to 
\[\norm{\id}_{(1,1)\to(1,\infty),\,\Ncal}\leq C\,.\]
We see here that it is central that the norms we use form an interpolating family of norms.
\end{remark}

The last proposition is not optimal, which indicates that point (ii) in \Cref{gross1} may also not be, even for a non-trivial evolution. To see this, consider the situation where $\Ncal$ is the algebra of diagonal operator in some orthonormal basis. In this case $|I|$ is equal to the dimension $d_\Hcal$ of the Hilbert space $\Hcal$ (the converse is also true: if $|I|=d_\Hcal$ then $\Ncal$ is commutative maximal). In this case $\Dent{\rho}{E_{\Ncal*}[\rho]}\le\log d_\Hcal$, with equality for the maximally coherent state $\Omega$:
\[\Omega=\frac{1}{d_\Hcal}\,\sum_{i,j=1}^{d_\Hcal}\,\outerp{e_i}{e_j}\,,\]
where $(e_i)_{i=1,...,d_\Hcal}$ is the orthonormal basis in which the operators in $\Ncal$ are diagonal. It means that \Cref{eq_normestimate2} saturates and that the equivalence is tight in \Cref{prop_norm_estimate}.\\

So far we only focus on the norm $\norm{\id}_{(2,2)\to(2,p),\,\Ncal}$ for different value of $p$. In \Cref{sec6}, however, we need the other kind of estimate, i.e. when the first parameter varies. In this case we can prove the following.

\begin{proposition}\label{prop_norm_estim2}
For all $1\leq p\leq q$, we have
\begin{equation}\label{eq_prop_norm_estim21}
  \norm{\id}_{(p,q)\to(q,q)\,,\,\Ncal}=\left(\max_{i\in I} d_{\Hcal_i}\right)^{\frac1p-\frac1q}\,,
\end{equation}
where the $d_{\Hcal_i}$ are the dimensions of the spaces $\Hcal_i$ occuring in the decomposition of $\DF$ given by \eqref{eqtheostructlind2}. For $p=2$ and $q=\infty$, this yields 
\begin{equation}\label{eq_prop_norm_estim22}
\norm{\id}_{(2,\infty)\to(\infty,\infty)\,,\,\Ncal}\leq \max_{i\in I}\sqrt{d_{\Hcal_i}}\,.
\end{equation}
\end{proposition}

\begin{proof}
Because of the two following trivial norm estimates
\begin{align*}
 &\norm{\id}_{\infty\to\infty}\leq 1\,,\\
 &\norm{\id}_{1\to1\,,\,\sigma_\tr}\leq 1
\end{align*}
and by applying twice the Riesz-Thorin interpolation Theorem (one for the first parameter and then one for the second), it is enough to prove
\[\norm{\id}_{(1,\infty)\to(\infty,\infty)\,,\,\Ncal}\leq \max_{i\in I} d_{\Hcal_i}\,.\]
But by duality, this is the same as
\[\norm{\id}_{(1,1)\to(\infty,1)\,,\,\Ncal}\leq \max_{i\in I} d_{\Hcal_i}\\,.\]
Let $X\in\Bcal(\Hcal)$ be positive semi-definite nd fix $\eps>0$. Then there exists a positive definite $A\in\DF$ with $\norm{A}_{1,\sigma_\tr}=1$ such that:
\begin{align*}
\norm{X}_{(\infty,1)\,,\,\Ncal} 
& \leq \norm{A^{\frac12}\,X\,A^{\frac12}}_{1,\,\sigma_\tr} + \eps \\
& = \tr\left[A\sigma_\tr\,X\right] + \eps  \\
& \leq \norm{A}_\infty \norm{X}_{1,\,\sigma_\tr}+\eps\,,
\end{align*}
where in the last line we use Hölder's inequality. Then we have
\begin{align*}
 \norm{A}_\infty
 & = \sum_{i\in I} \norm{A_i}_\infty   \\
 & \leq \sum_{i\in I}\,d_{\Hcal_i}\,\norm{A_i}_{1,\,\frac{\II_{\Hcal_i}}{d_{\Hcal_i}}}  \\
 & \leq \max_{i\in I} d_{\Hcal_i}\,,
\end{align*}
where in the last line we use that $\norm{A}_{1,\,\sigma_\tr}=1$. This concludes the proof.

\end{proof}

\end{document}